\documentclass[reqno]{amsart}

\usepackage{amsmath}
\usepackage{amsfonts}
\usepackage{cases}
\usepackage{graphics}
\usepackage{epsfig}
\usepackage{amssymb}
\usepackage{amsthm}
\usepackage{amscd}
\usepackage[all]{xy}
\usepackage{latexsym}
\usepackage{xcolor}
\usepackage{graphicx}
\usepackage{multirow}
\usepackage{geometry}
\usepackage{color}
\usepackage{CJKutf8} 
\usepackage{bm}
\usepackage{mathrsfs}
\usepackage{mathtools}
\usepackage{hyperref}

\newtheorem{thm}{Theorem}[section]

\newtheorem{conj}[thm]{Conjecture}

\theoremstyle{definition}
\newtheorem{Def}[thm]{Definition}
\newtheorem{rem}[thm]{Remark}

\numberwithin{equation}{section}
\numberwithin{figure}{section}
\numberwithin{table}{section}

\def\rchi{{\hbox{\raise1.5pt\hbox{$\chi$}}}}

\def\I{\ifmmode\mathrm{I}\else I\fi}
\def\II{\ifmmode\mathrm{I\hspace{-1.2pt}I}\else I\hspace{-1.2pt}I\fi}
\def\III{\ifmmode\mathrm{I\hspace{-1.2pt}I\hspace{-1.2pt}I}\else I\hspace{-1.2pt}I\hspace{-1.2pt}I\fi}
\def\IV{\ifmmode\mathrm{I\hspace{-1.2pt}V}\else I\hspace{-1.2pt}V\fi}
\def\V{\ifmmode\mathrm{V}\else V\fi}
\def\VI{\ifmmode\mathrm{V\hspace{-1.2pt}I}\else V\hspace{-1.2pt}I\fi}

\newcommand{\bea}{\begin{eqnarray}}
\newcommand{\eea}{\end{eqnarray}}
\newcommand{\be}{\begin{equation}}
\newcommand{\ee}{\end{equation}}

\DeclareMathOperator*{\Res}{Res}

\begin{document}
\large
\setcounter{section}{0}

\title[Accessory Parameter, Voros Periods and classical conformal blocks]
{Accessory Parameter of Confluent Heun Equations, \\ 
Voros Periods and classical irregular conformal blocks}

\author[K.\ Iwaki]{Kohei Iwaki}
\address{Graduate School of Mathematical Sciences,  The University of Tokyo, Japan} 
\email{k-iwaki@g.ecc.u-tokyo.ac.jp}

\author[H.\ Nagoya]{Hajime Nagoya}
\address{School of Mathematics and Physics, Kanazawa University} 
\email{nagoya@se.kanazawa-u.ac.jp}

\author[A.\ Shukuta]{Ayato Shukuta}
\address{Graduate School of Mathematical Sciences,  The University of Tokyo, Japan} 
\email{shukuta-ayato306@g.ecc.u-tokyo.ac.jp}


\begin{abstract}
For the Heun differential equation and all of its confluent equations, we derive formal series expansions of the accessory parameters using the Voros periods. 
We then compare these expansions with the classical conformal blocks recently obtained by Bonelli–Shchechkin–Tanzini, and examine the Zamolodchikov-type conjecture expected to hold between them, allowing for irregular singularities. 
In particular, as an extension of the previous works of Mironov--Morozov, Piatek--Pietrykowski and Lisovyy--Naidiuk, we provide a detailed prescription for choosing cycles on the spectral curve that yield the Voros period which corresponds to the classical (regular or irregular) conformal blocks through the accessory parameter.
\end{abstract}




\maketitle

\allowdisplaybreaks

\tableofcontents

\setlength{\parskip}{0.1ex}


\section{Introduction}
\label{sec:intro}

{\footnotesize

\begin{table}[t]
  \begin{tabular}{c|c|c} \hline
     $J$ & Name of equation & $Q_{\rm J}$ \\ \hline \hline
\parbox[c][5.5em][c]{0em}{}
{VI}
&
Heun
& ~~~\quad
\begin{minipage}{.42\textwidth}
\begin{center}
${\displaystyle \frac{\theta_0^2 - \frac{\hbar^2}{4}}{x^2} 
+ \frac{\theta_1^2 - \frac{\hbar^2}{4}}{(x-1)^2} + \frac{\theta_t^2 - \frac{\hbar^2}{4}}{(x-t)^2}}$  \\[+.5em]
${\displaystyle 
+ \frac{\theta_\infty^2 - \theta_0^2 - \theta_1^2 - \theta_t^2 + \frac{\hbar^2}{2}}{x(x-1)} 
- \frac{{\mathscr E}}{x(x-1)(x-t)} }$ \\[-.5em]
~~
\end{center}
\end{minipage}
\\\hline
\parbox[c][3.3em][c]{0em}{}
{V}
&
confluent Heun
& ~~~\quad
\begin{minipage}{.42\textwidth}
\begin{center}
$\displaystyle \frac{\theta_0^2 - \frac{\hbar^2}{4}}{x^2}
+ \frac{\theta_t^2 - \frac{\hbar^2}{4}}{(x-t)^2} 
+ \frac{1}{4} + \frac{\theta_\infty}{x} - \frac{{\mathscr E}}{x(x-t)}$ 
\end{center}
\end{minipage}
\\\hline
\parbox[c][3.3em][c]{0em}{}
{IV}
&
biconfluent Heun
& ~~~\quad
\begin{minipage}{.42\textwidth}
\begin{center}
$\displaystyle \frac{\theta_0^2 - \frac{\hbar^2}{4}}{x^2} + \frac{{\mathscr E}}{x} + 2 \theta_\infty + (x+t)^2$
\end{center}
\end{minipage}
\\\hline
\parbox[c][3.3em][c]{0em}{}
${\rm III}_1$
&
doubly confluent Heun
& ~~~\quad
\begin{minipage}{.42\textwidth}
\begin{center}
$\displaystyle \frac{t^2}{4x^4}  + \frac{t \theta_0}{x^3} - \frac{{\mathscr E}}{x^2} + \frac{\theta_\infty}{x} + \frac{1}{4}$ 
\end{center}
\end{minipage}
\\\hline

\parbox[c][3.3em][c]{0em}{}
${\rm III}_2$
&
reduced doubly confluent Heun
& ~~~\quad
\begin{minipage}{.42\textwidth}
\begin{center}
$\displaystyle \frac{t}{x^3} - \frac{{\mathscr E}}{x^2} + \frac{\theta_\infty}{x} + \frac{1}{4}$ 
\end{center}
\end{minipage}
\\\hline
\parbox[c][3.3em][c]{0em}{}
${\rm III}_3$
&
doubly reduced doubly confluent Heun
& ~~~\quad
\begin{minipage}{.42\textwidth}
\begin{center}
$\displaystyle \frac{t}{x^3} - \frac{{\mathscr E}}{x^2} + \frac{1}{x}$ 
\end{center}
\end{minipage}
\\\hline

\parbox[c][3.3em][c]{0em}{}
{II}
&
triconfluent Heun
& ~~~\quad
\begin{minipage}{.42\textwidth}
\begin{center}
$\displaystyle (x^2+t)^2 + 2 \theta_\infty x + {\mathscr E}$ 
\end{center}
\end{minipage}
\\\hline

\parbox[c][3.3em][c]{0em}{}
${\rm II}'$
&
reduced biconfluent Heun
& ~~~\quad
\begin{minipage}{.42\textwidth}
\begin{center}
$\displaystyle \frac{\theta_0^2 - \frac{\hbar^2}{4}}{x^2}
+ \frac{{\mathscr E}}{x} + t + x$ 
\end{center}
\end{minipage}
\\\hline

\parbox[c][3.3em][c]{0em}{}
{I}
&
reduced triconfluent Heun
& ~~~\quad
\begin{minipage}{.42\textwidth}
\begin{center}
$\displaystyle 4x^3 + 2 t x + {\mathscr E}$ 
\end{center}
\end{minipage}
\\\hline

  \end{tabular}

\bigskip
  \caption{List of potentials of the Heun equation $H_{\rm J}$ from \cite{LN21}. 
  See Appendix \ref{section:introduce-hbar} for the details of introduction of $\hbar$. 
   }
  \label{table:heun}
\end{table}

}

The Heun equation is the canonical form of a Fuchsian ordinary differential equation 
with four regular singular points on the Riemann sphere. 
Its explicit expression (in the Schr\"odinger-form) is given as 
\begin{equation}
H_{\rm VI} ~:~ 
\left( \hbar^2 \frac{d^2}{dx^2} - Q_{\rm VI}(x) \right) \psi = 0
\end{equation}
\begin{equation}
Q_{\rm VI} = \frac{\theta_0^2 - \frac{\hbar^2}{4}}{x^2} 
+ \frac{\theta_1^2 - \frac{\hbar^2}{4}}{(x-1)^2} + \frac{\theta_t^2 - \frac{\hbar^2}{4}}{(x-t)^2}
+ \frac{\theta_\infty^2 - \theta_0^2 - \theta_1^2 - \theta_t^2 + \frac{\hbar^2}{2}}{x(x-1)} 
- \frac{{\mathscr E}}{x(x-1)(x-t)}.
\end{equation}
Here, $\theta_s$ represents the exponent of the local monodromy at each regular singular point 
$s = 0,1, t$ and $\infty$, whereas ${\mathscr E}$ is the so-called accessory parameter, 
which is not determined solely by the local monodromy data. 
Table \ref{table:heun} shows that the (Schr\"odinger-form of) confluent Heun equations $H_{\rm J}$
which are obtained by confluence of singular points from the equation $H_{\rm VI}$ 
and possess irregular singular points, also admit an accessory parameter ${\mathscr E}$ (cf.\, \cite{SL2000}).
The parameter $\hbar$ is a formal small parameter introduced 
for the purpose of analyzing the equation through exact WKB analysis \cite{KT98}.

Because of the presence of the accessory parameter ${\mathscr E}$, 
the monodromy/Stokes matrices of these confluent families of Heun equations 
are far more difficult to analyze than that of the confluent family 
obtained from Gauss hypergeometric differential equation. 
Nevertheless, the monodromy and connection problems associated with the Heun equation 
continue to be important topics of research, as seen, for example, 
in recent applications to black hole analysis; see \cite{NdC14, NMLC18, AGH20, Hatsuda20} for example.
In addition, \cite{Wak16} also explores intriguing connections 
between the Heun equation and number theory through 
the non-commutative harmonic oscillator and the Rabi model.

Through the Zamolodchikov conjecture proposed in \cite{Zam86}, 
which will be discussed below, the Heun equations are also deeply related to conformal field theory. 
As illustrated, for example, by the celebrated discovery of the Kyiv formula 
for the Painlev\'e VI $\tau$-function obtained in \cite{GIL12}, 
the exploration of such connections between 
conformal field theory and the theories of ordinary differential equations 
and special functions has become one of the important topics in recent studies. 
In this paper, as part of this line of research, 
we investigate the relationship among the Voros periods in exact WKB analysis, 
accessory parameters, and classical conformal blocks
for all confluent Heun equations listed in Table \ref{table:heun}. 
Our results are based on the recent results of Bonelli--Shchechkin--Tanzini \cite{BST25}, 
and regarded as an extension of the works of Mironov--Morozov \cite{MM09}, 
Piatek--Pietrykowski \cite{PP14, PP16}, and the recent results of Lisovyy--Naidiuk \cite{LN21}. 
In the course of reviewing these preceding works, 
we explain the above key concepts.

\subsection{Voros period and accessory parameter}

It is well known that the exact WKB method \cite{Vor83, KT98} is extremely effective 
in the analysis of monodromy and Stokes matrices for second-order Schr\"odinger-type equations, such as the confluent Heun equations 
\begin{equation}
H_{\rm J} ~:~ 
\left( \hbar^2 \frac{d^2}{dx^2} - Q_{\rm J}(x) \right) \psi = 0, 
\end{equation}
whose potential functions are listed in Table \ref{table:heun}. 
In the exact WKB analysis, the spectral curve 
\begin{equation}
y^2 = Q_0(x) = \lim_{\hbar \to 0} Q_{\rm J}(x).
\end{equation}
and the Voros periods 
\begin{equation}
V_\gamma = \sum_{m \ge -1} \hbar^{m} \oint_\gamma S_m(x) \,dx, 
\end{equation}
play important roles \cite{DDP93, IN14, Iwaki-Les-Houches}. 
The latter is the generating function of the period integrals of $\{ S_m(x) \}_{m\ge -1}$, which give the coefficients of WKB formal solution, along a certain cycle $\gamma$ on the spectral curve; see Definition \ref{def:Voros-period}.
It was shown in \cite{SAKT91} and \cite[\S 3]{KT98} that the Borel sum (with respect to $\hbar$) of Voros periods describes the nontrivial entries of the monodromy matrices. 
This provides a clear explanation for the marked difference between the Gauss hypergeometric and Heun equations: 
The spectral curve of the Gauss hypergeometric equation has genus 0; therefore, its Voros periods can be described solely in terms of the residues at the regular singular points (i.e., the local monodromy exponents).
On the other hand, the spectral curve of the confluent Heun equations has genus $1$ and therefore admits genuinely nontrivial Voros periods.

As suggested by the celebrated conjecture of Zamolodchikov \cite{Zam86}, 
the Heun equation is also deeply related to conformal field theory. 
More precisely, the conjecture asserts 
\begin{itemize}
\item 
Conjecture A :  Conformal block admits a large central charge limit;
the limit is called the classical conformal block. 

\smallskip 
\item 
Conjecture B :  The classical conformal blocks in Conjecture A 
give the accessory parameters of the Heun equation $H_{{\rm VI}}$.
\end{itemize}
We also note that, through the Alday--Gaiotto--Tachikawa (AGT) correspondence \cite{AGT}, this large central charge limit has been identified with the Nekrasov--Shatashvili (NS) limit, in which the classical conformal block is expected to agree with the NS twisted superpotential \cite{NS09}. 

Through the accessory parameter, it is therefore natural to ask how the Voros periods are related to the classical conformal block. 
Their relation has already been studied by Mironov--Morozov \cite{MM09} and Piatek--Pietrykowski \cite{PP14} for the Mathieu equation, which is equivalent to the confluent Heun equation  $H_{{\rm III}_3}$.
This result was further developed in the relatively recent work of Lisovyy--Naidiuk in \cite{LN21}, which we now briefly review.


The original conjecture of Zamolodchikov concerns the classical conformal block with regular singularities.
Building on the study of irregular conformal blocks introduced by Gaiotto and Teschner 
\cite{Gaiotto09, GT12}, Mironov--Morozov \cite{MM09} and Piatek--Pietrykowski \cite{PP14, PP16} explicitly computed the large central charge limit of a certain irregular conformal block and pointed out its connection to the eigenvalues of the Mathieu equation. 
Their method relies on the analysis of the Voros period (all-order Bohr--Sommerfeld period) for a particular cycle $\gamma$. 
See also \cite{GGM19, GHN21} for further studies.
We also note that an analysis of the large central charge limit was also carried out in \cite{RZ15a, RZ15b}.


The aforementioned results studied the conformal blocks ``with regular singular expansion", which means that their singularities are at most of complex-power type. 
On the other hand, in \cite{LN21}, Lisovyy and Naidiuk proposed that a similar Zamolodchikov-type conjecture may also hold for irregular conformal blocks ``with irregular singular expansion''; that is, those admitting essential singularity written as exponential factors. 
The notion was rigorously defined by the second-named author in \cite{Nagoya15, Nagoya18}.
More specifically, among the confluent Heun equations in Table \ref{table:heun},  \cite{LN21} provided strong supporting evidence for the validity of the Conjectures A and B in the cases of $H_{\rm{V}}$ and $H_{\rm{IV}}$, 
for which the corresponding irregular singular expansion of the conformal blocks had already been introduced at the time of their work (\cite[Conjecture 3.1 and 3.2]{LN21}). 
In their analysis, similar to \cite{MM09, PP14}, the computation of the accessory parameter relied crucially on the study of Voros periods.
In particular, \cite{LN21} succeeded in characterising which cycle should be chosen when considering the Voros period in realizing the given irregular conformal blocks as the accessory parameter in the above concrete examples.  

One of the main objectives of this paper is, building on this strategy of Lisovyy--Naidiuk, to extend the Zamolodchikov-type conjecture for all confluent Heun equations in the presence of irregular singularities, based on the results of \cite{BST25} by Bonelli--Shchechkin--Tanzini described next.

\subsection{Classical conformal blocks from bilinear equations}

As in the case of the Heun differential equation, the Painlev\'e equations, together with their associated isomonodromic linear differential equations, are also closely related to conformal field theory. 
The celebrated Kyiv formula asserts that the $\tau$-function of the Painlev\'e VI equation can be constructed 
as a Zak (discrete Fourier) transform of 
the regular-singular Virasoro four-point conformal block. 
Through the AGT correspondence, the aforementioned conformal block 
is identified with the so-called Nekrasov partition function, 
which in turn yields remarkably explicit combinatorial formulas for the expansion coefficients 
of the Painlev\'e VI $\tau$-function. 
Since its discovery, various extensions of the Kyiv formula for 
confluent Painlev\'e equations based on 
irregular conformal blocks and topological recursion have also been investigated; 
see \cite{GIL13, Nagoya15, BLMST16, Nagoya18, Iwa19, PP23, IILZ25} for example.

In the above relation with Painlev\'e equations, however, the central charge $c$ 
of the conformal block is restricted to $1$. 
On the other hand, when $c$ takes a generic value, it was found in \cite{BMS17, GMS20} that a certain non-commutative quantum $\tau$-function satisfies an equation that naturally extends the Hirota-type bilinear identity in the $ c = 1$ case.
This approach has recently been generalized in the work \cite{BST25} of Bonelli--Shchechkin--Tanzini, and they proposed Hirota-type bilinear equations for all the Painlev\'e equations beyond central charge $1$. Moreover, by analyzing this bilinear equation satisfied by quantum $\tau$-function, \cite{BST25} introduced an object ${\mathscr Z}(\varepsilon_1,\varepsilon_2)$, with the so-called $\Omega$-background parameters $\varepsilon_1, \varepsilon_2$, as Fourier components of the quantum $\tau$-function. 
${\mathscr Z}$ is an analogue of the Nekrasov partition function.
The authors further develop the theory of quantum Painlev\'e equations in their subsequent paper \cite{BST25-2}, where they also make comparisons with the quantum Painlev\'e equations studied by the second-named author \cite{Nag04}.

Assuming that the AGT correspondence extends to the irregular setting, although it has not yet been fully formulated, it is naturally suggested that ${\mathscr Z}(\varepsilon_1,\varepsilon_2)$, obtained in \cite{BST25}, corresponds to the conformal block with the central charge 
$c = 1 + {6(\varepsilon_1+\varepsilon_2)^2}/{(\varepsilon_1 \varepsilon_2)}$, 
while the classical conformal block (the NS twisted superpotential) is given by 
\begin{equation} \label{eq:CCB-intro}
{\mathscr W} = \lim_{(\varepsilon_1,\varepsilon_2) \to (\hbar,0)} 
\left( - \varepsilon_1 \varepsilon_2 \log {\mathscr Z}(\varepsilon_1,\varepsilon_2)  \right). 
\end{equation}
As observed in \cite{GMS20, BST25}, some of ${\mathscr Z}(\varepsilon_1,\varepsilon_2)$  
are expected to reproduce the known irregular conformal blocks 
associated with the Virasoro algebra. 
However, for several cases 
-- including the one treated in \S \ref{sec:H-III-3-expansion-example-2} -- 
it is not known how they should be constructed directly from the Virasoro algebra. 
For this reason, it should be noted that these objects are, 
strictly speaking, better regarded as ``candidates for classical conformal blocks". 
We note that, except for a limited number of cases \cite{DGP24}, 
the existence of the classical conformal block \eqref{eq:CCB-intro} 
has not been established mathematically. 
Nevertheless, the computations of \cite{GMS20, BST25} strongly suggest that 
this limit should continue to exist even in the irregular setting.

\subsection{Main conjecture}
Our main observation provides compelling evidence for a Zamolodchikov-type conjecture expected to relate the classical conformal blocks ${\mathscr W}$ suggested by Bonelli--Shchechkin--Tanzini \cite{BST25} to the accessory parameters ${\mathscr E}$ of all confluent Heun equations in Table \ref{table:heun}, including the examples treated in aforementioned earlier works.

For the confluent Heun equations $H_{\rm J}$, 
we impose the condition for what we call ``partial isomonodromic deformation";  
namely, a deformation of $H_{\rm J}$ for which the Voros period $V_\gamma$ associated with 
a certain cycle $\gamma$ remains independent of $t$: 
\begin{equation} \label{eq:PIMD-introduction}
V_\gamma = \frac{2 \pi i \nu}{\hbar} - d \pi i.
\end{equation}
Here, 
$\nu$ is a complex parameter, independent of $t$, 
that parametrizes the Voros period, and $d$ is an integer. 
Under this condition, we show that the accessory parameter ${\mathscr E}$ 
is determined as an $\hbar$-formal series valued implicit function of the form 
\begin{equation}
{\mathscr E} = t^{d_{{\mathscr E}}} 
\sum_{k \ge 0} \hbar^{2k} \, {\mathscr G}_k(\Lambda).
\end{equation}
Here, $d_{\mathscr E} \in {\mathbb Q}$ and ${\mathscr G}_k(\Lambda) = \sum_{\ell \ge 0} \Lambda^{\ell} {\mathscr G}_{k}^{[\ell]}$ are holomorphic functions at $\Lambda = 0$, where $\Lambda$ is given by a fractional power of $t$.
We also show that the expansion coefficients ${\mathscr G}_{k}^{[\ell]}$ can be computed explicitly to arbitrary order. 
Although these results are described as Theorems \ref{thm:unique-existence} and \ref{thm:unique-existence-2} for a particular example $H_{{\rm III}_3}$, the same argument applies to all confluent Heun equations $H_{\rm J}$. 

For each Heun equation $H_{\rm J}$, we will see that different choices of the cycle $\gamma$ on the spectral curve yield different formal expansions of the accessory parameter. 
For example, in the $H_{{\rm III}_3}$ case, the choice of cycle determines whether one obtains the expansion of the accessory parameter in the limit $t \to \infty$ (\S \ref{sec:H-III-3-expansion-example-2}) or in the limit $t \to 0$ (\S \ref{subsec:example-AP-III-regular}). 
Among these, the expansion as $t \to 0$ is parallel to the previous works \cite{MM09, PP14}, whereas a different expansion is obtained as $t \to \infty$. 
We expect that the expansion we obtained for $t \to \infty$ coincides with that derived by Gavrylenko--Marshakov--Stoyan  \cite{GMS20}.

Although such an analysis of the accessory parameter via the WKB analysis and Voros periods was also carried out in the works \cite{MM09, PP14, GMS20}, the computational method for ${\mathscr G}_{k}^{[\ell]}$ is technically different (see Remark \ref{rem:GMS20}). 
In particular, we choose $\gamma$ to be the vanishing cycle in each  limit, and the elliptic integral giving the Voros period can be evaluated by a perturbative residue computation, as parallel to \cite{LN21}. 
We present the strategy for such a concrete choice of cycles in \S \ref{subsec:AP-algorithm}, and carry out the detailed analysis for each Heun equation in \S \ref{section:examples}.
We note that the essential idea for this method was already proposed in \cite{LN21}, and that our approach should be regarded as a slight modification of theirs\footnote{
Our parameter $\hbar$ plays a different role from that in the work of Lisovyy--Naidiuk \cite{LN21}. 
The parameter $\hbar$ in \cite{LN21} is closely related to our expansion parameter $\Lambda$. 
}.  

Based on these computational results, 
as an extension of the aforementioned previous works, 
we provide a list of the expected relations 
between the accessory parameters ${\mathscr E}$ obtained in this way 
and the corresponding classical (regular or irregular) conformal blocks 
${\mathscr W}$ of \cite{BST25}. 
Our main conjecture can be roughly stated as follows:

\begin{conj}
For all quantum expansions of ${\mathscr Z}$ obtained in \cite{BST25}, the classical limit  given in \eqref{eq:CCB-intro} exists. 
Moreover, for each ${\mathscr W}$ thus obtained, there exists a confluent Heun equation $H_{\rm J}$ in Table \ref{table:heun} and a cycle $\gamma$ on its spectral curve such that the accessory parameter ${\mathscr E}$ determined implicitly by the condition \eqref{eq:PIMD-introduction} essentially agrees with  $\partial_t {\mathscr W}$.
\end{conj}

We describe the explicit correspondence for each Heun equation in \S \ref{section:examples}. 
For every example, we specify the precise choice of the cycle $\gamma$ and verify that the first several coefficients in the expansion of the corresponding classical conformal block ${\mathscr W}$ obtained in \cite{BST25} indeed agree with those of the accessory parameter ${\mathscr E}$ obtained by the above method.
We hope that the list in \S \ref{section:examples} will be useful for various future applications.

We note that the expansion coefficients of the accessory parameter constructed here form a formal power series in $\hbar$, and that a rigorous mathematical analysis of their convergence or Borel summability remains an open problem for future investigation.
As the works in this direction for the NS twisted superpotential, we refer to \cite{GGM19, GHN21, GM23} for example.
Also, a comparison between the classical conformal block and the free energy obtained from ``$Q$-top recursion'' by Osuga \cite{Osuga23} is an interesting problem.

\subsection{Outline of the paper}

This paper is organized as follows.
In \S \ref{section:WKB-and-VP}, we review the definition of the Voros period. 
In the first half of \S \ref{section:main-conjecture}, we describe the method, which is motivated by the previous works \cite{MM09, PP14, LN21}, for constructing the accessory parameter as a formal power series defined implicitly from the Voros period. 
We take $H_{{\rm III}_3}$ as an example and derive the series expansions of the accessory parameter in the limits $t \to \infty$ and $t \to 0$, and we show the well-posedness of our algorithm. 
In the second half of \S \ref{section:main-conjecture}, we briefly review the results of \cite{BST25} on the construction of the (candidate of) conformal blocks from the bilinear equation for the quantum Painlev\'e $\tau$-function, and formulate the conjecture that the classical conformal blocks obtained there provide the series expression of the accessory parameter derived above.
In \S \ref{section:examples}, we derive various series expansions of the accessory parameter ${\mathscr E}$ for all $H_{\rm J}$, and test Conjectures A and B case by case; 
in particular, we will check that the series expansions of the accessory parameter ${\mathscr E}$ and that of the classical conformal block ${\mathscr W}$ determined from \cite{BST25} exactly agree in the lower-order terms.

\bigskip
\subsection*{Acknowledgements}
\vspace{-.5em}

We would like to thank 
Shogo Aoki, 
Nikorai Iorgov, 
Oleg Lisovyy, 
Yasuyuki Hatsuda, 
Kento Osuga, 
Yoshitsugu Takei, 
Kouichi Takemura 
and 
Yurii Zhuravlov
for fruitful and insightful discussions. 
The work is supported by 
JSPS KAKENHI Grand Numbers 
21H04994, 22H00094, 22K03350, 23K17654, 24K00525.
This work is also supported by the Research Institute for Mathematical Sciences, an International Joint Usage/Research Center located in Kyoto University,
and by FoPM, WINGS Program, the University of Tokyo.



\section{Review of WKB solution and Voros periods}
\label{section:WKB-and-VP}

Here we briefly review the construction of the WKB (formal) solution of the Heun equations, 
following \cite[\S 2]{KT98}.
The construction is commonly applicable to all the equations $H_{\rm J}$ listed in Table \ref{table:heun}. 
For later purposes, we assume that the accessory parameter ${\mathscr E}$
has a series expansion of the form 
\begin{equation}
{\mathscr E} = \sum_{m \ge 0} \hbar^{m} {\mathscr E}_{m}
\end{equation}
whose coefficients ${\mathscr E}_{m} = {\mathscr E}_{m}(t)$ 
are certain functions of $t$ which will be specified later. 
Accordingly, the Schr\"odinger potential $Q_{\rm J}$ in Table \ref{table:heun} 
admits a series expansion 
\begin{equation}
Q_{\rm J}(x,\hbar)   = \sum_{m \ge 0} \hbar^{m} Q_m(x)
\end{equation}
where $Q_m$ are rational functions of $x$.  
Although $Q_m$ also depends on the parameters $t$ and $\theta_s$, 
we omit the dependency to simplify notation.
Our parameters $\theta_s$ are labeled by a singular point $s$ 
of the equation (i.e., a pole of $Q_{\rm J}$) and chosen so that 
\begin{equation} 
\Res_{x=s} \sqrt{Q_0(x)} dx = \pm \theta_s
\end{equation}
holds up to sign.

A WKB solution of $H_{\rm J}$ is defined as a formal solution 
of the form 
\begin{equation} \label{eq:wkbsol}
\psi(x,\hbar) = \exp\left( \int S(x,\hbar) dx \right)
\end{equation} 
where $S$ is a formal (Laurent) series of $\hbar$ of the form
\begin{equation}
S(x,\hbar) = \sum_{m \ge -1} \hbar^{m} S_m(x). 
\end{equation}
It is easy to see that $S$ must satisfy the Riccati equation 
\begin{equation} \label{eq:Riccati}
\hbar^2 \left( S^2 + \frac{dS}{dx} \right) = Q_{\rm J}.
\end{equation}
Therefore, the leading term $S_{-1}(x)$ satisfies the algebraic equation
\begin{equation} \label{eq:cl}
S_{-1}^2 = Q_0,
\end{equation}
and the higher order terms are recursively determined by the recursion relation
\begin{equation} \label{eq:wkb-rec}
2 S_{-1} S_{m+1} + \sum_{\substack{m_1, m_2 \ge 0 \\ m_1 + m_2 = m}} S_{m_1} S_{m_2}
+ \frac{dS_m}{dx} = 
Q_{m+2}
\end{equation}
for $m \ge -1$.
For each unramified singular point $s$ of $H_{\rm J}$, we have 
\begin{equation} \label{eq:vp-pole}
\Res_{x=s} S_{m}(x) dx = 
\begin{cases}
\pm \theta_s & \text{if $m = -1$} \\
{p_s}/{4} & \text{if $m = 0$} \\
0 & \text{if $m \ge 1$},
\end{cases}
\end{equation}
where $p_s$ is the pole order of $Q_0$ at $x=s$. 
In general, $S_m$ may have simple poles at the second-order poles of $Q_0$,  but in our case (due to a suitable choice of $Q_2$), we note that $S_m$ with $m \ge 1$ are holomorphic at these points. 

Using the series $S$ thus constructed, 
the WKB solution is determined by \eqref{eq:wkbsol}. 
However, in this paper, we will focus not on the WKB solution itself 
but rather on the Voros period. 
To explain its definition, we first describe the geometric properties 
of the Riemann surface 
\begin{equation} \label{eq:cl-xy}
{\mathcal C} = \{(x,y) \in {\mathbb C}^2 ~|~ y^2 = Q_0(x) \}
\end{equation}
defined by \eqref{eq:cl}.
The Riemann surface ${\mathcal C}$ is called the classical limit of $H_{\rm J}$. 
The natural projection 
\begin{equation}
\pi: {\mathcal C} \ni (x,y) \mapsto x \in {\mathbb C}
\end{equation}
gives a 
double covering of the complex $x$-plane.
It follows from the recursion relation \eqref{eq:wkb-rec}
that $S_m(x)$ are holomorphic functions defined on 
\begin{equation}
{\mathcal C}' = {\mathcal C} \setminus \pi^{-1}(T),\quad
T = \{ v \in {\mathbb C} ~|~ Q_0(v) = 0 \}
\end{equation}
(or meromorphic functions defined on its compactification $\overline{{\mathcal C}}$). 
A point in $T$ is called a turning point of $H_{\rm J}$ in the theory of (exact) WKB analysis. 

For a generic choice of the parameters $t$, $\theta_s$ and ${\mathscr E}_0$,  
the Riemann surface ${\mathcal C}$ has genus $1$. 
Although we refrain from explicitly stating the conditions for the genericity of the parameters, 
we will assume throughout the following that ${\mathcal C}$ has genus one. 

\begin{Def} \label{def:Voros-period}
For a closed cycle $\gamma$ on ${\mathcal C}'$, 
the Voros period of the Heun equation $H_{\rm J}$ 
associated with $\gamma$ is defined as 
a formal series obtained by integrating $S$ term by term as follows:
\begin{equation} \label{eq:vp}
V_\gamma = \sum_{m \ge -1} \hbar^{m} V_m, \quad
V_m =\oint_{\gamma} S_m(x) dx.
\end{equation}
\end{Def}

The Voros period $V_\gamma$ is a formal series of $\hbar$ 
whose coefficients are (possibly non-trivial) elliptic integrals. 
As is shown in \cite{SAKT91} and \cite[Section 3]{KT98}, 
the monodromy (and Stokes) matrices 
of Schr\"odinger-type equations such as $H_{\rm J}$
are described in terms of the Borel sum of the Voros periods. 
It follows from \eqref{eq:vp-pole} that the Voros periods 
associated with a residue cycle around a singular point $s$ of $H_{\rm J}$
are essentially given by the local monodromy exponent $\theta_s$, 
and the non-trivial part of the monodromy/Stokes data of $H_{\rm J}$
is given by the Voros periods associated with $A$-cycle and $B$-cycle on ${\mathcal C}$
(cf.\, \cite[Section 3]{KT98}).


\section{Main conjecture on accessory parameters and classical conformal blocks}
\label{section:main-conjecture}


In this section, as an extension of the analysis by 
Lisovyy--Naidiuk \cite{LN21}, 
we propose the conjecture that, 
by specifying the value of the Voros period 
associated with an appropriate cycle, 
the accessory parameter is determined 
as a formal power series-valued implicit function, 
and that it gives the classical limit of 
the (regular or irregular) conformal block. 
This may be regarded as an irregular-singular version of 
the Zamolodchikov conjecture \cite{Zam86}. 
We illustrate the claim here by taking $H_{{\rm III}_3}$ as an example, 
and in the next section we extend the analysis to all other Heun equations.

\subsection{Accessory parameter from Voros period}

\subsubsection{Computational algorithm of the accessory parameter}
\label{subsec:AP-algorithm}

As stated in Section \ref{sec:intro}, our goal is to derive the series expansion of 
the accessory parameter of $H_{\rm J}$ and compare it with the classical conformal blocks of various type. 
To this end, our approach is to study a 
``WKB-theoretic partial isomonodromic deformation" of $H_{\rm J}$, 
which we now explain. 

The condition we will impose on the associated Voros period is formulated as 
\begin{equation} \label{eq:partial-IMD}
\exp(V_{\gamma}) = \pm \exp(2 \pi i \nu/\hbar), 
\end{equation}
where $\gamma$ is a closed cycle  on ${\mathcal C}'$, and 
$\nu$ is a complex parameter which is independent of $t$. 
The sign $\pm$ in \eqref{eq:partial-IMD} differs depending on each example under consideration. 
The problem under consideration here can be regarded as the question of whether
it is possible to deform the accessory parameter in such a way that 
the Voros period associated with a certain cycle remains independent of $t$. 
Since the Voros period is a quantity that characterizes the monodromy data, 
we shall refer to the deformation of the accessory parameter considered here 
as a``partial isomonodromic deformation".

In comparison with the classical conformal block, it is crucial to determine the Voros period, associated with which cycle, should be considered. 
Motivated by the analysis carried out by Lisovyy--Naidiuk \cite{LN21} for $H_{\rm V}$ and $H_{{\rm IV}}$, we propose that, in general, one may proceed according to the following strategy:

\begin{itemize}
\item[(I)]
First we take a rescale of the variables of the form
\begin{equation} \label{eq:rescale}
(x, {\mathscr E}) \mapsto 
(t^{d_x} X, t^{d_{{\mathscr E}}}{\mathscr G})
\end{equation}
with some $d_x, d_{{\mathscr E}} \in {\mathbb Q}$
(if necessary). 
We also put the following ansatz on the rescaled accessory parameter 
${\mathscr G} = t^{- d_{\mathscr E}} {\mathscr E}$: 
\begin{equation} \label{eq:ansatz}
{\mathscr G} = \sum_{m \ge 0} \hbar^{2m} {\mathscr G}_m, 
\qquad 
{\mathscr G}_m = \sum_{\ell \ge 0} \Lambda^{\ell} {\mathscr G}_{m}^{[\ell]} 
\end{equation}
where the expansion variable $\Lambda$ is a certain fractional power of $t$, 
and the coefficients ${\mathscr G}_{m}^{[\ell]}$ are independent of $t$.
The above rescaling transforms the classical limit ${\mathcal C}$ given in \eqref{eq:cl-xy}
to ``rescaled classical limit'' defined by
\begin{equation}
{\mathcal C}^{\rm res} = \{(X,Y) \in {\mathbb C}^2 ~|~ Y^2 = Q_0^{\rm res}(X) \}, 
\end{equation}
where we have further put $y = t^{d_y} Y$ with a certain $d_y \in {\mathbb Q}$
so that   
\begin{equation}
Q_0^{\rm res} dX^2 = 
t^{-2 d_{y}} Q_0 dx^2 \Bigl|_{(x,{\mathscr E}) = (t^{d_x} X, t^{d_{{\mathscr E}}}{\mathscr G}_0)}
\end{equation}
has a finite and non-zero limit as $\Lambda \to 0$.

\smallskip
\item[(II)] 
If the ``limiting spectral curve'' 
\begin{equation}
{\mathcal C}^{\rm res}_{\rm deg} = \{(X,Y) \in {\mathbb C}^2 ~|~ Y^2 = Q_0^{[0]}(X) \}, 
\qquad
Q_0^{[0]}(X) = \lim_{\Lambda \to 0} Q_0^{\rm res}(X)
\end{equation}
has genus~0, we choose $\gamma$ to be a vanishing cycle that collapses to a point as $\Lambda \to 0$. 
Otherwise, we choose ${\mathscr G}_{0}^{[0]}$ to be a zero of the discriminant so that the limiting spectral curve becomes a singular curve with genus $0$, and again take $\gamma$ to be the corresponding vanishing cycle, namely, the one encircling the branch points that coalesce as $\Lambda \to 0$.
Then, 
we impose the partial isomonodromy condition \eqref{eq:partial-IMD} on the Voros period associated with the cycle $\gamma$.  
The choice of $\gamma$ allows us to reduce the elliptic integrals in \eqref{eq:vp} to a residue computation at the point to which $\gamma$ shrinks.
Thus, we can determine the coefficient  ${\mathscr G}_{m}^{[\ell]}$ as functions of  the monodromy parameter $\nu$ and the local monodromy exponents $\theta_s$ explicitly.

\smallskip
\item[(III)] 
From the expansion of ${\mathscr G}$ thus obtained, 
we translate it into the expansion of the accessory parameter based on \eqref{eq:rescale}. 
Furthermore, by formally interchanging the order of expansion in 
$\hbar$ and $\Lambda$, we obtain a power series in 
$\Lambda$ whose coefficients are power series in $\hbar$. 
We then compare this double series expansion in 
$\Lambda$ and $\hbar$ with the classical conformal block.
\end{itemize}

Our approach is a slight modification of the method used by Lisovyy--Naidiuk \cite{LN21}; 
in fact, our parameter $\hbar$ plays a different role from that in the work. 
We also note that there are several possible choices of 
the rescaling data $(d_x, d_y, d_{\mathscr E})$
for each Heun equation, 
and we will have several asymptotic expansion of the accessory parameter depending on the choice. 
Strictly speaking, in the definition of the Voros period, one needs to specify 
a branch of the formal solution of the Riccati equation \eqref{eq:Riccati}
(that is, a branch of $S_{-1} = \sqrt{Q_0}$, which is a solution to equation \eqref{eq:cl}). 
However, this choice is not essential for the computations below, so we omit the discussion.

In the next subsections, we take the 
doubly reduced doubly confluent Heun equation 
\begin{equation}
H_{{\rm III}_3}  ~:~ 
\hbar^2 \frac{d^2 \psi}{dx^2} = Q_{{\rm III}_3}(x, \hbar) \, \psi,
\quad
Q_{{\rm III}_3} = \frac{t}{x^3} 
- \frac{{\mathscr E}}{x^2} 
+ \frac{1}{x}.
\end{equation}
to illustrate how our computational method works. 
In particular, we derive different expansions of 
the accessory parameter depending on the choice of scaling degrees and the cycle $\gamma$ on the spectral curve. 
One expansion corresponds to conformal blocks with 
regular singular behavior, while the other is expected to 
provide an expansion of (ramified) irregular conformal blocks. 
These computational results suggest that the WKB method 
offers a unified framework applicable 
to both regular and irregular cases.

\smallskip
\subsubsection{Expansion of the accessory parameter of the Heun equation $H_{{\rm III}_3}$  as $t \to \infty$} 
\label{sec:H-III-3-expansion-example-2}


Here we will derive the formal expansion of the accessory parameter as $t \to \infty$.  
We perform the rescaling with 
$(d_{x}, d_{y}, d_{\mathscr E}) = (1/2, 1/4, 1/2)$ 
and take $\Lambda = 
t^{- 1/4}$; namely, 
we will seek the accessory parameter of the form 
\begin{equation}
{\mathscr E} 
= t^{1/2}  \sum_{k \ge 0} \hbar^{2k} {\mathscr G}_{k}(\Lambda).
\end{equation}
The rescaled spectral curve becomes
\begin{equation}
Q_{0}^{\rm res} = \dfrac{1}{X^3}-\dfrac{\mathscr{G}_0}{X^2}+\dfrac{1}{X}
\quad \xrightarrow{ \Lambda  \to 0 } \quad 
Q_0^{[0]} =  \dfrac{1}{X^3}-\dfrac{\mathscr{G}^{[0]}_0}{X^2}+\dfrac{1}{X}.
\end{equation}
We take ${\mathscr G}_0^{[0]} = 2$
so that the limiting spectral curve\footnote{
The choice 
${\mathscr G}_0^{[0]} = - 2$ 
also gives a genus 0 curve $Y^2 = {(X+1)^2}/{X^3}$. 
In this case, the coefficients of the accessory parameter can be obtained in the same way by choosing 
$\gamma$ as a small circle around $X=-1$. 
This case is essentially equivalent to the case ${\mathscr G}_0^{[0]} = 2$, 
and the result follows from \eqref{eq:AP-HIII-3-2} below by replacing
${\mathscr G}_m^{[\ell]}$ with $- e^{- \pi i \ell/2} \, {\mathscr G}_m^{[\ell]}$. 
} 
\begin{equation}
Y^2 = Q_0^{[0]}(X) = \frac{(X-1)^2}{X^3}
\end{equation}
is of genus $0$. 
Then, in this limit, it is readily seen that two of the three zeros of 
$Q^{\rm res}_0$ coalesce at $X = 1$. 
Therefore, for sufficiently small $\Lambda$, 
we choose the cycle $\gamma$ to be a circle that encloses the two turning points 
which tend to the double zero $X=1$ of $Q_{0}^{[0]}$ inside, 
while keeping the other turning point outside.
Then, the elliptic integrals defining the Voros periods are described 
by the residue at $X = 1$ of the coefficients in the small $\Lambda$-expansion. 
With this choice of cycle $\gamma$, the following holds.

\begin{thm} \label{thm:unique-existence}
For any $\nu \in {\mathbb C}$, 
there exists a unique sequence $\{ {\mathscr G}_{k}(\Lambda) \}_{k \ge 0}$ 
of holomorphic functions at $\Lambda = 0$ such that
the Voros period associated with $\gamma$ becomes 
\begin{equation} \label{eq:partial-IMD-example-H-III-3}
V_\gamma = \frac{2 \pi i \nu}{\hbar} - \pi i
\end{equation}
(i.e., $V_{-1} = 2 \pi i \nu$, $V_0 = - \pi i$ and $V_m = 0 ~~(m \ge 1)$; 
that implies the equality \eqref{eq:partial-IMD} with the $-$-sign). 
Moreover, the series coefficients ${\mathscr G}_{k}^{[\ell]}$ 
of ${\mathscr G}_{k}(\Lambda)$ at $\Lambda = 0$, given in \eqref{eq:ansatz}, 
are explicitly computable at arbitrary order. 
\end{thm}

\begin{proof}

We divide the proof into several steps.
\begin{itemize}
\item[(i)] 
First, assume that the expansion \eqref{eq:ansatz} of ${\mathscr G}$ 
may also contain odd powers of $\hbar$, as follows: 
\begin{equation} \label{eq:odd-even-expansion-of-G}
{\mathscr G} = \sum_{m \ge 0} \hbar^{m} {\mathscr U}_{m}(\Lambda).
\end{equation}
Under this assumption, we use the implicit function theorem to prove the existence 
and uniqueness of functions 
\begin{equation}
{\mathscr U}_{m}(\Lambda) = \sum_{\ell \ge 0} \Lambda^{\ell} {\mathscr U}_{m}^{[\ell]}
\end{equation}
for $m \ge 0$, which are holomorphic at $\Lambda=0$, 
satisfying the condition \eqref{eq:partial-IMD-example-H-III-3}. 
\item[(ii)]
In the course of the above argument, we will provide an 
algorithm for computing the expansion coefficients 
${\mathscr U}_{m}^{[\ell]}$ explicitly at arbitrary order.
\item[(iii)]
We then show that, ${\mathscr U}_{m}$ with odd $m$ vanish identically, 
from which it follows that ${\mathscr G}$ is a series of the form \eqref{eq:ansatz} 
with ${\mathscr G}_{k} = {\mathscr U}_{2k}$.
\end{itemize}
In what follows, we prove the proposition along this line of approach.

\bigskip
\noindent 
\underline{(i) : Unique existence of ${\mathscr U}_{m}(\Lambda)$.} \quad 
Following the above setup, 
let us impose condition \eqref{eq:partial-IMD-example-H-III-3}.
The comparison of the coefficients of $\hbar^{-1}$ 
in the condition \eqref{eq:partial-IMD-example-H-III-3} reads:
\begin{equation} \label{eq:perturbative-residue-0-alt}
V_{-1} = \oint_\gamma \sqrt{Q_0(x)} dx 
= \Lambda^{-1} \oint_\gamma \sqrt{Q^{\rm res}_0(X)} \, dX 
= 2 \pi i \nu, 
\end{equation}
which is equivalent to 
\begin{equation} \label{eq:perturbative-residue-0-alt-2}
\oint_\gamma \sqrt{\dfrac{1}{X^3}-\dfrac{\mathscr{U}_0}{X^2}+\dfrac{1}{X}} \, dX 
= 2 \pi i \nu \Lambda.
\end{equation} 
The equality is satisfied at the point $(\Lambda, \mathscr{U}_0) = (0, \mathscr{U}^{[0]}_0) = (0,2)$
since the integrand of the left-hand side of \eqref{eq:perturbative-residue-0-alt-2}  
becomes holomorphic at $X=1$ 
and has trivial residue there. 
If we set ${\mathscr U}_0 = {\mathscr U}^{[0]}_0 +u = 2 + u$, then 
\begin{equation} \label{eq:condition-for-implicit-function-theorem}
\left[
\frac{\partial}{\partial u}  
\oint_\gamma \sqrt{Q^{\rm res}_0(X)} \, dX 
\right]
\biggl|_{(\Lambda,u) = (0,0)}
= 2 \pi i \Res_{X = 1} \left( - \frac{dX}{2X^2 \sqrt{Q_0^{[0]}}} \right) 
= - \pi i \ne 0.
\end{equation}
Therefore, 
the implicit function theorem guarantees the unique existence of 
the holomorphic function $u = u(\Lambda)$
satisfying $u(0) = 0$ and \eqref{eq:perturbative-residue-0-alt-2}. 
The Taylor expansion $u(\Lambda) = \sum_{\ell \ge 1} {\mathscr U}_0^{[\ell]} \Lambda^\ell$
of $u$ gives the higher-order coefficients of ${\mathscr U}_0(\Lambda) \, (= \mathscr{G}_0(\Lambda))$.

The same logic applies to ${\mathscr U}_m$ for $m \geq 1$. 
We assume that all ${\mathscr U}_{m'}$ for $0 \le m' \le m$ 
have been determined. 
From the recursion \eqref{eq:wkb-rec}, we have 
\begin{equation}
S_{m} = - \frac{1}{2S_{-1}}\left( \sum_{\substack{m_1, m_2 \ge 0 \\ m_1 + m_2 = m-1}} S_{m_1} S_{m_2} + \frac{dS_{m-1}}{dX}\right) - \dfrac{{\mathscr U}_{m+1}(\Lambda)}{2X^2 S_{-1}},
\end{equation}
and we are thus led to require that 
\begin{equation}
\oint_{\gamma} S_m(X)\, dX = 
\begin{cases} 
- \pi i & m = 0 \\
0 & m \ge 1
\end{cases}
\end{equation}
to have condition \eqref{eq:partial-IMD-example-H-III-3}.
By equation \eqref{eq:condition-for-implicit-function-theorem},  
we can apply the implicit function theorem again, 
which guarantees the existence of the function ${\mathscr U}_{m+1}(\Lambda)$ 
that is holomorphic at $\Lambda = 0$ and satisfying our requirement.

\bigskip
\noindent 
\underline{(ii) : Explicit computation of coefficient ${\mathscr U}_{m}^{[\ell]}$.} \quad
Let us once again examine the coefficient of $\hbar^{-1}$ 
in the equation \eqref{eq:partial-IMD-example-H-III-3} we imposed: 
\begin{equation} \label{eq:perturbative-residue-0}
2 \pi i  \nu \Lambda 
= \oint_\gamma \sqrt{Q^{\rm res}_0(X)} \, dX 
=  2\pi i \, \sum_{\ell \ge 0} \Lambda^{\ell} \, 
\mathop{\mathrm{Res}}_{X = 1} S_{-1}^{[\ell]}(X) \, dX,  
\end{equation}
where $S_{-1}^{[\ell]}$ are determined as 
\begin{equation} \label{eq:S-1-ell}
\sqrt{Q^{\rm res}_0(X)} = 
\sqrt{\dfrac{1}{X^3}-\dfrac{\mathscr{U}_0}{X^2}+\dfrac{1}{X}} 
= \sum_{\ell \ge 0} \Lambda^{\ell}S_{-1}^{[\ell]}(X), 
\end{equation} 
and explicitly given as follows:
\[
S_{-1}^{[0]}(X) = \sqrt{Q_0^{[0]}(X)} = \frac{X-1}{X^{3/2}}, \quad
S_{-1}^{[1]}(X) = - \frac{{\mathscr U}_{0}^{[1]}}{2X^2 S_{-1}^{[0]}(X)}, 
\]
\begin{equation}
    S_{-1}^{[2]}(X) = 
- \frac{{\mathscr U}_{0}^{[2]}}{2X^{2} S_{-1}^{[0]}(X)}
- \frac{(1 - X\, {\mathscr U}_{0}^{[1]})^{2}}
{8X^{6} S_{-1}^{[0]}(X)^3}, 
\quad \dots.
\end{equation}
The condition \eqref{eq:perturbative-residue-0}
gives a set of constraints on ${\mathscr U}_{0}^{[\ell]}$
which determines them uniquely as follows:
\begin{equation}
{\mathscr U}_{0}^{[1]} = - 2 \, \nu, \quad
{\mathscr U}_{0}^{[2]} = \frac{\nu^2}{8}, \quad
{\mathscr U}_{0}^{[3]} = \frac{\nu^3}{128}, \quad
{\mathscr U}_{0}^{[4]} = \frac{5 \, \nu^4}{4096}, \quad
{\mathscr U}_{0}^{[5]} = \frac{33 \, \nu^5}{131072}, \quad
\dots
\end{equation}
A similar computation can be applied to the higher-order coefficients in $\hbar$ as well.
In principle, the coefficients ${\mathscr U}_{m}^{[\ell]}$ 
are recursively computable at any order. 
Some of the expansion coefficients are given as follows.
\[
{\mathscr U}_{1}^{[0]} = 0, \quad
{\mathscr U}_{1}^{[1]} = 0, \quad
{\mathscr U}_{1}^{[2]} = 0, \quad
{\mathscr U}_{1}^{[3]} = 0, \quad
{\mathscr U}_{1}^{[4]} = 0, \quad
{\mathscr U}_{1}^{[5]} = 0, \quad
\dots
\]
\vspace{-.5em}
\[
{\mathscr U}_{2}^{[0]} = 0, \quad
{\mathscr U}_{2}^{[1]} = 0, \quad
{\mathscr U}_{2}^{[2]} = \frac{9}{32}, \quad
{\mathscr U}_{2}^{[3]} = \frac{3 \nu}{512} , \quad
{\mathscr U}_{2}^{[4]} = \frac{17 \nu^2}{8192}, \quad
{\mathscr U}_{2}^{[5]} = \frac{205 \nu^3}{262144}, \quad
\dots
\]
\vspace{+.2em}
\[
{\mathscr U}_{3}^{[0]} = 0, \quad
{\mathscr U}_{3}^{[1]} = 0, \quad
{\mathscr U}_{3}^{[2]} = 0, \quad
{\mathscr U}_{3}^{[3]} = 0, \quad
{\mathscr U}_{3}^{[4]} = 0, \quad
{\mathscr U}_{3}^{[5]} = 0, \quad
\dots
\]
\vspace{-.1em}
\begin{equation}
{\mathscr U}_{4}^{[0]} = 0, \quad
{\mathscr U}_{4}^{[1]} = 0, \quad
{\mathscr U}_{4}^{[2]} = 0, \quad
{\mathscr U}_{4}^{[3]} = 0, \quad
{\mathscr U}_{4}^{[4]} = \frac{9}{65536} , \quad
{\mathscr U}_{4}^{[5]} = \frac{405\nu}{2097152} , \quad
\dots
\end{equation}

\bigskip
\noindent 
\underline{(iii) : Vanishing of ${\mathscr U}_{m}$ with odd $m$.} \quad
From the above computation, we can observe that all coefficients of 
${\mathscr U}_1$ and ${\mathscr U}_3$ seem to vanish.
To conclude the proof, we show that all ${\mathscr U}_m$ with odd $m$ are trivial.


Let us introduce a formal solution
$\widetilde{S}(X,\hbar) = \sum_{m \ge -1} \hbar^{m} \widetilde{S}_{m}(X)$ 
to the same Riccati equation \eqref{eq:Riccati}
obtained by imposing a certain condition different from 
\eqref{eq:partial-IMD-example-H-III-3} on ${\mathscr U}_{m}$.
Namely, 
\begin{itemize}
\item 
The even-index terms ${\mathscr U}_0, {\mathscr U}_2, \dots$ in 
$\widetilde{S}_{-1}(X), \widetilde{S}_1(X), \ldots$ are 
fixed by imposing the condition \eqref{eq:partial-IMD-example-H-III-3} as before.
\item 
We set ${\mathscr U}_1 = {\mathscr U}_3 = \cdots = 0$ for the odd-index terms, 
and denote by $\widetilde{S}_0(X), \widetilde{S}_2(X), \ldots$
the coefficient of the formal solution of the Riccati equation \eqref{eq:Riccati} 
determined in this way.
\end{itemize}
The decomposition of $\widetilde{S}(X,\hbar)$ in odd/even degree terms in $\hbar$ 
is given as
$\widetilde{S}(X,\hbar) = \widetilde{S}_{\rm odd}(X,\hbar)+\widetilde{S}_{\rm even}(X,\hbar)$.
We can also verify that
\begin{equation} \label{eq:even-odd-Stilde}
\widetilde{S}_{\text{even}}(X,\hbar) = 
-\dfrac{1}{2} \dfrac{d}{dX} \log \widetilde{S}_{\text{odd}}(X,\hbar)
\end{equation}
holds by the same argument of \cite[\S 2]{KT98}.

Below, we prove the following statement by induction on $k$:
\begin{quote}
$(\ast)_k$ : $S_{m}(X) = \widetilde{S}_{m}(X)$ holds for $-1 \le m \le 2k$, 
and ${\mathscr U}_1 = {\mathscr U}_3 =\cdots = {\mathscr U}_{2k+1} = 0$ holds.
\end{quote}

From the construction and the uniqueness of ${\mathscr U}_0$, 
it is clear that $S_{-1}(X) = \widetilde{S}_{-1}(X)$ holds.
The coefficients of $\hbar^{0}$ 
in the condition \eqref{eq:partial-IMD-example-H-III-3} reads
\begin{equation} \label{eq:partial-IMD-example-hbar-1}
    - \pi i=\oint_\gamma S_{0}(X) dX 
    = \oint_\gamma \dfrac{- \frac{dS_{-1}(X)}{dX}-\frac{{\mathscr U}_1(\Lambda)}{X^2}}{2S_{-1}(X)} dX,
\end{equation}
where we have used \eqref{eq:wkb-rec} to obtain the expression of $S_0(X)$. 
Since 
\begin{equation}
- \pi i= \oint_{\gamma} \dfrac{- \frac{dS_{-1}(X)}{dX}}{2S_{-1}(X)} dX
\end{equation}
holds, 
it follows that the condition \eqref{eq:partial-IMD-example-hbar-1} 
is satisfied if and only if ${\mathscr U}_1 = 0$.
The uniqueness of ${\mathscr U}_1$ satisfying \eqref{eq:partial-IMD-example-H-III-3}  
implies 
$S_{0}(X) = \widetilde{S}_{0}(X)$. Thus we have verified $(\ast)_0$.

As the induction hypothesis, we assume that $(\ast)_k$ holds for some $k \ge 0$.
It follows from 
the fact that ${\mathscr U}_{2k+2}$ is uniquely determined by condition  
\eqref{eq:partial-IMD-example-H-III-3} that we have $S_{2k+1}(X) = \widetilde{S}_{2k+1}(X)$.
The recurrence relations that determine $S_{2k+2}(X)$ and $\widetilde{S}_{2k+2}(X)$ reads:
\begin{align}
S_{2k+2} & = - \frac{1}{2S_{-1}}\left( \sum_{\substack{m_1, m_2 \ge 0 \\ m_1 + m_2 = 2k+1}} S_{m_1} S_{m_2} + \frac{dS_{2k+1}}{dX}\right) - \dfrac{{\mathscr U}_{2k+3}(\Lambda)}{2X^2 S_{-1}}, \\
\widetilde{S}_{2k+2} & = 
- \frac{1}{2\widetilde{S}_{-1}}\left( 
\sum_{\substack{m_1, m_2 \ge 0 \\ m_1 + m_2 = 2k+1}} \widetilde{S}_{m_1} \widetilde{S}_{m_2} 
+ \frac{d\widetilde{S}_{2k+1}}{dX}\right).
\end{align}
By the induction hypothesis, we have
\begin{equation} \label{eq:S-and-Stilde-even}
S_{2k+2} - \widetilde{S}_{2k+2} = - \dfrac{{\mathscr U}_{2k+3}(\Lambda)}{2X^2 S_{-1}}.
\end{equation}
The important fact here is that, by 
\eqref{eq:even-odd-Stilde}, 
the even-index terms 
$\widetilde{S}_{2k+2}$ 
admit anti-derivative with respect to $X$. 
More concretely, the equality 
\begin{align} 
\log \widetilde{S}_{\rm odd}&=\log(\widetilde{S}_{-1}\hbar^{-1}
+\widetilde{S}_1\hbar^1+\widetilde{S}_3\hbar^3+\cdots) 
\notag 
\\
&=\log(\widetilde{S}_{-1}\hbar^{-1})  + 
\biggl(\dfrac{\widetilde{S}_1}{\widetilde{S}_{-1}}\hbar^2
+\dfrac{\widetilde{S}_3}{\widetilde{S}_{-1}}\hbar^4+\cdots\biggr)
- \dfrac{1}{2}\biggl(\dfrac{\widetilde{S}_1}{\widetilde{S}_{-1}} \hbar^2
+ \dfrac{\widetilde{S}_3}{\widetilde{S}_{-1}}\hbar^4+\cdots\biggr)^2+\cdots
\label{eq:log-Sodd-tilde}
\end{align}
and \eqref{eq:even-odd-Stilde} imply that the terms $\widetilde{S}_{2k+2}$ with $k \ge 0$
are expressible as total derivatives of rational expressions in the 
$\widetilde{S}_{-1}, \widetilde{S}_{0}, \cdots, \widetilde{S}_{2k+1}$. 
Hence, it follows that $\oint_\gamma \widetilde{S}_{2k+2}(X) \, dX = 0$ holds for $k \ge 0$, 
and we have
\begin{equation}
\oint_\gamma {S}_{2k+2}(X) \, dX = 
- {\mathscr U}_{2k+3}(\Lambda) \oint_{\gamma} \dfrac{dX}{2X^2 S_{-1}}
\end{equation}
due to \eqref{eq:S-and-Stilde-even}.
Since our requirement \eqref{eq:partial-IMD-example-H-III-3} is that 
the left-hand side of the above equation vanish, 
uniqueness of ${\mathscr U}_{2k+3}$ implies that ${\mathscr U}_{2k+3}(\Lambda) = 0$.
This establishes the claim $(\ast)_{k+1}$, and hence, 
by induction, we have proved ${\mathscr U}_{m} = 0$ for odd $m$. This completes the proof.
\end{proof}

\begin{rem} 
From the proof above, we can conclude the following: 
If we denote by $d$ the zero order of the leading term of $S_{-1}=\sqrt{Q_0}$ at the point where the residue is taken, the condition 
\begin{equation}
V_\gamma = \dfrac{2\pi i\nu}{\hbar}-d\pi i 
\end{equation}
implies that ${\mathscr U}_m=0$ hold for all odd $m$. 
In \S \ref{section:examples}, when analyzing other equations as well, 
we will impose this condition as the analogue of 
\eqref{eq:partial-IMD-example-H-III-3}.
\end{rem}

After explicit computation, we obtained the following result: 
\begin{subequations} 
\label{eq:AP-HIII-3-2}
\begin{align} 
{\mathscr G}_0 & =  2 - 2 \, \nu \, \Lambda 
+ \dfrac{\nu^2}{8}\, \Lambda^2 
+\dfrac{\nu^3}{128} \,\Lambda^3
+\dfrac{5 \, \nu^4}{4096} \, \Lambda^4
+\dfrac{33 \, \nu^5}{131072} \, \Lambda^5
+ O(\Lambda^6), \\
{\mathscr G}_1 & =  
 \frac{9}{32} \,   \Lambda^{2}
+ \frac{3 \, \nu}{512} \, \Lambda^{3}
+ \frac{17 \, \nu^2}{8192} \, \Lambda^{4}
+ \frac{205 \, \nu^3}{262144} \, \Lambda^{5}
+ O(\Lambda^6), \\
{\mathscr G}_2 & 
= \frac{9}{65536}\, \Lambda^{4} 
+ \frac{405 \, \nu}{2097152} \,  \Lambda^{5}
+ O(\Lambda^6), \quad \dots
\end{align}
\end{subequations}
Furthermore, by formally interchanging the order of $\hbar$-expansion and large $t$-expansion, we obtain the following formal series expansion for the accessory parameter:
\begin{align}
{\mathscr E} & = 
2 \, t^{1/2} 
- 2\, \nu \, t^{1/4} 
+ \frac{4 \, \nu^2+9 \, \hbar^2}{32} 
+ \frac{4\, \nu^3+3 \, \hbar^2\nu}{512} \, t^{-1/4}
+ \frac{80 \, \nu^4 +136 \,\hbar^2 \nu^2 +9 \, \hbar^2}{65536} \, t^{-1/2} 
\notag \\
& \quad + \frac{528 \,\nu^5 + 1640 \, \hbar^2 \nu^3 + 405\, \hbar^4 \nu}{2097152} \, t^{-3/4}
+ \frac{9\, \left(224\, \nu^6 + 1120\, \hbar^2 \nu^4 + 654 \, \hbar^4 \nu^2 + 27 \, \hbar^6\right)}{33554432} \, t^{-1} 
\notag \\
& \quad  + \frac{33728 \, \nu^7 + 249872\, \hbar^2 \nu^5 + 276004 \, \hbar^4 \nu^3 + 41607 \, \hbar^6 \nu}{2147483648} \, t^{-5/4} 
+ O(t^{-3/2}).
\label{eq:final-AP-HIII-3-2}
\end{align}

\smallskip
\subsubsection{Expansion of the Accessory Parameter of Heun equation $H_{{\rm III}_3}$  
as $t \to 0$}
\label{subsec:example-AP-III-regular}


Here we will derive the formal expansion of the accessory parameter as $t \to 0$.  
We note that this expansion was obtained in \cite{MM09, PP14} 
(see also \cite[Appendix A and B.4]{LN21} for another derivation).

In this case we do not have to perform any rescaling (i.e., $d_{x} = d_{y} = d_{\mathscr E} = 0$), and can take $\Lambda = t$ as the expansion parameter.
The (rescaled) classical limit and its limit as $\Lambda \to 0$ reads
\begin{equation}
Q_0^{\rm res}(X) = \frac{\Lambda}{X^3} 
- \frac{{\mathscr G}_0}{X^2} + \frac{1}{X} 
\quad \xrightarrow{ \Lambda \to 0 } \quad 
Q_0^{[0]} = - \frac{{\mathscr G}_0^{[0]}}{X^2} + \frac{1}{X}. 
\end{equation}
In this limit, it is readily seen that two of the three zeros of 
$Q^{\rm res}_0(X)$ coalesce at the origin, 
and the limiting curve $Y^2 = Q_0^{[0]}(X)$ is a genus $0$ curve. 
We note that the limiting curve is equivalent to the Bessel curve in \cite{IKT18-2}. 
We therefore choose, as the vanishing cycle $\gamma$ in this limit, 
the closed path encircling these two zeros together with the origin.
The cycle $\gamma$ can be identified with the residue cycle around $X = 0$ 
after taking the limit $\Lambda \to 0$. 
With this choice of cycle $\gamma$, the following holds.

\begin{thm} \label{thm:unique-existence-2}
For any $\nu \in {\mathbb C} \setminus \{0\}$, there exist a unique sequence $\{ {\mathscr G}_{k}(\Lambda) \}_{k \ge 0}$ of holomorphic functions at $\Lambda = 0$ such that the Voros period associated with $\gamma$ becomes 
\begin{equation} \label{eq:partial-IMD-example-H-III-3-2}
V_\gamma = \frac{2 \pi i \nu}{\hbar} + \pi i
\end{equation}
(i.e., $V_{-1} = 2 \pi i \nu$, $V_0 = \pi i$ and $V_m = 0 ~~(m \ge 1)$; that implies the equality \eqref{eq:partial-IMD} with the $-$-sign). 
Moreover, the series coefficients ${\mathscr G}_{k}^{[\ell]}$ of ${\mathscr G}_{k}(\Lambda)$ can be explicitly computable at arbitrary order. 
\end{thm}

\begin{proof}
The strategy of the proof is completely the same as that of Theorem \ref{thm:unique-existence}. 
However, we would like to explain why the condition $\nu \ne 0$ is required in this case.

The comparison of the coefficients of $\hbar^{-1}$ 
in the condition \eqref{eq:partial-IMD-example-H-III-3-2} reads
\begin{equation} \label{eq:partial-IMD-example-hbar-0-2}
V_{-1} 
= \oint_\gamma \sqrt{Q^{\rm res}_0(X)} \, dX
= 2 \pi i \nu. 
\end{equation}
Since $\nu$ does not depend on $t$ (and hence, on $\Lambda$), 
our requirement implies that this equality also hold at $\Lambda = 0$. 
Our choice of the cycle $\gamma$ implies 
\begin{equation}
2 \pi i \nu 
=  \oint_\gamma \sqrt{Q_0^{[0]}(X)} \, dX
= 2 \pi i \Res_{X=0} \sqrt{- \frac{{\mathscr U}_0^{[0]}}{X^2} + \frac{1}{X}} \, dX
= 2 \pi i \sqrt{- \mathscr{U}_0^{[0]}}.
\end{equation}
The condition uniquely fixes the leading term and we have 
\begin{equation} \label{eq:fixing-leading-H-III-3-2}
\mathscr{U}_0^{[0]} = \mathscr{G}_0^{[0]} = - \nu^2.
\end{equation} 
If we set ${\mathscr U}_0 = - \nu^2 + u$, then 
\begin{equation} \label{eq:condition-for-implicit-function-theorem-2}
\left[
\frac{\partial}{\partial u}  
\oint_\gamma \sqrt{Q^{\rm res}_0(X)} \, dX 
\right]
\biggl|_{(\Lambda,u) = (0,0)}
= 2 \pi i \Res_{X = 0} \left( - \frac{dX}{2X^2 \sqrt{Q_0^{[0]}}} \right) 
= - \frac{\pi i}{\nu}.
\end{equation}
Therefore, we need the condition $\nu \ne 0$, and the implicit function theorem guarantees the unique existence of the holomorphic function $u = u(\Lambda)$ satisfying $u(0) = 0$ and \eqref{eq:partial-IMD-example-hbar-0-2}. 
Under this condition, the proof of the unique existence of the higher-order coefficients 
${\mathscr U}_{m}$, as well as the proof of the vanishing of the coefficients for odd $m$, are completely same as those in Theorem \ref{thm:unique-existence}, and are therefore omitted here.
\end{proof}

After explicit computation, we obtained the following result: 
\begin{subequations} 
\begin{align} 
{\mathscr G}_0 & = - \nu^2 - \frac{\Lambda}{2 \, \nu^2} 
- \frac{5 \, \Lambda^2}{32 \, \nu^6} 
- \frac{9 \, \Lambda^3}{64 \, \nu^{10}} 
- \frac{1469\,\Lambda^4}{8191 \, \nu^{14}} 
+ O(\Lambda^5), \\
{\mathscr G}_1 & = \frac{1}{4} - \frac{\Lambda}{8 \, \nu^4} 
- \frac{21 \, \Lambda^2}{64 \, \nu^8} 
- \frac{55 \, \Lambda^3}{64 \, \nu^{12}} 
- \frac{18445 \, \Lambda^4}{8191 \, \nu^{16}} + O(\Lambda^5), \\
{\mathscr G}_2 & = - \frac{\Lambda}{32 \, \nu^6} 
-\frac{219 \, \Lambda^2}{512 \, \nu^{10}} 
-\frac{1495 \, \Lambda^3}{512 \, \nu^{14}} 
-\frac{69853 \, \Lambda^4}{8192 \, \nu^{18}}  + O(\Lambda^5), \quad \dots
\end{align}
\end{subequations}
Thus we have obtained a formal series ${\mathscr G}$ 
satisfying \eqref{eq:partial-IMD-example-H-III-3-2}
whose coefficients ${\mathscr G}_k$ are holomorphic around $t = 0$. 
These equalities are equivalent to the results of \cite{MM09, PP14}.

Furthermore, by formally interchanging the order of 
$\hbar$-expansion and small $\Lambda$-expansion (i.e., small $t$-expansion), 
we obtain the following formal series expansion for the accessory parameter:
\begin{align}
{\mathscr E} & = 
\left(-\nu^2 + \dfrac{\hbar^2}{4}\right) +\left(- \dfrac{1}{2\nu^2}- \dfrac{\hbar^2}{8 \, \nu^4}- \dfrac{\hbar^4}{32 \, \nu^6}+O(\hbar^6)\right)\, t \notag \\
&\quad +\left(- \dfrac{5}{32 \, \nu^6} - \dfrac{21 \, \hbar^2}{64 \, \nu^8}
- \dfrac{219 \, \hbar^4}{512 \, \nu^{10}} + O(\hbar^6)\right)\, t^2 \notag \\
&\quad+\left(- \dfrac{9}{64 \, \nu^{10}} - \dfrac{55 \, \hbar^2}{64 \, \nu^{12}}
- \dfrac{1495 \, \hbar^4}{512 \, \nu^{14}} + O(\hbar^6)\right)\, t^3+O(t^4)  \notag \\
& = 
\left(-\nu^2 + \dfrac{\hbar^2}{4} \right)
- \frac{2}{4\,\nu^2 - \hbar^2} \, t 
- \frac{20\, \nu^2+7 \,\hbar^2}{2\,(4\,\nu^2 - \hbar^2)^3 (\nu^2-\hbar^2)} \, t^2 
\notag \\
& \quad  
- \frac{4 \, (144 \, \nu^4 + 232 \, \hbar^2  \nu^2+ 29 \, \hbar^4 )}{(4 \, \nu^2 - \hbar^2)^5 (4 \, \nu^2 - 9\hbar^2)(\nu^2-\hbar^2)} \, t^3
+ O(t^4).
\label{eq:final-AP-HIII-3-1}
\end{align}

In next subsection, we will compare the series \eqref{eq:final-AP-HIII-3-2} and \eqref{eq:final-AP-HIII-3-1} with the classical irregular conformal blocks.

\smallskip
\begin{rem} \label{rem:rational-expression}
Using ${\mathscr G}^{[\ell]}_m$ determined above, we reconsider how the expansions 
\begin{equation}
    S = \sum_{\ell \ge 0} \Lambda^{\ell} \, S^{[\ell]}(X, \hbar)\quad,\quad
    {\mathscr G} = \sum_{\ell \ge 0} \Lambda^{\ell} {\mathscr G}^{[\ell]}(\hbar) 
\end{equation}
at $\Lambda=0$ are obtained.
Substituting this into the Riccati equation \eqref{eq:Riccati} and now comparing coefficients in the power series expansion in $\Lambda$, we obtain 
\begin{subequations}  \label{eq:Ricatti-HIII-1}
\begin{align}
\hbar^2\biggl((S^{[0]})^2+\dfrac{dS^{[0]}}{dX}\biggr)&=\dfrac{1}{X}-\dfrac{\mathscr{G}^{[0]}}{X^2},\label{Sh-Ri1}\\
\hbar^2\biggl(2S^{[0]}S^{[1]}+\dfrac{dS^{[1]}}{dX}\biggr)&=\dfrac{1}{X^3}-\dfrac{\mathscr{G}^{[1]}}{X^2},\label{Sh-Ri2}\\
\hbar^2\biggl(S^{[0]}S^{[\ell]}+\cdots+S^{[\ell]}S^{[0]}+\dfrac{dS^{[\ell]}}{dX}\biggr)&=-\dfrac{\mathscr{G}^{[\ell]}}{X^2}\quad(\ell\geq 2). \label{Sh-Ri3}
\end{align}
\end{subequations}
The condition \eqref{eq:partial-IMD-example-H-III-3-2} imposes 
\begin{subequations} \label{eq:residue-condition-HIII-1}
\begin{align}
\Res_{X=0} S^{[0]}&=\dfrac{\nu}{\hbar}+\dfrac{1}{2}\eqcolon\lambda,\label{Sh-Re1}\\
\Res_{X=0} S^{[\ell]}&=0\quad(\ell\geq 1).\label{Sh-Re2}
\end{align}
\end{subequations}
Then we can uniquely determine $S^{[0]}$ as a formal power series starting from the $X^{-1}$ term which satisfies (\ref{Sh-Ri1}) and (\ref{Sh-Re1}) :
\begin{equation}
S^{[0]}=\lambda X^{-1} + \dfrac{1}{2 \, \lambda \, \hbar^2} 
- \dfrac{1}{4\,\lambda^2(2\lambda+1) \, \hbar^4} \, X+O(X).
\end{equation}
Hence, by comparing the coefficients of $X^{-2}$ in (\ref{Sh-Ri1}), we obtain 
\begin{equation}
\mathscr{G}^{[0]} = 
-\hbar^2 \, \lambda (\lambda-1)= - \nu^2 + \frac{\hbar^2}{4}. 
\end{equation}

Likewise, under the condition that $2\lambda$ is not an integer (i.e., ${2\nu}/{\hbar}$ is not an integer), we can uniquely determine $S^{[\ell]}$ inductively for $\ell=1, 2, \ldots$ as a formal power series starting from a term $X^{-\ell-1}$, together with $\mathscr{G}^{[\ell]}$, as follows. 

If we denote by $S^{[\ell]}_m$ the coefficient of $X^m$ in $S^{[\ell]}$, 
the recursion relation (\ref{eq:Ricatti-HIII-1}) implies that 
it must satisfies the equation of the form 
$(2\lambda + m) \, S^{[\ell]}_m = F^{[\ell]}_m$, where $F^{[\ell]}_m$ is a polynomial of $S^{[\ell']}_{m'}$ with $\ell' < \ell$, or 
$\ell' = \ell$ and $m' < m$. 
An important observation here is that $F^{[\ell]}_m$ contains the term $- \hbar^{-2} {\mathscr G}^{[\ell]}$ if and only if $m=-1$. 
Therefore, if $2\lambda$ is not an integer, then $S_m^{[\ell]}$ is uniquely determined for all $m \neq -1$. 
On the other hand, (\ref{Sh-Re2}) implies $S_{-1}^{[\ell]} = 0$, which also implies $F_{-1}^{[\ell]} = 0$, and this condition determines $\mathscr{G}^{[\ell]}$ uniquely. 

Some explicit calculation results are as follows:
\begin{subequations}
    \begin{align}
S^{[1]} & 
=\dfrac{1}{(2 \,  \lambda-2) \, \hbar^2} \, X^{-2} + \dfrac{1}{4 \, \lambda^3 \, (2\lambda-2)(2\lambda+1) \, \hbar^6}
+ O(X), \\
S^{[2]} & 
= -\dfrac{1}{(2\lambda-3)(2\lambda-2)^2 \, \hbar^4} \, X^{-3} + \dfrac{1}{\lambda \, (2\lambda-3)(2\lambda-2)^3 \, \hbar^6} \, X^{-2}+ O(X^{0}), 
\end{align}
\end{subequations}
and 
\begin{subequations}
    \begin{align}
\mathscr{G}^{[1]} & =-\frac{1}{\lambda \, (2 \, \lambda-2) \, \hbar^2} = - \frac{2}{4 \, \nu^2 - \hbar^2}, \\
\mathscr{G}^{[2]} & =-\frac{5 \, \lambda^2-5 \, \lambda+3}{\lambda^3 \,  (2 \, \lambda-2)^3(2\,\lambda-3)(2\,\lambda+1) \, \hbar^6}=-\frac{20 \, \nu^2 + 7 \, \hbar^2}{2 \, (4 \, \nu^2-\hbar^2)^3(\nu^2-\hbar^2)}.
\end{align}
\end{subequations}
From this method, it follows that $\mathscr{G}^{[\ell]}$ is a rational function of $\hbar$ and $\nu$, 
and the coefficients ${\mathscr G}^{[\ell]}_{m}$ computed above can be recovered as 
its Taylor coefficients at $\hbar = 0$. 
The final equality in \eqref{eq:final-AP-HIII-3-1} is derived in this way. 

It is also found that when this method is applied to accessory parameters exhibiting irregular singular behavior such as the one studied in \S \ref{sec:H-III-3-expansion-example-2}, the coefficients are not rational functions but rather polynomials in $\nu$ and $\hbar$, as written in \eqref{eq:final-AP-HIII-3-2}.


\end{rem}

\bigskip
\subsection{Classical conformal block and Zamolodchikov-type conjecture}

As mentioned in the introduction, Bonelli--Shchechkin--Tanzini have constructed the candidates ${\mathscr Z}(\varepsilon_1, \varepsilon_2)$ of the Nekrasov-type partition functions, or the corresponding conformal blocks, as series solutions to these bilinear equations for the quantum Painlev\'e $\tau$-function \cite{BST25}.  
What we aim to do in this subsection is to extend Zamolodchikov's conjecture by introducing a series ${\mathscr W}$ as the classical limit of ${\mathscr Z}$ proposed in \cite{BST25}, and compare this with the accessory parameter ${\mathscr E}$ computed in the previous section. 
In particular, we demonstrate this procedure using the case of quantum Painlev\'e ${\rm III}_3$ equation (${\rm QPIII}_3$) as an explicit example. 
(Comparisons for other examples will be discussed in \S \ref{section:examples}.)

\subsubsection{Bilinear equations for the quantum Painlev\'e $\tau$-function}

Following \cite{BST25}, we describe a method for computing series solutions of (candidates of) irregular conformal blocks that satisfy the bilinear equations for quantum Painlev\'e $\tau$-functions. 
In particular, we demonstrate this procedure using ${\rm QPIII}_3$ as an explicit example. 

The quantum Painlev\'e $\tau$-function of \cite{BST25} is defined by the Zak transform of a function $\mathscr{Z}(a; \varepsilon_1, \varepsilon_2 | t)$ as 
\begin{align*}
\tau(a, \eta; \varepsilon_1, \varepsilon_2 | t) = \sum_{n \in \mathbb{Z}} e^{in\eta} \mathscr{Z}(a + n\varepsilon_2; \varepsilon_1, \varepsilon_2 | t),
\end{align*}
where the variables $(a, \eta)$ are assumed to be non-commutative 
and satisfy 
\begin{equation}\label{eq:noncommutative relation}
    a e^{i\eta} = e^{i\eta} (a + 2(\varepsilon_1+\varepsilon_2)).
\end{equation}
In \cite{BST25}, based on the bilinear relations for the 4-point conformal block  derived in \cite{BS15}, they obtained these bilinear equations by using the degenerate scheme of the Painlev\'e equations. 
The bilinear equations for ${\rm QPIII}_3$ $\tau$-function are 
\begin{subequations}   \label{eq:bilinear-134}
\begin{align}
& D_{\varepsilon_1, \varepsilon_2 [\ln t]}^1 (\tau^{(1)}, \tau^{(2)}) = 0, 
\label{eq:bilinear-134-a}
\\
& D_{\varepsilon_1, \varepsilon_2 [\ln t]}^3 (\tau^{(1)}, \tau^{(2)}) = (\varepsilon_1+\varepsilon_2) D_{\varepsilon_1, \varepsilon_2 [\ln t]}^2 (\tau^{(1)}, \tau^{(2)}),
\\ 
     &  
     D_{\varepsilon_1, \varepsilon_2 [\ln t]}^4 (\tau^{(1)}, \tau^{(2)}) + 2 \left( \varepsilon_1 \varepsilon_2 \frac{d}{d \ln t} \right) D_{\varepsilon_1, \varepsilon_2 [\ln t]}^2 (\tau^{(1)}, \tau^{(2)}) 
     \notag \\ 
     & \quad 
     - (\varepsilon_1 \varepsilon_2 + (\varepsilon_1+\varepsilon_2)^2) D_{\varepsilon_1, \varepsilon_2 [\ln t]}^2 (\tau^{(1)}, \tau^{(2)}) + \frac{t}{4} \tau^{(1)} \tau^{(2)} = 0. 
     \label{eq:bilinear-134-c}
\end{align}
\end{subequations}
Here, the generalized Hirota differential operators $D_{\varepsilon_1, \varepsilon_2 [\ln t]}^k$ are defined by
\begin{align}
    D^{k}_{\varepsilon_1, \varepsilon_2 [\ln t]}(f, g)=&
    \sum_{i=0}^k \varepsilon_1^i\varepsilon_2^{k-i}\binom{k}{i}
	\left(t\frac{d}{dt}\right)^if\, 
\left(t\frac{d}{dt}\right)^{k-i}g, 
\end{align}
and the shifted $\tau$-functions are defined by
\begin{equation}
\tau^{(1)}(a, \eta; \varepsilon_1, \varepsilon_2 | t) 
= \tau(a, \eta; 2\varepsilon_1, \varepsilon_2 - \varepsilon_1 | t)
\quad , \quad 
\tau^{(2)}(a, \eta; \varepsilon_1, \varepsilon_2 | t)
= \tau(a, \eta; \varepsilon_1 - \varepsilon_2, 2\varepsilon_2 | t). 
\end{equation}
The bilinear equations \eqref{eq:bilinear-134-a}--\eqref{eq:bilinear-134-c} effectively splits into two bilinear relations for $\mathscr{Z}(a; \varepsilon_1, \varepsilon_2 | t)$ due to the non-commutative relations \eqref{eq:noncommutative relation}:
\begin{equation}
    \sum_{n \in \mathbb{Z} + \frac{i}{2}} \mathfrak{D}^{(j)} \Bigl(  
    \mathscr{Z}(a + 2n\varepsilon_1; 2\varepsilon_1, \varepsilon_2 - \varepsilon_1 | t), 
    \mathscr{Z}(a + 2n\varepsilon_2; \varepsilon_1 - \varepsilon_2, 2\varepsilon_2 | t) \Bigr) = 0, 
    \quad (i=0,1, \, j=1,3,4), 
    \label{eq:bilinear split}
\end{equation}
where $\mathfrak{D}^{(j)}$ is the corresponding bilinear differential operator given in \eqref{eq:bilinear-134-a}--\eqref{eq:bilinear-134-c}.

Bonelli--Shchechkin--Tanzini \cite{BST25} succeeded in obtaining various ``quantum expansions'' that satisfy the bilinear equations \eqref{eq:bilinear-134-a}--\eqref{eq:bilinear-134-c} by imposing that ${\mathscr Z}$ admits an appropriate representation while taking into account the non-commutativity \eqref{eq:noncommutative relation} of the variables. 
Below, we present the explicit forms of the expansions of ${\mathscr Z}$ in the limits $t \to \infty$ and $t \to 0$.

\subsubsection{Expansion of the classical conformal block 
as $t \to \infty$}
\label{subsec:example-CCB-III-irregular}

Here we recall the quantum expansion of ${\mathscr Z}$, 
called ``square exp singularity of ${\rm QPIII}_3$ / $N_f=0$", 
given in \cite[(4.71)--(4.74) in \S 4.5]{BST25}. 
We note that the results of the calculations in this section are included in the result \cite{GMS20} of Gavrylenko--Marshakov--Stoyan, but are presented here for the reader’s convenience.

The ansatz imposed by \cite{BST25} for ${\mathscr Z}$ consists of three products 
\begin{equation}
    \mathscr{Z}(a_D; \varepsilon_1, \varepsilon_2 | t) = 
    \mathscr{Z}_{\text{cl}}(a_D; \varepsilon_1, \varepsilon_2 | t) \,
    \mathscr{Z}_{\text{1-loop}}(a_D; \varepsilon_1, \varepsilon_2) \,
    \mathscr{Z}_{\text{inst}}(a_D; \varepsilon_1, \varepsilon_2 | t).
\end{equation} 
Here, following \cite{BST25}, we use the variable $a_D$ instead of $a$ in describing the ``strong coupling expansion'' $t \to \infty$; the variable $a$ will be used in the ``weak coupling expansion'' $t \to 0$ in \S \ref{subsec:example-CCB-III-regular}. 
Let us further suppose that each factor is expressed as 
\begin{subequations} \label{eq:Z-QPIII-3-irregular}
\begin{align} 
    &\mathscr{Z}_{\text{cl}}(a_D; \varepsilon_1, \varepsilon_2 | s) = s^{\frac{\xi_2 + a_D^2 / 2}{\varepsilon_1 \varepsilon_2}} e^{\frac{\beta s^2 + \xi_1 s + \delta a_D s}{\varepsilon_1 \varepsilon_2}},
    \\
    &\mathscr{Z}_{\text{1-loop}}(a_D; \varepsilon_1, \varepsilon_2) = e^{-\frac{\chi a_D^2}{2\varepsilon_1 \varepsilon_2}} \,\Gamma_2(a_D+\mu;\varepsilon_1, \varepsilon_2)^{-1},
    \\
    &\mathscr{Z}_{\text{inst}}(a_D; \varepsilon_1, \varepsilon_2 | s) = 1 + \sum_{k=1}^{\infty} Q_{3k}(a_D) (\varepsilon_1 \varepsilon_2 s)^{-k}.
    \label{eq:Z-QPIII-3-irregular-anzats-3}
\end{align} 
\end{subequations}
where $s=-32 \, i \, t^{1/4}$, $\xi_1,\xi_2,\beta,\delta,\chi,\mu$ are complex parameters and $Q_{3k}(a_D)$ is a polynomial in $a_D$ of degree $3k$, and $\Gamma_2(x|\omega_1,\omega_2)$ is the Barnes double gamma function, which satisfies the difference equations
\begin{equation} \label{eq:difference-1-loop}
\Gamma_2(x+\omega_1|\omega_1,\omega_2)=\frac{\Gamma_2(x|\omega_1,\omega_2)}{\Gamma_1(x|\omega_2)},\quad \Gamma_2(x+\omega_2|\omega_1,\omega_2)=\frac{\Gamma_2(x|\omega_1,\omega_2)}{\Gamma_1(x|\omega_1)}.
\end{equation}
Here, $\Gamma_1(x|\omega)$ is the Barnes multiple gamma function of order 1 satisfying
\begin{equation}
 \Gamma_1(x+\omega|\omega)=x\Gamma_1(x|\omega), \quad 
 \Gamma_1(x|-\omega)=\frac{1}{\Gamma_1(x+\omega|\omega)}. 
\end{equation}

Substituting these expressions into the bilinear relations \eqref{eq:bilinear split}, we may observe that the parameters $\xi_1,\xi_2,\beta,\delta,\chi,\mu$ and the polynomials $Q_{3k}(a_D)$ are determined inductively. 
They are given by
\begin{equation}
\xi_1=0,\quad 
\xi_2=\frac{9 (\varepsilon_1+\varepsilon_2)^2}{8}+\frac{1}{4},\quad 
\beta=\frac{1}{256}, \quad \delta=\frac{1}{4},\quad 
\chi=0,\quad \mu=-\frac{\varepsilon_1+\varepsilon_2}{2},
\end{equation}
and
\begin{subequations}
\begin{align}
   & Q_3(a_D;\varepsilon_1,\varepsilon_2)
   =\frac{a_D \left(4 a_D^2+3 \varepsilon_1^2+8 \varepsilon_1 \varepsilon_2+3 \varepsilon_2^2\right)}{4}, 
    \\
   &  Q_6(a_D;\varepsilon_1,\varepsilon_2) \notag \\
   & \quad = -\frac{1}{32} 
   \Biggl[  -16 a_D^6  + \left( -24 \varepsilon_1^2 + 16 \varepsilon_1 \varepsilon_2 - 24 \varepsilon_2^2 \right) a_D^4 
\notag \\
& \qquad + \left( -9 \varepsilon_1^4 + 88 \varepsilon_1^3 \varepsilon_2 + 238 \varepsilon_1^2 \varepsilon_2^2 + 88 \varepsilon_1 \varepsilon_2^3 - 9 \varepsilon_2^4 \right) a_D^2
\notag \\ & \qquad 
+ \left( 9 \varepsilon_1^5 \varepsilon_2 + 48 \varepsilon_1^4 \varepsilon_2^2 + 78 \varepsilon_1^3 \varepsilon_2^3 + 48 \varepsilon_1^2 \varepsilon_2^4 + 9 \varepsilon_1 \varepsilon_2^5 \right) 
\Biggr],
\\
& Q_9(a_D;\varepsilon_1,\varepsilon_2) 
\notag \\
& \quad =\frac{a_D}{384} 
\Biggl[ 
64 a_D^8 + \left( 144 \varepsilon_1^2 - 576 \varepsilon_1 \varepsilon_2 + 144 \varepsilon_2^2 \right) a_D^6
\notag \\
& \qquad + \left( 108 \varepsilon_1^4 - 1776 \varepsilon_1^3 \varepsilon_2 - 552 \varepsilon_1^2 \varepsilon_2^2 - 1776 \varepsilon_1 \varepsilon_2^3 + 108 \varepsilon_2^4 \right) a_D^4 
\notag \\
&
\qquad + \left( 27 \varepsilon_1^6 - 1116 \varepsilon_1^5 \varepsilon_2 + 7057 \varepsilon_1^4 \varepsilon_2^2 + 18552 \varepsilon_1^3 \varepsilon_2^3 + 7057 \varepsilon_1^2 \varepsilon_2^4 - 1116 \varepsilon_1 \varepsilon_2^5 + 27 \varepsilon_2^6 \right) a_D^2 
\notag \\
&
\qquad + \left( - 81 \varepsilon_1^7 \varepsilon_2 + 2592 \varepsilon_1^6 \varepsilon_2^2 + 13425 \varepsilon_1^5 \varepsilon_2^3 + 21312 \varepsilon_1^4 \varepsilon_2^4 + 13425 \varepsilon_1^3 \varepsilon_2^5 + 2592 \varepsilon_1^2 \varepsilon_2^6 - 81 \varepsilon_1 \varepsilon_2^7 \right) 
\Biggr],
\\
& Q_{12}(a_D;\varepsilon_1,\varepsilon_2) 
\notag \\
& \quad =\frac{1}{6144} 
\Biggl[
256 a_D^{12} 
+ \left( 768 \varepsilon_1^2 - 5632 \varepsilon_1 \varepsilon_2 + 768 \varepsilon_2^2 \right) a_D^{10} 
\notag \\
& \qquad + \left( 864 \varepsilon_1^4 - 19968 \varepsilon_1^3 \varepsilon_2 + 33216 \varepsilon_1^2 \varepsilon_2^2 - 19968 \varepsilon_1 \varepsilon_2^3 + 864 \varepsilon_2^4 \right) a_D^8
\notag \\
& \qquad + \left( 432 \varepsilon_1^6 - 21312 \varepsilon_1^5 \varepsilon_2 + 210448 \varepsilon_1^4 \varepsilon_2^2 + 166656 \varepsilon_1^3 \varepsilon_2^3 + 210448 \varepsilon_1^2 \varepsilon_2^4 - 21312 \varepsilon_1 \varepsilon_2^5 + 432 \varepsilon_2^6 \right) a_D^6 
\notag \\
& \qquad + \left( 81 \varepsilon_1^8 - 7776 \varepsilon_1^7 \varepsilon_2 + 206052 \varepsilon_1^6 \varepsilon_2^2 - 830176 \varepsilon_1^5 \varepsilon_2^3 - 2213466 \varepsilon_1^4 \varepsilon_2^4 - 830176 \varepsilon_1^3 \varepsilon_2^5\right.
\notag \\
& \qquad\quad  \left.+ 206052 \varepsilon_1^2 \varepsilon_2^6 - 7776 \varepsilon_1 \varepsilon_2^7 + 81 \varepsilon_2^8 \right) a_D^4 
\notag \\
& \qquad 
 + \left( - 486 \varepsilon_1^9 \varepsilon_2 + 41040 \varepsilon_1^8 \varepsilon_2^2 - 808128 \varepsilon_1^7 \varepsilon_2^3 - 3941328 \varepsilon_1^6 \varepsilon_2^4 - 6105012 \varepsilon_1^5 \varepsilon_2^5 - 3941328 \varepsilon_1^4 \varepsilon_2^6
\right.
\notag \\
&\qquad\quad \left.- 808128 \varepsilon_1^3 \varepsilon_2^7 + 41040 \varepsilon_1^2 \varepsilon_2^8 - 486 \varepsilon_1 \varepsilon_2^9 \right) a_D^2 
\notag \\
& \qquad + \left( 243 \varepsilon_1^{10} \varepsilon_2^2 - 44064 \varepsilon_1^9 \varepsilon_2^3 - 313740 \varepsilon_1^8 \varepsilon_2^4 - 831456 \varepsilon_1^7 \varepsilon_2^5 - 1124046 \varepsilon_1^6 \varepsilon_2^6 - 831456 \varepsilon_1^5 \varepsilon_2^7
\right.
\notag \\
&\quad\qquad \left.- 313740 \varepsilon_1^4 \varepsilon_2^8 - 44064 \varepsilon_1^3 \varepsilon_2^9 + 243 \varepsilon_1^2 \varepsilon_2^{10} \right)\Biggr],
\\
& Q_{15}(a_D;\varepsilon_1,\varepsilon_2) 
\notag \\ 
& \quad =\frac{a_D}{122880} 
\Biggl[ 1024 a_D^{14} 
 + \left( 3840 \varepsilon_1^2 - 40960 \varepsilon_1 \varepsilon_2 + 3840 \varepsilon_2^2 \right) a_D^{12} 
\notag \\
& \qquad + \left( 5760 \varepsilon_1^4 - 171520 \varepsilon_1^3 \varepsilon_2 + 600320 \varepsilon_1^2 \varepsilon_2^2 - 171520 \varepsilon_1 \varepsilon_2^3 + 5760 \varepsilon_2^4 \right) a_D^{10}
\notag\\
& \qquad + \left( 4320 \varepsilon_1^6 - 253440 \varepsilon_1^5 \varepsilon_2 + 3337120 \varepsilon_1^4 \varepsilon_2^2 - 2296320 \varepsilon_1^3 \varepsilon_2^3
\right.
\notag \\
& \qquad \quad \left.+ 3337120 \varepsilon_1^2 \varepsilon_2^4 - 253440 \varepsilon_1 \varepsilon_2^5 + 4320 \varepsilon_2^6 \right) a_D^8
\notag \\
& \qquad + \bigl( 1620 \varepsilon_1^8 - 164160 \varepsilon_1^7 \varepsilon_2 + 4892880 \varepsilon_1^6 \varepsilon_2^2 - 34381120 \varepsilon_1^5 \varepsilon_2^3 - 36999432 \varepsilon_1^4 \varepsilon_2^4
\notag \\
& \qquad\quad - 34381120 \varepsilon_1^3 \varepsilon_2^5 + 4892880 \varepsilon_1^2 \varepsilon_2^6 - 164160 \varepsilon_1 \varepsilon_2^7 + 1620 \varepsilon_2^8 \bigr) a_D^6
\notag\\
& \qquad + \bigl( 243 \varepsilon_1^{10} - 43200 \varepsilon_1^9 \varepsilon_2 + 2537055 \varepsilon_1^8 \varepsilon_2^2 - 52037920 \varepsilon_1^7 \varepsilon_2^3 + 113231582 \varepsilon_1^6 \varepsilon_2^4 + 333908800 \varepsilon_1^5 \varepsilon_2^5 
\notag \\
& \qquad \quad  + 113231582 \varepsilon_1^4 \varepsilon_2^6 - 52037920 \varepsilon_1^3 \varepsilon_2^7 + 2537055 \varepsilon_1^2 \varepsilon_2^8 - 43200 \varepsilon_1 \varepsilon_2^9 + 243 \varepsilon_2^{10} \bigr) a_D^4
\notag\\
& \qquad + \bigl( - 2430 \varepsilon_1^{11} \varepsilon_2 + 374220 \varepsilon_1^{10} \varepsilon_2^2 - 19525950 \varepsilon_1^9 \varepsilon_2^3 + 275671824 \varepsilon_1^8 \varepsilon_2^4 + 1287899580 \varepsilon_1^7 \varepsilon_2^5 
\notag \\
&\quad\qquad  + 1960171080 \varepsilon_1^6 \varepsilon_2^6 + 1287899580 \varepsilon_1^5 \varepsilon_2^7 + 275671824 \varepsilon_1^4 \varepsilon_2^8 - 19525950 \varepsilon_1^3 \varepsilon_2^9
\notag \\
&\quad\qquad + 374220 \varepsilon_1^2 \varepsilon_2^{10} - 2430 \varepsilon_1 \varepsilon_2^{11} \bigr) a_D^2 
\notag \\
& \qquad + \bigl( 3645 \varepsilon_1^{12} \varepsilon_2^2 - 942840 \varepsilon_1^{11} \varepsilon_2^3 + 54505737 \varepsilon_1^{10} \varepsilon_2^4 + 363517920 \varepsilon_1^9 \varepsilon_2^5 + 918147258 \varepsilon_1^8 \varepsilon_2^6 
\notag \\
&\quad\qquad + 1221269040 \varepsilon_1^7 \varepsilon_2^7 
+ 918147258 \varepsilon_1^6 \varepsilon_2^8 + 363517920 \varepsilon_1^5 \varepsilon_2^9 + 54505737 \varepsilon_1^4 \varepsilon_2^{10} 
\notag \\
&\quad\qquad - 942840 \varepsilon_1^3 \varepsilon_2^{11} + 3645 \varepsilon_1^2 \varepsilon_2^{12} \bigr) \Biggr]. 
\end{align} 
\end{subequations}
Due to the factor $(\varepsilon_1 \varepsilon_2)^{-k}$ in \eqref{eq:Z-QPIII-3-irregular-anzats-3}, the series coefficients of ${\mathscr Z}$ have poles at $ \varepsilon_1, \varepsilon_2 = 0$, and their pole order increases as $k$ becomes larger. 
However, as observed in \cite{BST25}, taking the logarithm leads to nontrivial cancellations, and it is expected that the order of the poles becomes at most one:
\begin{align}
& - \varepsilon_1 \varepsilon_2 \log {\mathscr Z} \notag \\
& \quad = 
-\dfrac{s^2}{256}-\dfrac{a_Ds}{4}-\dfrac{1}{8}(4a_D^2+9\varepsilon_1^2+20\varepsilon_1\varepsilon_2+9\varepsilon_2^2)\log s 
- \dfrac{a_D(4a_D^2+3\left( \varepsilon_1 + \varepsilon_2 \right)^2 + 2\varepsilon_1\varepsilon_2)}{4s} 
\notag \\
& \qquad +
\dfrac{80a_D^4+136a_D^2\varepsilon^2+9\left( \varepsilon_1 + \varepsilon_2 \right)^4+48a_D^2\varepsilon_1\varepsilon_2+12(\varepsilon_1+\varepsilon_2)^2\varepsilon_1\varepsilon_2}{32s^2}
\notag\\
& \qquad -\frac{a_D}{48s^3}  \Biggl[
528 a_D^4 
+ \left( 1640 \varepsilon_1^2 + 3584 \varepsilon_1 \varepsilon_2 + 1640 \varepsilon_2^2 \right) a_D^2 
\notag\\
&
\quad\qquad +\left( 405 \varepsilon_1^4 + 1920 \varepsilon_1^3 \varepsilon_2 + 3006 \varepsilon_1^2 \varepsilon_2^2 + 1920 \varepsilon_1 \varepsilon_2^3 + 405 \varepsilon_2^4 \right) 
 \Biggr]
\notag\\
& \qquad + \frac{1}{32 s^4}  \Biggl[ 
2016 a_D^6 
+ 96 \left( 105 \varepsilon_1^2 + 218 \varepsilon_1 \varepsilon_2 + 105 \varepsilon_2^2 \right) a_D^4 
\notag\\
&\qquad\quad + \left( 5886 \varepsilon_1^4 + 25376 \varepsilon_1^3 \varepsilon_2 + 38564 \varepsilon_1^2 \varepsilon_2^2 + 25376 \varepsilon_1 \varepsilon_2^3 + 5886 \varepsilon_2^4 \right) a_D^2 
\notag\\
&\qquad\quad + 3 \left( \varepsilon_1 + \varepsilon_2 \right)^2 \left( 81 \varepsilon_1^4 + 402 \varepsilon_1^3 \varepsilon_2 + 602 \varepsilon_1^2 \varepsilon_2^2 + 402 \varepsilon_1 \varepsilon_2^3 + 81 \varepsilon_2^4 \right) 
 \Biggr]
\notag\\
&\qquad -\frac{a_D}{80 s^5} \Biggl[ 33728 a_D^6 
+ 112 \left( 2231 \varepsilon_1^2 + 4456 \varepsilon_1 \varepsilon_2 + 2231 \varepsilon_2^2 \right) a_D^4
\notag\\
&\qquad\quad + 4 \left( 69001 \varepsilon_1^4 + 277040 \varepsilon_1^3 \varepsilon_2 + 411950 \varepsilon_1^2 \varepsilon_2^2 + 277040 \varepsilon_1 \varepsilon_2^3 + 69001 \varepsilon_2^4 \right) a_D^2 
\notag\\
&\qquad\quad + 3 \left( 13869 \varepsilon_1^6 + 87080 \varepsilon_1^5 \varepsilon_2 + 216891 \varepsilon_1^4 \varepsilon_2^2 + 287600 \varepsilon_1^3 \varepsilon_2^3 + 216891 \varepsilon_1^2 \varepsilon_2^4 + 87080 \varepsilon_1 \varepsilon_2^5 + 13869 \varepsilon_2^6 \right) \Biggr]
\notag\\
&\qquad +O(s^{-6}), 
\label{eq:logZ-HIII-3-2}
\end{align}
where we have omitted the contribution from the $1$-loop part; 
see Remark \ref{rem:1-loop-part} below. 
As far as the computations in \cite{BST25} go, the expansion coefficients 
on the right-hand side do not involve $\varepsilon_1, \varepsilon_2$ in the denominator. 
Therefore we can expect that the following limit exists\footnote{
To compare with the accessory parameter \eqref{eq:final-AP-HIII-3-2}, 
we set $\varepsilon_1 = i \hbar$ instead of  $\varepsilon_1 = \hbar$.}:
\begin{align}
{\mathscr W} 
& =  - \varepsilon_1 \varepsilon_2 \log {\mathscr Z}  
\Bigl|_{(\varepsilon_1, \varepsilon_2) = (i \hbar, 0)} 
\notag \\
& = -\frac{s^2}{256}
- \frac{a_D}{4} \, s
- \frac{4 \, a_D^2 - 9 \, \hbar^2}{8}\,\log s
- \frac{4 \, a_D^3 - 3 \, \hbar^2 a_D }{4} \, s^{-1}
+ \frac{80 \, a_D^4 - 136 \, \hbar^2 a_D^2  + 9 \, \hbar^4}{32} \, s^{-2}
\notag \\ & \quad  
- \frac{528 \, a_D^5 - 1640 \, \hbar^2 a_D^3  + 405 \,\hbar^4 a_D }{48} \, s^{-3}
+ \frac{9\,\left(224 \, a_D^6 - 1120\, \hbar^2 a_D^4  + 654 \, \hbar^4 a_D^2  - 27\, \hbar^6\right)}{32} \, s^{-4}
\notag \\ & \quad  
- \frac{33728 \, a_D^7 - 249872 \, \hbar^2  a_D^5 + 276004 \, \hbar^4 a_D^3  - 41607 \, \hbar^6 a_D}{80} \, s^{-5}
+ O(s^{-6}).
\label{eq:CCB-HIII-3-2}
\end{align}
This is the (candidate of the) classical conformal block.  
Then, the Zamolodchikov-type conjecture in the presence of irregular singularities is formulated as follows.

\begin{conj}[{Cf., \cite{GMS20}}]
The classical conformal block \eqref{eq:CCB-HIII-3-2} exists. 
Moreover, under the identification 
\begin{equation}
s = -32 \, i \, t^{1/4}, \quad 
a_D = i \, \nu,  
\end{equation}
the accessory parameter \eqref{eq:final-AP-HIII-3-2}
and the classical conformal block \eqref{eq:CCB-HIII-3-2}
are related as follows:
\begin{equation} \label{eq:conj-H-III3-2}
{\mathscr E} = t\frac{d}{dt} {\mathscr W}. 
\end{equation} 
\end{conj}

From the above computations, we can confirm that the equality \eqref{eq:conj-H-III3-2} holds up to $t^{-5/4}$. 

\begin{rem} \label{rem:GMS20}
In \cite[(5.36) and (5.40)]{GMS20}, Gavrylenko--Marshakov--Stoyan refer to the Zamolodchikov-type conjecture. 
While we analyze the Heun equation directly, \cite{GMS20} instead proceed via the Painlev\'e ${\rm III}_3$ equation and further employ the blow-up equation to determine the accessory parameter (i.e., the eigenvalue of the Mathieu equation). 
A rigorous comparison of these two approaches remains a task for future work.
\end{rem}

\smallskip
\begin{rem} \label{rem:1-loop-part}
Since the derivative with respect to $t$ will be considered in the main proposal \eqref{eq:conj-H-III3-2} in the conjecture, we omit the contribution from the 1-loop part ${\mathscr Z}_{\text{1-loop}}$ in \eqref{eq:logZ-HIII-3-2}--\eqref{eq:CCB-HIII-3-2}. 
We will adopt the same convention in the examples below, but note that, when deriving such an expansion from the bilinear equation \eqref{eq:bilinear split}, the difference relation \eqref{eq:difference-1-loop} for the 1-loop part plays a crucial role.
\end{rem}

\smallskip
\subsubsection{Expansion of the classical conformal block 
as $t \to 0$}
\label{subsec:example-CCB-III-regular}

As in the previous case, we show that the expansion of ${\mathscr Z}$ as $t \to 0$ can also be derived from the bilinear relation \eqref{eq:bilinear split}, similarly to the results in \cite{BS15} and \cite[Appendix C]{BST25}. 
The resulting expansion of $\mathscr{Z}(a; \varepsilon_1, \varepsilon_2 | t)$ corresponds to the full $\mathcal{N}=2$, $D=4$ supersymmetric $SU(2)$ partition function without hypermultiplets, which is expressed in terms of the pure Nekrasov partition function. 

In this case we impose the following factorization property
\begin{equation*}
    \mathscr{Z}(a; \varepsilon_1, \varepsilon_2 | t) = \mathscr{Z}_{\text{cl}}(a; \varepsilon_1, \varepsilon_2 | t) \,\mathscr{Z}_{\text{1-loop}}(a; \varepsilon_1, \varepsilon_2) \,\mathscr{Z}_{\text{inst}}(a; \varepsilon_1, \varepsilon_2 | t)
\end{equation*}
and suppose that each factor is expressed as 
\begin{align*}
    &\mathscr{Z}_{\text{cl}}(a; \varepsilon_1, \varepsilon_2 | t) = t^{\frac{(\varepsilon_1+\varepsilon_2)^2/4 - a^2}{\varepsilon_1 \varepsilon_2}} ,
    \\
    &\mathscr{Z}_{\text{1-loop}}(a; \varepsilon_1, \varepsilon_2) = \Gamma_2(2a|\varepsilon_1, \varepsilon_2) \, \Gamma_2(-2a|\varepsilon_1, \varepsilon_2),
    \\
    &\mathscr{Z}_{\text{inst}}(a; \varepsilon_1, \varepsilon_2 | t) = 1 + \sum_{k=1}^{\infty} P_{k}(a;\varepsilon_1,\varepsilon_2) (\varepsilon_1 \varepsilon_2 t)^{k}. 
\end{align*}
It follows from the bilinear equation \eqref{eq:bilinear split} that the coefficients $P_k$ are expressed for example as 
\begin{align}
     &P_1(a; \varepsilon_1, \varepsilon_2) = - \frac{2}{{\varepsilon_1^2 \varepsilon_2^2} \, (4 \, a^2-(\varepsilon_1+\varepsilon_2)^2)}, 
     \\
     &P_2(a; \varepsilon_1, \varepsilon_2) = 
     - \frac{8 \, (\varepsilon_1+\varepsilon_2)^2 + \varepsilon_1 \varepsilon_2 - 8 \, a^2}{{\varepsilon_1^4 \varepsilon_2^4} \,
     \left( 4 \, a^2-(\varepsilon_1+\varepsilon_2)^2 \right)
     \left(4 \, a^2-(2\varepsilon_1+\varepsilon_2)^2 \right)
     \left(4 \, a^2- (\varepsilon_1+2\varepsilon_2)^2 \right) 
     }, \quad \dots
\end{align}
Then, we may observe that   
\begin{align}
    & - \varepsilon_1\varepsilon_2\log \mathscr{Z} = 
    \left( a^2 - \frac{(\varepsilon_1+\varepsilon_2)^2}{4} \right) \,\log t + 
     \frac{2}{4 \, a^2 - (\varepsilon_1+\varepsilon_2)^2} \, t 
    \notag \\ 
    & \quad\qquad
    +\frac{20 \, a^2 + 7 \, \varepsilon_1^2 + 16 \, \varepsilon_1\varepsilon_2 + 7 \, \varepsilon_2^2}{
    \left(4 \, a^2-(\varepsilon_1+\varepsilon_2)^2\right)^2 
    \left(4 \, a^2- (2\varepsilon_1+\varepsilon_2)^2\right) 
    \left(4 \, a^2 - (\varepsilon_1+2\varepsilon_2)^2\right)}t^2 
    + O(t^3). 
\end{align}
has coefficients which are expected to be holomorphic at $\varepsilon_1, \varepsilon_2 = 0$.
Therefore, we can expect that the following limit exists:
\begin{align}
{\mathscr W} 
& =  - \varepsilon_1 \varepsilon_2 \log {\mathscr Z}  
\Bigl|_{(\varepsilon_1, \varepsilon_2) = (\hbar, 0)} 
= \left( a^2 - \frac{\hbar^2}{4} \right) \log t 
+ \frac{2t}{4\,a^2 - \hbar^2}
+ \frac{(20\,a^2 + 7\,\hbar^2)t^2}{4\,(4\,a^2 - \hbar^2)^3(a^2 - \hbar^2)} + O(t^3).
\label{eq:CCB-HIII-3-1}
\end{align}
This is the NS twisted superpotential with $N_f = 0$ \cite{NS09}. 
Through the AGT correspondence, 
this is expected to agree with\footnote{ \label{foot:virasoro-hbar}
We have introduced $\hbar$ by rescalings:  
$\delta_\sigma \mapsto \delta_\sigma/\hbar^2$ and $\widetilde{\mathscr W} \mapsto \widetilde{\mathscr W}/\hbar^2$. 
When considering the classical conformal block constructed using the Virasoro algebra in \S \ref{section:examples}, we introduce $\hbar$ in a similar manner.
} 
\begin{align}\label{eq:CCB-HIII-3-1-Virasoro}
    \widetilde{\mathscr{W}} 
    = \delta_\sigma \, \log t + \dfrac{t}{2 \, \delta_\sigma}+\dfrac{5 \, \delta_\sigma - 3 \, \hbar^2}{16 \, \delta_\sigma^3 \, (4\,\delta_\sigma + 3 \, \hbar^2)} \, t^2+O(t^3), 
\end{align}
which is the classical limit of the irregular conformal block with a regular singular expansion, given in \cite{Gaiotto09} in terms of the Virasoro algebra; see  (\cite[Section 2.2.3, (2.20c)]{LN21}). 
More precisely, the equality ${\mathscr W}  = - \widetilde{\mathscr W}$
is expected to hold under the identification $\delta_\sigma = - a^2 + \hbar^2/4$.

In this case, the Zamolodchikov-type conjecture has already been stated in \cite{MM09, PP14}, but for completeness, we would like to state it here as well.

\begin{conj}[Cf.\,\cite{MM09, PP14}]
The classical conformal blocks \eqref{eq:CCB-HIII-3-1}--\eqref{eq:CCB-HIII-3-1-Virasoro} exist. 
Moreover, under the identification 
\begin{equation}
a = \nu \quad, \quad \delta_\sigma = - \nu^2 + \frac{\hbar^2}{4}, 
\end{equation}
the accessory parameter \eqref{eq:final-AP-HIII-3-1}
and the classical conformal blocks \eqref{eq:CCB-HIII-3-1}--\eqref{eq:CCB-HIII-3-1-Virasoro}
are related as follows:
\begin{equation} \label{eq:conj-HIII-3-1}
\mathscr{E} = - t\dfrac{d}{dt}\mathscr{W} = t\dfrac{d}{dt} \widetilde{\mathscr{W}}. 
\end{equation} 
\end{conj}

Since the expressions become lengthy, we do not list 
the coefficients in \eqref{eq:final-AP-HIII-3-1} and \eqref{eq:CCB-HIII-3-1} beyond this point.  
However, using Mathematica we have computed several further terms, 
and the agreement of the coefficients in \eqref{eq:conj-HIII-3-1}
has been confirmed up to slightly higher orders.

\section{Examples}
\label{section:examples}

In the previous section, taking the equation 
$H_{{\rm III}_3}$ as an example, we explained how the accessory parameter can be computed from the Voros period and described its relation to the classical conformal block. 
In this section, we show that an analogous relationship is expected to hold for all confluent types of the Heun equation. 
More specifically, for each Heun equation, we derive formal series expansions of various accessory parameters by applying the method developed in the previous section, and formulate, in the form of conjectures, their relationship with the series expansions of the classical limits of the conformal blocks obtained in \cite{BST25}. 
Since the expansions of the accessory parameters are derived in exactly the same way as in Theorems \ref{thm:unique-existence} and \ref{thm:unique-existence-2} of the previous section, we do not repeat the proof for each example here.
In some cases (e.g., Conjectures \ref{conj:H-V-2} and \ref{conj:H-IV-1} etc.), the Zamolodchikov-type conjecture has already been stated in previous works; however, for completeness, we include them here as well.

\bigskip
\subsection{Heun equation $H_{\rm VI}$}

The equation $H_{{\rm VI}}$ has the following potential
\begin{equation}
Q_{{\rm VI}} = 
\frac{\theta_0^2 - \frac{\hbar^2}{4}}{x^2} 
+ \frac{\theta_1^2 - \frac{\hbar^2}{4}}{(x-1)^2} + \frac{\theta_t^2 - \frac{\hbar^2}{4}}{(x-t)^2}
+ \frac{\theta_\infty^2 - \theta_0^2 - \theta_1^2 - \theta_t^2 + \frac{\hbar^2}{2}}{x(x-1)} 
- \frac{{\mathscr E}}{x(x-1)(x-t)}.
\end{equation}
In the following, we derive the expansion of 
${\mathscr E}$ in the limit $t \to 0$. 
Note that the quadratic differential
$Q_{\rm VI}(x) \, dx^2$ is invariant under 
the following changes of parameters:
\begin{itemize}
    \item 
    $x \mapsto 1 - x$,\, 
    $t \mapsto 1 -t$,\, 
    ${\mathscr E} \mapsto - {\mathscr E}$, \, 
    $\theta_0 \mapsto \theta_1$, \, 
    $\theta_1 \mapsto \theta_0$.
    \item 
    $x \mapsto \frac{1}{x}$, \, 
    $t \mapsto \frac{1}{t}$, \, 
    ${\mathscr E} \mapsto  \frac{{\mathscr E}}{t} - 2 \left( \theta_t^2 - \frac{\hbar^2}{4}\right)\frac{t-1}{t}$, \,
    $\theta_0 \mapsto \theta_\infty$, \, 
    $\theta_\infty \mapsto \theta_0$.
\end{itemize}
As a consequence of these symmetries, 
the behavior of ${\mathscr E}$ in the limits 
$t \to 1$ and $t \to \infty$ can be obtained from the 
behavior as $t \to 0$, and we therefore omit it here.

\bigskip
\noindent
{\bf Series expansion of ${\mathscr E}$.} \, 
In this case, we do not have to rescale the variables 
(i.e., $(d_{x}, d_{y}, d_{\mathscr E}) = (0, 0, 0)$) 
and take $\Lambda = t$.
The spectral curve becomes
\[
Q_{0}^{\rm res} =\frac{\theta_0^2}{X^2} 
+ \frac{\theta_1^2}{(X-1)^2} + \frac{\theta_t^2}{(X-\Lambda)^2}
+ \frac{\theta_\infty^2 - \theta_0^2 - \theta_1^2 - \theta_t^2}{X(X-1)} 
- \frac{{\mathscr G}_0}{X(X-1)(X-\Lambda)} 
\]
\begin{equation}
\quad\qquad \xrightarrow{ \Lambda \to 0 } \quad 
Q_0^{[0]} = \frac{\theta_0^2}{X^2} 
+ \frac{\theta_1^2}{(X-1)^2} + \frac{\theta_t^2}{X^2}
+ \frac{\theta_\infty^2 - \theta_0^2 - \theta_1^2 - \theta_t^2}{X(X-1)} 
- \frac{{\mathscr G}^{[0]}_0}{X^2(X-1)} . 
\end{equation}
In this limit, two of the four zeros of 
$Q^{\rm res}_0$ coalesce at $X=0$, 
and the limiting curve ${\mathcal C}^{\rm res}_{\rm deg}$ is a genus $0$ curve (which is called the Gauss (hypergeometric) curve in \cite{IKT18-2}). 
Therefore, for sufficiently small $\Lambda$, 
we choose the cycle $\gamma$ to be a circle that encloses the two turning points which tend to $X=0$
together with the origin inside, 
while keeping the other zeros and poles outside. 
Then, the elliptic integrals defining the Voros period are described 
by the residues at $X = 0$ of the coefficients in the small $\Lambda$-expansion.

When imposing the condition 
\begin{equation}
V_{-1} = 2 \pi i \nu, \quad 
V_{0} = \pi i, \quad 
V_m = 0 \quad (m \ge 1)
\end{equation} 
for the Voros period associated with $\gamma$
(which implies \eqref{eq:partial-IMD} with $-$-sign), 
we find that, under the assumption $\nu \ne 0$,  
the first few ${\mathscr G}_{m}$ 
obtained by the method described in Section \ref{subsec:AP-algorithm}, 
admit the following expansions:
\begin{subequations} \label{eq:AP-HVI}
\begin{align} 
{\mathscr G}_0 & = (\nu^2- \theta_0^2 - \theta_t^2)
 -\frac{
  \nu^4 - (\theta_0^2+\theta_1^2 - 3\theta_t^2 -\theta_\infty^2) \, \nu^2 + (\theta_0^2-\theta_t^2)(\theta_1^2-\theta_\infty^2)}{2 \, \nu^2} \, \Lambda 
 \notag \\ 
 & \quad 
 - \frac{
3\,\nu^8
- 2\,A \,\nu^6
+ B \,\nu^4
+ 6 \,C\, \nu^2
-5 \,D
}{32 \, \nu^6} \, \Lambda^2 + O(\Lambda^3), \\
{\mathscr G}_1 & =  \frac{1}{4} - 
\frac{3 \, \nu^4+(\theta_0^2-\theta_t^2)(\theta_1^2-\theta_\infty^2) 
}{8\,\nu^4} \, \Lambda
- \frac{
\nu^8
+ A \,\nu^6
+ 2 \,B\, \nu^4
+ 15 \,C\, \nu^2
-21 \,D
}{64 \nu^8} \, \Lambda^2
+O(\Lambda^3), \\
{\mathscr G}_2 & = -\frac{(\theta_0^2 - \theta_t^2)(\theta_1^2 - \theta_\infty^2)}
{32 \, \nu^6} \, \Lambda 
+ \frac{
\nu^8
- 8 \,A\, \nu^6
- 16 \,B\, \nu^4
- 126 \,C\, \nu^2
+ 219 \,D
}{512 \, \nu^{10}} \, \Lambda^{2}  +O(\Lambda^3), \dots,
\end{align}
\end{subequations}
where
\begin{subequations}
\begin{align}
A & = \theta_0^2+\theta_1^2+\theta_t^2+\theta_\infty^2 \\
B & = - (\theta_0^2+\theta_1^2+\theta_t^2+\theta_\infty^2)^2
  +2 (\theta_0^2 - \theta_1^2)(\theta_t^2 - \theta_\infty^2)
+ 2 (\theta_0^2 - \theta_\infty^2)(\theta_t^2 - \theta_1^2) \\  
C & = (\theta_0^2-\theta_t^2)^2(\theta_1^2+\theta_\infty^2)
  + (\theta_0^2+\theta_t^2)(\theta_1^2-\theta_\infty^2)^2, \\
D & = (\theta_0^2-\theta_t^2)^2(\theta_1^2-\theta_\infty^2)^2.
\end{align}
\end{subequations}
Furthermore, by formally interchanging the order of $\hbar$-expansion and small $t$-expansion, we obtain the following formal series expansion for the accessory parameter:
\begin{align}
{\mathscr E} 
& = \left(\nu^2 - \theta_0^2 - \theta_t^2 +  \frac{\hbar^2}{4}   \right) 
- 
\left(\frac{\nu^4  - (\theta_0^2+\theta_1^2 - 3\theta_t^2 -\theta_\infty^2) \, \nu^2 
+ (\theta_0^2-\theta_t^2)(\theta_1^2-\theta_\infty^2)}{2 \, \nu^2} 
+ O(\hbar^2)
\right) \, t + O(t^3)
\notag \\
& = \left(\nu^2 - \theta_0^2 - \theta_t^2 +  \frac{\hbar^2}{4}   \right) 
- \frac{ P_1  }{ 8 (4 \, \nu^2 - \hbar^2) } \, t 
+ \frac{ P_2 }{512 \, (4 \, \nu^2 - \hbar^2)^3(\nu^2 - \hbar^2)} \, t^2 + O(t^3),
\label{eq:final-AP-HVI}
\end{align}
where 
\begin{subequations}
\begin{align}
P_1 & = 16 \, \Bigl( 
\nu^4 -  (\theta_0^2 + \theta_1^2 + 3 \theta_t^2 - \theta_\infty^2)\,\nu^2 + (\theta_0^2 - \theta_t^2)(\theta_1^2 - \theta_\infty^2) \Bigr) 
\notag \\
& \quad 
+ 4 \,  \left(
 2 \nu^2 + \theta_0^2 + \theta_1^2 + 3 \theta_t^2 - \theta_\infty^2
\right) \, \hbar^2 - 3 \, \hbar^4, 
\\
P_2 & = 
- 3072\, \nu^{10}+ 2048\, A\, \nu^8 - 1024\, B\, \nu^6 - 6144\, C\, \nu^4 + 5120\, D\, \nu^2  
\notag \\
& \quad 
+ \bigl(
4864\, \nu^8
- 4096\, A\, \nu^6 
+ 768\, B\, \nu^4 
+ 3072\, C\, \nu^2  
+ 1792\, D 
\bigr) \,\hbar^2
\notag \\
& \quad 
- \bigl(
1920\, \nu^6
- 2304\, A\, \nu^4 
+ 192\, B\, \nu^2
+ 384\, C  
\bigr) \,\hbar^4
\notag \\
& \quad 
+ \bigl(96\, \nu^4 + 512\, A\, \nu^2 +16\, B \bigr) \, \hbar^6
+ \bigl(68\, \nu^2 + 40\, A \bigr) \, \hbar^8
- 9\, \hbar^{10}.
\end{align}
\end{subequations}
The rational expression of the coefficients 
in the last equality in \eqref{eq:final-AP-HVI}
is obtained by the same method as described in Remark \ref{rem:rational-expression}.
These coefficients can, in principle, be computed explicitly to arbitrary order.

\smallskip
\noindent
{\bf Comparison with ${\mathscr W}$.} \, 
The relevant conformal block has been considered since the time when the Zamolodchikov conjecture \cite{Zam86} was proposed.
It is defined in terms of the Virasoro algebra, and admits a regular singular expansion. 
It is also identified with the Nekrasov partition function with $N_f = 4$ through the AGT correspondence, and the relevant classical conformal block (the NS twisted superpotential) is explicitly given as follows (cf.\,\,\cite[Section 2.1, (2.6)]{LN21}):
\begin{align}\label{eq:CCB-HVI}
    \mathscr{W} & =
    (\delta_{\sigma}-\delta_0-\delta_t) \, \log t
- \frac{(\delta_{\sigma}-\delta_0+\delta_t)(\delta_{\sigma}-\delta_\infty+\delta_1)}{2 \, \delta_{\sigma}}\, t 
\notag \\
& \quad + \Biggl(
\frac{(\delta_{\sigma}-\delta_0+\delta_t)^2(\delta_{\sigma}-\delta_\infty+\delta_1)^2}{8 \, \delta_{\sigma}^2} 
\biggl(
\frac{1}{\delta_{\sigma}-\delta_0+\delta_t}+\frac{1}{\delta_{\sigma}-\delta_\infty+\delta_1}-\frac{1}{2 \, \delta_\sigma}
\biggr)
\notag \\
& \qquad 
+ \frac{\bigl( \delta_{\sigma}^2 + 2\,\delta_{\sigma}(\delta_0+\delta_t)
-3\, (\delta_0-\delta_t)^2\bigr)
\bigl(\delta_{\sigma}^2 + 2\,\delta_{\sigma}(\delta_\infty+\delta_1)
-3 \, (\delta_\infty-\delta_1)^2\bigr)}
{16 \, \delta_{\sigma}^2 \, (4 \, \delta_{\sigma}+3 \, \hbar^2)}
\Biggr) \, t^2
+O(t^3).
\end{align}
Here, we have introduced $\hbar$ appropriately by rescalings; see footnote \ref{foot:virasoro-hbar}.
Then, the original conjecture by Zamolodchikov can be formulated as follows:
\begin{conj}[{\cite{Zam86}}] \label{conj:H-VI}
The classical conformal block \eqref{eq:CCB-HVI} exists. 
Moreover, under the identification 
\begin{equation}
\delta_\sigma=-\nu^2 + \dfrac{h^2}{4},
\quad
\delta_{0}=-\theta_0^2 + \dfrac{\hbar^2}{4}, \quad
\delta_{t}=-\theta_t^2 + \dfrac{\hbar^2}{4}, \quad
\delta_{1}=-\theta_1^2 + \dfrac{\hbar^2}{4}, \quad
\delta_{\infty}=-\theta_\infty^2 + \dfrac{\hbar^2}{4},
\end{equation}
the accessory parameter \eqref{eq:final-AP-HVI}
and the classical conformal block \eqref{eq:CCB-HVI}
are related as follows:
\begin{equation} \label{eq:conj-H-VI}
\mathscr{E} = t(t-1)\dfrac{d}{dt}\mathscr{W}. 
\end{equation} 
\end{conj}

From the above computational results, we can observe that the equality \eqref{eq:conj-H-VI} holds up to $O(t^{-3})$.
The agreement of the coefficients in \eqref{eq:conj-H-VI} has been confirmed up to slightly higher orders; 
however, those coefficients are too lengthy to include here and are therefore omitted.

\bigskip
\subsection{Confluent Heun equation $H_{\rm V}$}
\label{section:H-V}

The equation $H_{{\rm V}}$ has the following potential
\begin{equation} 
Q_{{\rm V}} = \frac{\theta_0^2 - \frac{\hbar^2}{4}}{x^2}
+ \frac{\theta_t^2 - \frac{\hbar^2}{4}}{(x-t)^2} 
+ \frac{1}{4} + \frac{\theta_\infty}{x} - \frac{{\mathscr E}}{x(x-t)}.
\end{equation}

\subsubsection{Expansion at $t \to 0$} ~ 

\noindent
{\bf Series expansion of ${\mathscr E}$.} \, 
We perform the rescaling with 
$(d_{x}, d_{y}, d_{\mathscr E}) = (0, 0, 0)$ 
and take $\Lambda = t$.
The rescaled spectral curve becomes
\begin{equation}
Q_{0}^{\rm res} = 
\frac{\theta_0^2}{X^2}
+ \frac{\theta_t^2}{(X-\Lambda)^2}
+ \frac{1}{4}
+ \frac{\theta_\infty}{X}
- \frac{\mathscr{G}_0}{X(X-\Lambda)}
\quad \xrightarrow{ \Lambda \to 0 } \quad 
Q_0^{[0]} = 
\frac{\theta_0^2}{X^2}
+ \frac{\theta_t^2}{X^2}
+ \frac{1}{4}
+ \frac{\theta_\infty}{X}
- \frac{\mathscr{G}_0^{[0]}}{X^2}.
\end{equation}
In this limit, it is readily seen that one of the two zeros of 
$Q^{\rm res}_0$ and the double pole $\Lambda$ tends to the origin, 
and the limiting curve ${\mathcal C}^{\rm res}_{\rm deg}$ is a genus $0$ curve (which is called the Kummer curve in \cite{IKT18-2}). 
Therefore, for sufficiently small $\Lambda$, 
we choose the cycle $\gamma$ to be a circle that encloses the two turning points 
which tend to  
the zeros of $Q_{0}^{[0]}$ together with the origin and $\Lambda$ inside, 
while keeping the other turning point outside. 
Then, the elliptic integrals defining the Voros period are described 
by the residues at $X = 0$ of the coefficients in the small $\Lambda$-expansion.

When imposing the condition 
\begin{equation}
V_{-1} = 2 \pi i \nu, \quad 
V_{0} = \pi i , \quad 
V_m = 0 \quad (m \ge 1)
\end{equation} 
for the Voros period associated with $\gamma$
(which implies \eqref{eq:partial-IMD} with $-$-sign), 
under the assumption $\nu \ne 0$, 
we find 
\begin{subequations} \label{eq:AP-HV-1}
\begin{align} 
{\mathscr G}_0 & =  (- \nu^2+ \theta_0^2 + \theta_t^2)
 - \frac{\theta_{\infty}\left(\nu^2 - \theta_0^2 + \theta_t^2 \right)}{2\,\nu^2} \, \Lambda 
 - \frac{\nu^6 + A \,\nu^4 -3B \,\nu^2+ 5\,C}{32 \, \nu^6} \, \Lambda^2 + O(\Lambda^3),
\\
{\mathscr G}_1 & = - \frac{1}{4} 
+ \frac{\theta_{\infty}(\theta_0^2 - \theta_t^2)}{8\,\nu^4} \, \Lambda
+ \frac{\nu^6 -4A \,\nu^4 +15B \,\nu^2  - 42\,C}{128 \,\nu^8} \, \Lambda^2
+ O(\Lambda^3),
%
%
\\
{\mathscr G}_2 & = 
\frac{\theta_{\infty}\bigl(\theta_0^2-\theta_t^2\bigr)}{32\,\nu^6} \Lambda 
+ \frac{4 \, \nu^6 - 16 A\, \nu^4 +63B \,\nu^2 - 219\,C}{512 \,\nu^{10}} \,  \Lambda^2 
+ O(\Lambda^3), 
\end{align}
\end{subequations}
where
\begin{equation}
A  = 2\theta_0^2+2\theta_t^2 + \theta_{\infty}^2, \quad
B  =  (\theta_0^2-\theta_t^2)^2
+2\theta_{\infty}^2(\theta_0^2+\theta_t^2), \quad
C  = \theta_{\infty}^2(\theta_0^2-\theta_t^2)^2.
\end{equation}
Furthermore, by formally interchanging the order of $\hbar$-expansion 
and small $t$-expansion, 
we obtain the following formal series expansion for the accessory parameter:
\begin{align}
{\mathscr E} 
& = \left(- \nu^2 + \theta_0^2 + \theta_t^2 + \frac{\hbar^2}{4}   \right) 
+  \left( - \frac{\theta_\infty( \nu^2-\theta_0^2 + \theta_t^2)}{2 \, \nu^2} + \frac{\theta_\infty(\theta_0^2 - \theta_t^2)}{8 \, \nu^4} \,\hbar^2 
+ O(\hbar^4) \right) \, t + O(t^2) 
\notag \\
& = \left( - \nu^2 + \theta_0^2 + \theta_t^2 + \frac{\hbar^2}{4}  \right) + 
\frac{\theta_\infty\left(- 4( \nu^2 -\theta_0^2 + \theta_t^2  ) + \hbar^2 \right)}{2\left(4 \, \nu^2 - \hbar^2\right)} \, t + O(t^2).
\label{eq:final-AP-HV-1}
\end{align}

\bigskip
\noindent
{\bf Comparison with ${\mathscr W}$.} 
The relevant conformal block, with a regular singular expansion, was introduced by \cite{Gaiotto09} in terms of the Virasoro algebra.
It is also identified with the Nekrasov partition function with $N_f = 3$ through the AGT correspondence. 
The associated classical conformal block is expected to coincide with the NS twisted superpotential, and its expansion coefficients are explicitly given as follows (cf.\, \cite[Section 2.2.2, (2.14)]{LN21}):
\begin{align}\label{eq:CCB-HV-1}
    \mathscr{W} & =
    (\delta_{\sigma} - \delta_0-\delta_t) \, \log t
- \frac{(\delta_{\sigma}-\delta_0+\delta_t) \, \theta_{\infty}}{2 \, \delta_{\sigma}}\, t 
\notag \\
& \quad + \Biggl( 
\frac{(\delta_{\sigma}^2 - (\delta_0-\delta_t)^2)\theta_{\infty}^2}{16 \, \delta_{\sigma}^3}
- \frac{(3\, \theta_{\infty}^2+\delta_{\sigma})
\bigl(\delta_{\sigma}^2+2 \, \delta_{\sigma}(\delta_0+\delta_t)
-3 \, (\delta_0-\delta_t)^2\bigr)}
{16 \, \delta_{\sigma}^2(4\, \delta_{\sigma} + 3 \, \hbar^2)}
\Biggr) \, t^2
+O(t^3).
\end{align}

\begin{conj} 
\label{conj:H-V-1}
The classical conformal block \eqref{eq:CCB-HV-1} exists. 
Moreover, under the identification 
\begin{equation}
\delta_\sigma=-\nu^2 + \dfrac{\hbar^2}{4},
\quad
\delta_{0}=-\theta_0^2 + \dfrac{\hbar^2}{4}, \quad
\delta_{t}=-\theta_t^2 + \dfrac{\hbar^2}{4}, 
\end{equation}
the accessory parameter \eqref{eq:final-AP-HV-1} 
and the classical conformal block \eqref{eq:CCB-HV-1}
are related as follows:
\begin{equation} \label{eq:H-V-1}
\mathscr{E} = t\dfrac{d}{dt}\mathscr{W}.
\end{equation} 
\end{conj}

From the above computational results, we can observe that the equality \eqref{eq:H-V-1} holds up to $O(t^{-3})$.


\bigskip
\subsubsection{Type 1 expansion at $t \to \infty$}
~\\[-1.em] 

\smallskip
\noindent
{\bf Series expansion of ${\mathscr E}$.} \, 
Before going into details, we note that most of the results in this subsection 
are essentially contained in \cite[\S 3.2]{LN21} by Lisovyy--Naidiuk.
The new contribution we have added concerns the comparison between 
the existing classical conformal block and \cite{BST25}.

We perform the rescaling with 
$(d_{x}, d_{y}, d_{\mathscr E}) = (0, 0, 1)$ 
and take $\Lambda = t^{-1}$.
The rescaled spectral curve becomes
\begin{equation}
Q_{0}^{\rm res} = 
\frac{\theta_0^2}{X^2}
+ \frac{\theta_t^2}{(X-\Lambda^{-1})^2}
+ \frac{1}{4}
+ \frac{\theta_\infty}{X}
- \frac{\Lambda^{-1} \mathscr{G}_0}{X(X-\Lambda^{-1})}
\quad \xrightarrow{ \Lambda \to 0 } \quad 
Q_0^{[0]} = 
\frac{\theta_0^2}{X^2}
+ \frac{1}{4}
+ \frac{\theta_\infty}{X}
+ \frac{\mathscr{G}_0^{[0]}}{X}.
\end{equation}
In this limit, it is readily seen that one of the two zeros of $Q^{\rm res}_0$ and the double pole $\Lambda^{-1}$ tend to $\infty$, and the limiting curve ${\mathcal C}^{\rm res}_{\rm deg}$ is a genus $0$ curve (which is called the Kummer curve in \cite{IKT18-2}).  
Therefore, for sufficiently small $\Lambda$, 
we choose the cycle $\gamma$ to be a circle that encloses the two turning points which tend to the zeros of $Q_{0}^{[0]}$ together with the origin inside, 
while keeping the turning points that merge with infinity and $\Lambda^{-1}$ outside.
Then, the elliptic integrals defining the Voros period are described 
by the residues at $X = \infty$ of the coefficients in the small $\Lambda$-expansion.

When imposing the condition
\begin{equation}
V_{-1} = 2 \pi i \nu, \quad 
V_m = 0 \quad (m \ge 0)
\end{equation} 
for the Voros period associated with $\gamma$
(which implies \eqref{eq:partial-IMD} with $+$-sign), 
the resulting computation of the accessory parameter is as follows:
\begin{subequations} \label{eq:AP-HV-3}
\begin{align} 
{\mathscr G}_0 & = 
(\nu-\theta_{\infty})
+ (2 \, \nu^2-2 \, \theta_{\infty}\nu) \, \Lambda
+ \Bigl(
-4 \, \nu^3
+ 6 \, \theta_{\infty}\nu^2
+ 2 \, (\theta_t^2-\theta_{\infty}^2)\nu
+ 2 \, \theta_0^2\nu
- 2 \, \theta_0^2\theta_{\infty}
\Bigr) \, \Lambda^2
\notag \\
&
\quad  + \Bigl(
20 \, \nu^4
- 40 \, \theta_{\infty}\nu^3
+ (24 \, \theta_{\infty}^2-12 \, \theta_t^2-12 \, \theta_0^2) \, \nu^2
+ (4 \, \theta_{\infty}\theta_t^2+20 \, \theta_0^2\theta_{\infty})\nu
- 4 \, \theta_0^2(2\theta_{\infty}^2-\theta_t^2)
\Bigr) \, \Lambda^3 
\notag \\
& \quad 
+O(\Lambda^4), \\
{\mathscr G}_1 & = 
\left(-\nu + \frac{\theta_{\infty}}{2}\right) \, \Lambda^2
+ \bigl(
 10\, \nu^2 
 - 10 \, \theta_{\infty}\nu
 -\theta_0^2
- \theta_t^2
+ 2 \, \theta_{\infty}^2
\bigr) \, \Lambda^3 + O(\Lambda^4), \\
{\mathscr G}_2 & = \frac{\Lambda^3}{4}
- \left( \frac{\nu}{2} - \frac{7 \, \theta_\infty}{8} \right) \,\Lambda^4 + O(\Lambda^5), \dots.
\end{align}
\end{subequations}
Furthermore, by formally interchanging the order of $\hbar$-expansion and large $t$-expansion, 
we obtain the following formal series expansion for the accessory parameter:
\begin{align}
{\mathscr E} 
& = (\nu-\theta_{\infty})\,t
+ (2\nu^2-2\theta_{\infty}\nu)
\notag \\ 
& \quad - \Biggl(   4 \, \nu^3
- 6 \, \theta_{\infty}  \nu^2
- \Bigl(2 \, (\theta_0^2 + \theta_t^2 - \theta_{\infty}^2) - \hbar^2\Bigr) \, \nu
+ 2\theta_0^2\theta_{\infty}- \frac{\theta_{\infty}\hbar^2}{2} \Biggr) \, t^{-1} 
\notag \\[+.3em]
& \quad 
+ \Biggl( 20 \, \nu^4
- 40 \, \theta_\infty \nu^3
- 2 \, \Bigl( 6 \, (\theta_0^2 + \theta_t^2 - 2\theta_\infty^2) - 5 \,\hbar^2\Bigr) \, \nu^2 
+ 2 \, \theta_\infty \Bigl( 2 \, (5\theta_0^2 + \theta_t^2 - \theta_\infty^2) - 5 \, \hbar^2 \Bigr) \, \nu
\notag \\
& \qquad
+  4 \, \theta_0^2 (\theta_t^2 - 2\theta_\infty^2) 
- (\theta_0^2 + \theta_t^2 -2\theta_\infty^2) \, \hbar^2 
+ \frac{\hbar^4}{4}  \Biggr) \, t^{-2} 
+ O(t^{-3}).
\label{eq:final-AP-HV-3}
\end{align}

\smallskip
\noindent
{\bf Comparison with ${\mathscr W}$.} \,
The relevant conformal block was introduced in \cite{Nagoya15} 
by the second-named author based on the Virasoro algebra and irregular vertex operator, 
and the Zamolodchikov-type conjecture was studied in \cite{LN21}.
The existence of its classical limit was conjectured in \cite[Section 2.3.2, (2.29)--(2.31)]{LN21}, and the expected expansion coefficients of the classical conformal block are given as follows (with $\hbar$ introduced appropriately):
\begin{align}
\widetilde{\mathscr W} & =
(\nu-\theta_{\infty})\,t
+ 2 \, \nu(\nu-\theta_{\infty}) \, \log t
+ \bigl(
4 \, \nu^3
- 6 \, \theta_{\infty} \nu^2
+ 2 \, (\delta_0+\delta_t+\theta_{\infty}^2) \, \nu
- 2 \, \delta_0\theta_{\infty}
\bigr) \, t^{-1} 
\notag \\
& + \Bigl(
-2\, \bigl(\delta_0-\nu(\theta_{\infty}-3\nu)\bigr)
\bigl(\delta_t+(2\theta_{\infty}-3\nu)(\theta_{\infty}-\nu)\bigr)
- 2\nu(\theta_{\infty}-\nu)\bigl((\theta_{\infty}-2\nu)^2-{\hbar^2}\bigr)
\Bigr) \, t^{-2} + O(t^{-3}).
\label{eq:CCB-HV-3}
\end{align}
 
We also expect that \eqref{eq:final-AP-HV-3} and \eqref{eq:CCB-HV-3} are related to the quantum expansion of ${\mathscr Z}$,  called ``linear exp singularity of ${\rm QPV}$ / $N_f=3$ with light hypers", given in \cite[(4.9)--(4.17) in \S 4.2]{BST25}. 
The quantum expansion in \cite{BST25} is given by
\begin{align}
& - \varepsilon_1 \varepsilon_2 \log {\mathscr Z}  \notag \\
& \quad = 
\left(a_D + \frac{e_1 + \varepsilon_1 + \varepsilon_2}{2} \right) \, s
- \left( 2 \, a_D^2 - \frac{w_2 - (\varepsilon_1 + \varepsilon_2)^2}{2} \right) \, \log s
\notag \\ 
& \qquad + \Bigl(  4 \, a_D^3 - \left(w_2 - (\varepsilon_1 + \varepsilon_2)^2\right) \, a_D + e_3 \Bigr) \, s^{-1} 
\notag \\ 
& \qquad + \biggl(
10 \, a_D^4
- \left(3 \, w_2 - 5 \, (\varepsilon_1 + \varepsilon_2)^2\right) \, a_D^2
+ 4 \, e_3 \, a_D
+ \frac{\left(w_2 - (\varepsilon_1 + \varepsilon_2)^2\right)^2}{8}
- \frac{w_4}{2}
\biggr) \, s^{-2} + O(s^{-3}).
\end{align}
The results of \cite{BST25} suggest that the following classical conformal block exists\footnote{In \cite{BST25}, the mass parameters $w_i$ and $e_i$ 
may depend on $\varepsilon_1, \varepsilon_2$, 
but for notational simplicity, in this paper we use the same symbols 
to denote their limits $\varepsilon_2 \to 0$.
In addition, we omit the accents (such as tildes and checks) used for 
the mass parameters in \cite{BST25} in this paper.
}:
\begin{align}
{\mathscr W} & = - \varepsilon_1 \varepsilon_2 \log {\mathscr Z}  
\Bigl|_{(\varepsilon_1, \varepsilon_2) = (\hbar, 0)}  
\notag \\
& 
= \left( a_D + \frac{e_1 + \hbar}{2} \right) \, s
- \left(
2 \, a_D^2 - \frac{w_2 - \hbar^2}{2} 
\right) \, \log s
\notag \\ 
& \quad + \Bigl( 
4 \, a_D^3 - \left(w_2 - \hbar^2\right) \, a_D + e_3 
\Bigr) \, s^{-1} 
\notag \\ 
& \quad + \biggl(
10 \, a_D^4
- \left(3 \, w_2 - 5 \, \hbar^2\right) \, a_D^2
+ 4 \, e_3 \,  a_D
+ \frac{\left(w_2 - \hbar^2\right)^2}{8}
- \frac{w_4}{2}
\biggr) \, s^{-2} + O(s^{-3}).
\label{eq:CCB-HV-3-alt}
\end{align}

As an analogue of Conjecture \ref{conj:H-V-1}, we expect 

\begin{conj}[cf., {\cite[Conjecture 3.1]{LN21}}] 
\label{conj:H-V-2}
The classical conformal blocks \eqref{eq:CCB-HV-3} and \eqref{eq:CCB-HV-3-alt} exist. 
Moreover, under the identification 
\[
s = - t, 
\quad a_D = \nu - \frac{\theta_\infty}{2}, 
\quad e_1 = - \theta_\infty - \hbar, 
\quad e_3 = \theta_\infty(\theta_0^2-\theta_t^2), 
\]
\[
w_2 = 2\theta_0^2+2\theta_t^2+\theta_\infty^2, 
\quad w_4 = (\theta_0^2 - \theta_t^2)^2 +2 \theta_\infty(\theta_0^2+\theta_t^2), 
\]
\begin{equation}
\delta_0=-\theta_0^2+\dfrac{\hbar^2}{4}, \quad
\delta_t=-\theta_t^2+\dfrac{\hbar^2}{4},
\end{equation}
the accessory parameter \eqref{eq:final-AP-HV-3}
and the classical conformal blocks \eqref{eq:CCB-HV-3} and \eqref{eq:CCB-HV-3-alt}
are related as follows:
\begin{equation} \label{eq:H-V-3-alt}
{\mathscr E} = - t \frac{d}{dt} {\mathscr{W}} 
+ \left(\theta_0^2 + \theta_t^2 - \frac{\hbar^2}{2} \right)
= t \frac{d}{dt} \widetilde{\mathscr{W}} . 
\end{equation} 
\end{conj}

From the above computational results, we can observe that the equality \eqref{eq:H-V-3-alt} holds up to $O(t^{-3})$.

\bigskip
\subsubsection{Type 2 expansion at $t \to \infty$} ~ 

\smallskip
\noindent
{\bf Series expansion of ${\mathscr E}$.} \, 
We perform the rescaling with 
$(d_{x}, d_{y}, d_{\mathscr E}) = (1, 1, 2)$ 
and take $\Lambda = t^{-1}$.
The rescaled spectral curve becomes
\begin{equation}
Q_{0}^{\rm res} = 
\frac{\theta_0^2 \, \Lambda^2}{X^2}
+ \frac{\theta_t^2 \, \Lambda^2 }{(X-1)^2}
+ \frac{1}{4}
+ \frac{\theta_\infty \, \Lambda }{X}
- \frac{\mathscr{G}_0}{X(X-1)}
\quad \xrightarrow{ \Lambda \to 0 } \quad 
Q_0^{[0]} = 
\frac{1}{4}
- \frac{\mathscr{G}_0^{[0]}}{X(X-1)}.
\end{equation}
We take $\mathscr{G}_0^{[0]} = -1/16$ to have the limiting spectral curve
\begin{equation}
Y^2 = Q^{[0]}_0 (X) = \frac{(X-\frac{1}{2})^2}{4X(X-1)}
\end{equation}
has genus $0$. 
In this limit, it is readily seen that two zeros of $Q^{\rm res}_0$ tends to $1/2$. 
Therefore, for sufficiently small $\Lambda$, 
we choose the cycle $\gamma$ to be a circle that encloses the two turning points 
which tend to $1/2$ together with the origin and $\Lambda$ inside, 
while keeping the other turning point outside. 
Then, the elliptic integrals defining the Voros period are described 
by the residue at $X = 1/2$ of the coefficients in the small $\Lambda$-expansion.
 
When imposing the condition, 
\begin{equation}
V_{-1} = 2 \pi i \nu, \quad 
V_{0} = - \pi i , \quad 
V_m = 0 \quad (m \ge 1)
\end{equation} 
for the Voros period associated with $\gamma$
(which implies \eqref{eq:partial-IMD} with $+$-sign), 
the resulting computation of the accessory parameter is as follows:
\begin{subequations} \label{eq:AP-HV-2}
\begin{align} 
{\mathscr G}_0 & = - \frac{1}{16} 
- \frac{i \, \nu + \theta_\infty}{2} \, \Lambda 
+ \left(\frac{\nu^2}{2} - \theta_0^2 - \theta_t^2 - \theta_\infty \right) \, \Lambda^2 
 \notag \\ 
 & \quad - \left( \frac{i \, \nu^3}{2} - 4 i \,(2\theta_0^2+2\theta_t^2+\theta_\infty^2) \, \nu  + 8 \, (\theta_0^2-\theta_t^2) \theta_\infty \right) \, \Lambda^3
\notag \\
&
\quad - \left(  \frac{5 \, \nu^4}{4} 
- 24\,(2\theta_0^2 + 2\theta_t^2 + \theta_\infty^2)\,\nu^2 
- 128\, i\, \theta_\infty (\theta_0^2 - \theta_t^2)\,\nu
+ 16\, \Bigl((\theta_0^2 - \theta_t^2)^2
+ 2 \, \theta_\infty^2 (\theta_0^2 + \theta_t^2)\Bigr)
\right) \Lambda^4
\notag \\
&
\quad +O(\Lambda^5), \\
{\mathscr G}_1 & = \frac{\Lambda^2}{8} 
- \frac{11 i \nu}{8} \Lambda^3 
- \left( \frac{65 \, \nu^2}{8} -6 \, (2 \, \theta_0^2+2 \, \theta_t^2 + \theta_\infty^2)   \right) \, \Lambda^4
+O(\Lambda^5), \\
{\mathscr G}_2 & = - \frac{105}{64} \, \Lambda^4 + \frac{5621 \, i \nu}{128} \, \Lambda^5 
+ O(\Lambda^6), \dots.
\end{align}
\end{subequations}
Furthermore, by formally interchanging the order of $\hbar$-expansion and large $t$-expansion, 
we obtain the following formal series expansion for the accessory parameter:
\begin{align}
{\mathscr E} 
& = - \frac{t^2}{16} - \frac{i \, \nu + \theta_\infty}{2} \, t + 
\left( \frac{\nu^2}{2} -\theta_0^2 - \theta_\infty^2 - \theta_t^2
 + \frac{\hbar^2}{8}  \right)
\notag \\ 
& \quad - \left( 
\frac{i \, \nu^3}{2} 
- 4\,i\,\Bigl( 2\, \theta_0^2 + 2 \, \theta_t^2 + \theta_\infty^2  - \frac{11 \, \hbar^2}{32}  \Bigr)\,\nu
+ 8\,\theta_\infty(\theta_0^2 - \theta_t^2)
\right) \, t^{-1} 
+ O(t^{-2}).
\label{eq:final-AP-HV-2}
\end{align}

\smallskip
\noindent
{\bf Comparison with ${\mathscr W}$.} \, 
Let us recall the quantum expansion of ${\mathscr Z}$, 
called ``square exp singularity of ${\rm QPV}$ / $N_f=3$ with heavy hypers", 
given in \cite[(4.26)--(4.35) in \S 4.2]{BST25}\footnote{
In \cite{BST25}, the notation $\alpha_D = i a_D$ was used; 
here, we adopt a formulation expressed solely in terms of $a_D$.}:  
\begin{align}
& - \varepsilon_1 \varepsilon_2 \log {\mathscr Z} = 
\frac{t^2}{32}
+ \frac{i \, a_D + e_1 + \varepsilon_1 + \varepsilon_2}{2} \, t
+ \left(
- \frac{a_D^2}{2} + w_2 
-\frac{5 \, (\varepsilon_1 + \varepsilon_2)^2}{8}
+ \frac{\varepsilon_1 \varepsilon_2}{4}
\right)\log t
\notag \\ 
& \quad + \left( - \frac{i \, a_D^3}{2}
+ \left(4 \, w_2 + \frac{7 \, \varepsilon_1 \varepsilon_2}{4} - \frac{11 \, (\varepsilon_1 + \varepsilon_2)^2}{8} \right) \, i a_D
+ 8 \, e_3 \right) \, t^{-1} 
\notag \\ 
& \quad + \Biggl(
-\frac{5 \, a_D^4}{8}
+ \left(12 \, w_2 + \frac{45 \, \varepsilon_1 \varepsilon_2}{8}  - \frac{65 \, (\varepsilon_1 + \varepsilon_2)^2}{16}\right) \, a_D^2
- 64 \, i \, e_3 \, a_D
\notag \\ 
& \qquad\qquad 
- 8 \, w_4 
- 2 \, w_2 \varepsilon_1 \varepsilon_2
- \frac{(\varepsilon_1 \varepsilon_2)^2}{2}
+ \frac{61 \, \varepsilon_1 \varepsilon_2 (\varepsilon_1 + \varepsilon_2)^2}{32}
\biggr) \, t^{-2} + O(t^{-3}).
\end{align}
The computational results of \cite{BST25} suggest that 
one can expect that 
the following classical conformal block exists:
\begin{align}
{\mathscr W} 
& = - \varepsilon_1 \varepsilon_2 \log {\mathscr Z}  
\Bigl|_{(\varepsilon_1, \varepsilon_2) = (\hbar, 0)} 
\notag \\
& = \frac{t^2}{32}
+ \frac{i \, a_D + e_1 + \hbar}{2} \,  t
+ \left(
- \frac{a_D^2}{2} + w_2 -\frac{5 \, \hbar^2}{8} 
\right) \, \log t
\notag \\ 
& \quad + \left( - \frac{i a_D^3}{2}
+ \left(4 \, w_2 - \frac{11 \, \hbar^2}{8} \right) \, i a_D
+ 8 \, e_3 \right) \, t^{-1} 
\notag \\ 
& \quad + \biggl(
-\frac{5 \, a_D^4}{8} 
+ \left(12 \, w_2 - \frac{65 \, \hbar^2}{16} \right) \, a_D^2
- 64 i \, e_3 \, a_D
- 8 \, w_4 
\biggr) \, t^{-2} + O(t^{-3}).
\label{eq:del-CCB-HV-2}
\end{align}
Then, our conjecture claims 
\begin{conj}
The classical conformal block \eqref{eq:del-CCB-HV-2} exists. 
Moreover, under the identification 
\[
a_D = \nu, \quad e_1 = \theta_\infty - \hbar, 
\quad e_3 = - \theta_\infty(\theta_0^2-\theta_t^2), 
\quad w_2 = 2\theta_0^2+2\theta_t^2+\theta_\infty^2, 
\]
\begin{equation}
w_4 = (\theta_0^2 - \theta_t^2)^2 +2 \,\theta_\infty(\theta_0^2+\theta_t^2) 
- \frac{3 \, (2\theta_0^2+2\theta_t^2+\theta_\infty^2) \, \hbar }{8} 
+ \frac{105 \,\hbar^2}{1024},
\end{equation}
the accessory parameter  \eqref{eq:final-AP-HV-2} 
and the classical conformal block \eqref{eq:del-CCB-HV-2}
are related as follows:
\begin{equation} \label{eq:main-conj-HV-2}
{\mathscr E} = t \frac{d}{dt} {\mathscr W} + \left( \theta_0^2+\theta_t^2-\frac{\hbar^2}{2}\right). 
\end{equation} 
\end{conj}

From the above computational results, we can observe that the equality \eqref{eq:main-conj-HV-2} holds up to $O(t^{-2})$.
Using Mathematica, we can verify that the coefficients at the next order agree as well.

\bigskip 
\subsection{Biconfluent Heun equation $H_{\rm IV}$}

The equation $H_{{\rm IV}}$ has the following potential
\begin{equation}
Q_{{\rm IV}} = \dfrac{\theta_0^2-\frac{\hbar^2}4}{x^2}-\dfrac{\mathscr{E}}{x}+2\theta_\infty + (x+t)^2.
\end{equation}


\bigskip
\subsubsection{Type 1 expansion at $t \to \infty$} ~ 

\smallskip
\noindent
{\bf Series expansion of ${\mathscr E}$.} \, 
We perform the rescaling with 
$(d_{x}, d_{y}, d_{\mathscr E}) = (-1, 0, 1)$ 
and take $\Lambda = t^{-2}$.
The rescaled spectral curve becomes
\begin{equation}
Q_{0}^{\rm res} = 
\dfrac{\theta_0^2}{X^2}-\dfrac{\mathscr{G}_0}{X}+2\theta_\infty \Lambda +(\Lambda \, X+1)^2
\quad \xrightarrow{ \Lambda \to 0 } \quad 
Q_0^{[0]} = \dfrac{\theta_0^2}{X^2}- \dfrac{\mathscr{G}_0^{[0]}}{X}+1. 
\end{equation}
In this limit, it is readily seen that one of the two zeros of 
$Q^{\rm res}_0$ coalesces at infinity, 
and the limiting curve ${\mathcal C}^{\rm res}_{\rm deg}$ is a genus $0$ curve (which is called the Kummer curve in \cite{IKT18-2}). 

Therefore, for sufficiently small $\Lambda$, 
we choose the cycle $\gamma$ to be a circle that encloses the two turning points 
which tend to  
the zeros of $Q_{0}^{[0]}$ together with the origin inside, 
while keeping the other turning point outside. 
Then, the elliptic integrals defining the Voros period are described 
by the residue at $X = \infty$ of the coefficients in the small $\Lambda$-expansion.

When imposing the condition 
\begin{equation}
V_{-1} = 2 \pi i \nu, \quad 
V_m = 0 \quad (m \ge 0)
\end{equation} 
for the Voros period associated with $\gamma$
(which implies \eqref{eq:partial-IMD} with $+$-sign), 
the resulting computation of the accessory parameter is as follows:
\begin{subequations} \label{eq:AP-HIV-1}
\begin{align} 
{\mathscr G}_0 & = 
2 \, \nu
+ \left(3 \, \nu^2 + 2 \, \theta_\infty\nu  -\theta_0^2 \right) \, \Lambda
- \left( 6 \, \nu^3 + 6 \, \theta_\infty\nu^2 - (3 \, \theta_0^2 - \theta_\infty^2) \, \nu - 2 \, \theta_0^2\theta_\infty  \right) \, \Lambda^2 
\notag \\[+.3em]
& \quad 
+ \dfrac{105 \, \nu^4
+ 140 \, \theta_\infty \nu^3
- (66 \, \theta_0^2 - 48 \, \theta_\infty^2)\nu^2
- (72 \, \theta_0^2\theta_\infty - 4 \, \theta_\infty^3 )\nu
+ \theta_0^4 - 16 \, \theta_0^2\theta_\infty^2}{4} \, \Lambda^3
+O(\Lambda^4), \\
{\mathscr G}_1 & = 
\dfrac{\Lambda}{4} 
- \dfrac{ 3 \, \nu +  \theta_\infty}{2} \, \Lambda^2
+ \dfrac{111 \, \nu^2
+ 74 \, \theta_\infty \nu
 - 13 \, \theta_0^2 + 8 \, \theta_\infty^2}{8} \, \Lambda^3
\notag \\
& \quad
- \dfrac{1125 \, \nu^3
+ 1125 \, \theta_\infty\nu^2
- (285 \, \theta_0^2 - 280 \, \theta_\infty^2) \, \nu
-149 \, \theta_0^2\theta_\infty + 16 \, \theta_\infty^3}{8} \, \Lambda^4
+O(\Lambda^5), \\
{\mathscr G}_2 & = 
\dfrac{25}{64} \, \Lambda^3
- \dfrac{281 \, \left( 3 \, \nu + \theta_\infty \right)}{64} \, \Lambda^4
+O(\Lambda^5), \dots.
\end{align}
\end{subequations}
Furthermore, by formally interchanging the order of $\hbar$-expansion and large $t$-expansion, 
we obtain the following formal series expansion for the accessory parameter:
\begin{align}
{\mathscr E} & = 
2 \, \nu \,  t + \left( 3 \, \nu^2 + 2 \, \theta_\infty\nu -\theta_0^2  + \dfrac{\hbar^2}{4} \right) t^{-1} \notag \\
&\quad +\left( 
- \left( 6 \, \nu^3 + 6 \, \theta_\infty\nu^2 - (3 \, \theta_0^2 - \theta_\infty^2) \nu - 2 \, \theta_0^2\theta_\infty  \right)
- \dfrac{(3 \, \nu+\theta_\infty) \,\hbar^2}{2} \right)\, t^{-3} 
\notag \\[+.3em]
&\quad +\Biggl(
 \dfrac{105 \, \nu^4
+ 140 \, \theta_\infty \nu^3
- (66 \, \theta_0^2 - 48 \, \theta_\infty^2) \, \nu^2
- (72 \, \theta_0^2 \theta_\infty - 4 \, \theta_\infty^3 ) \, \nu
+ \theta_0^4 - 16 \, \theta_0^2\theta_\infty^2
}{4} 
\notag \\
&\quad\quad + \dfrac{\left( 111 \, \nu^2
+ 74 \, \theta_\infty \nu
 - 13 \, \theta_0^2 + 8 \, \theta_\infty^2  \right) \, \hbar^2}{8} 
+ \dfrac{25 \, \hbar^4}{64} \Biggr)\, t^{-5}+O(t^{-7}).
\label{eq:final-AP-HIV-1}
\end{align}

 
\smallskip
\noindent
{\bf Comparison with ${\mathscr W}$.} \, 
The relevant conformal block was introduced in \cite{Nagoya15} 
by the second-named author in terms of the Virasoro algebra, 
and the Zamolodchikov-type conjecture was studied in \cite{LN21}.
The existence of its classical limit was conjectured in \cite[Section 2.3.3, (2.36)--(2.38)]{LN21}, and the expected expansion coefficients of the classical conformal block are given as follows (with $\hbar$ introduced appropriately):
\begin{align}\label{eq:CCB-HIV-1-Virasoro}
\widetilde{\mathscr{W}}&=
\nu \,  t^2 
+ (3 \, \nu^2  + 2 \, \theta_\infty \nu + \delta_0) \, \log t
+\dfrac{1}{2} \, \left( 
6 \, \nu^3 + 6 \, \theta_\infty \nu^2
+\left( \theta_\infty^2 + 3 \, \delta_0 + \dfrac{3 \, {\hbar^2}}{4}\right)\nu+2 \, \delta_0\theta_\infty\right) \,
t^{-2} \notag \\
&\quad -\dfrac{1}{4} \, \Biggl( \dfrac{105}{4} \, \nu^4+35 \, \theta_\infty \nu^3
+\left( 12 \, \theta_\infty^2+\dfrac{33 \, \delta_0}{2}+\dfrac{39\,{\hbar^2}}{4} \right) \, \nu^2  \notag\\
&\quad\qquad +  \left( \theta_\infty^2+18 \, \delta_0+\dfrac{19\,{\hbar^2}}{4} \right) \, \theta_\infty \nu 
+\left( 4 \, \theta_\infty^2+\dfrac{\delta_0}{4}+\dfrac{3 \, {\hbar^2}}{2} \right) \, \delta_0 \Biggr) \, t^{-4} 
+ O(t^{-6}),
\end{align}
where, for later comparison, we present here the expression obtained by replacing $\nu$ in \cite{LN21} with $-\nu$. 

We also expect that the classical conformal block \eqref{eq:CCB-HIV-1-Virasoro} is related to the quantum expansion of ${\mathscr Z}$, called ``linear exp singularity of ${\rm QPIV}$ / $H_2$ with light hypers'', 
given in \cite[(4.85)--(4.94) in \S 4.6]{BST25}. 
The quantum expansion in \cite{BST25} is given by
\begin{align}
 - \varepsilon_1 \varepsilon_2 \log {\mathscr Z} 
&  = 
\frac{2 \, a_D + 4 \, m_3 - \varepsilon_1 - \varepsilon_2}{4} \, s
- \frac{12 \, a_D^2 + 16 \, e_2 + (\varepsilon_1+\varepsilon_2)^2}{8} \, \log s
\notag \\ 
& \quad +  
\frac{3\left(4 \, a_D^3 + \left(8 \, e_2 + (\varepsilon_1 + \varepsilon_2)^2 \right) \, a_D  + 8 \, {e}_3 \right)}{2} \, s^{-1}
\notag \\ 
& \quad 
+ 
\biggl(
\frac{105}{4} \, a_D^4 
+ 3 \, \left(22 \, e_2 - \frac{\varepsilon_1 \varepsilon_2}{4} + \frac{37 \,  (\varepsilon_1 + \varepsilon_2)^2}{8} \right)  \, a_D^2
+ 84 \, {e}_3 a_D 
\notag \\ 
& \qquad
+ \frac{\left( 16 \, e_2 + (\varepsilon_1 + \varepsilon_2)^2\right)
\left( 16 \, e_2 - 4 \, \varepsilon_1 \varepsilon_2 + 25 \, (\varepsilon_1 + \varepsilon_2)^2 \right)}{64} 
\biggr) \, s^{-2}
+ O(s^{-3}).
\end{align}
The results of \cite{BST25} suggest that the following classical conformal block exists:
\begin{align}
{\mathscr W} & = - \varepsilon_1 \varepsilon_2 \log {\mathscr Z}  
\Bigl|_{(\varepsilon_1, \varepsilon_2)=({\hbar}, 0)}  
\notag \\
& 
= \frac{2 \, a_D + 4 \, m_3 - \hbar}{4} \,s
- \frac{12 \, a_D^2 + 16 \, e_2 + \hbar^2}{8} \, \log s
\notag \\ 
& \quad +   
\frac{3 \, \left(4 \, a_D^3 + (8 \, e_2 + \hbar^2) \, a_D + 8 \, {e}_3\right)}{2} \, s^{-1}
\notag \\ 
& \quad + 
\frac{
1680 \, a_D^4
+ (4224 \, e_2 + 888 \, \hbar^2) \, a_D^2
+ 5376 \, {e}_3 a_D
+ 256 \, e_2^2 + 416 \,  e_2 \hbar^2 + 25 \, \hbar^4
}{64} \,  s^{-2} + O(s^{-3}).
\label{eq:CCB-HIV-1}
\end{align}

Then, the Zamolodchikov-type conjecture, which has already been suggested in \cite{LN21}, claims the following:


\begin{conj}[cf., {\cite[Conjecture 3.2]{LN21}}] \label{conj:H-IV-1}
The classical conformal blocks \eqref{eq:CCB-HIV-1-Virasoro} and \eqref{eq:CCB-HIV-1} exist. 
Moreover, under the identification 
\[
s = 2t^2, \quad 
a_D = - \nu - \frac{\theta_\infty}{3}, \quad
m_3 = \frac{2\theta_\infty+3\hbar}{12}, 
\]
\begin{equation}
e_2 = -\frac{3 \, \theta_0^2 + \theta_\infty^2}{12}, \quad
{e}_3 = \frac{9 \,\theta_0^2 \theta_\infty - \theta_\infty^3}{108}, \quad 
\delta_0=-\theta_0^2+\dfrac{\hbar^2}{4},
\end{equation}
the accessory parameter 
\eqref{eq:final-AP-HIV-1} and the classical conformal blocks \eqref{eq:CCB-HIV-1-Virasoro} and \eqref{eq:CCB-HIV-1}
are related as follows:
\begin{equation} \label{eq:main-conj-HIV-1}
{\mathscr E} = - \frac{d}{dt}  {\mathscr{W}} = \frac{d}{dt} \widetilde{\mathscr{W}}.
\end{equation} 
\end{conj}

The above computational results show that the equality \eqref{eq:main-conj-HIV-1} holds up to $O(t^{-6})$.

\bigskip
\subsubsection{Type 2 expansion at $t \to \infty$} ~ 

\smallskip
\noindent
{\bf Series expansion of ${\mathscr E}$.} \, 
 We perform the rescaling with 
$(d_{x}, d_{y}, d_{\mathscr E}) = (1, 2, 3)$ 
and take $\Lambda = t^{-2}$.
The rescaled spectral curve becomes
\begin{equation}
Q_{0}^{\rm res} = 
(X+1)^2-\dfrac{\mathscr {G}_0}{X}
+2\theta_\infty \Lambda + \dfrac{\theta_0^2 \, \Lambda^2}{X^2}
\quad \xrightarrow{ \Lambda  \to 0 } \quad 
Q_0^{[0]} =  (X+1)^2- \dfrac{\mathscr{G}^{[0]}_0}{X}.
\end{equation}
We take ${\mathscr G}_0^{[0]} = -{4}/{27}$
so that the limiting spectral curve 
\begin{equation}
Y^2 = Q_0^{[0]}(X) = \frac{(X+\frac13)^2(X+\frac43)}{X}
\end{equation}is of genus $0$. 
Then, in this limit, it is readily seen that two of the three original zeros of $Q^{\rm res}_0$ coalesce at $X = -1/3$. 
Therefore, for sufficiently small $\Lambda$, we choose the cycle $\gamma$ to be a circle that encloses the two turning points 
which tend to the double zero $X=-1/3$ of $Q_{0}^{[0]}$ inside, 
while keeping the other two turning points outside.
Then, the elliptic integrals defining the Voros period are described 
by the residue at $X = -1/3$ of the coefficients in the small $\Lambda$-expansion.
 
When imposing the condition 
\begin{equation}
V_{-1} = 2 \pi i \nu, \quad 
V_{0} = -\pi i, \quad 
V_{m} = 0 \quad (m \ge 1)
\end{equation} 
for the Voros period associated with $\gamma$ (which implies \eqref{eq:partial-IMD} with $-$-sign), the resulting computation of the accessory parameter is as follows:
\begin{subequations} \label{eq:AP-HIV-2}
\begin{align} 
{\mathscr G}_0 & = 
-\dfrac{4}{27}
+ \dfrac{2 \, \left( i\sqrt{3} \, \nu - \theta_\infty  \right)}{3} \, \Lambda
+ \left( \nu^2 - 3 \, \theta_0^2 - \theta_\infty^2 \right) \, \Lambda^2 
\notag \\
& \quad 
+ \left(  \dfrac{2i \, \nu^3 - 9i \, \nu \, (3\theta_0^2 + \theta_\infty^2) }{\sqrt{3}}  + 9 \, \theta_0^2\theta_\infty - \theta_\infty^3 \right) \, \Lambda^3
+O(\Lambda^4), \\
{\mathscr G}_1 & = 
\dfrac{5}{12} \, \Lambda^2
+ \dfrac{7i \, \nu}{2\sqrt{3}} \, \Lambda^3
- \dfrac{17\,  \left( 17\nu^2 - 9 \,(3 \theta_0^2 + \theta_\infty^2 ) \right)}{24} \, \Lambda^4
+O(\Lambda^5), \\
{\mathscr G}_2 & = 
-\dfrac{1105}{576} \, \Lambda^4
+O(\Lambda^5), \quad \dots
\end{align}
\end{subequations}
Furthermore, by formally interchanging the order of $\hbar$-expansion and large $t$-expansion, we obtain the following formal series expansion for the accessory parameter:
\begin{align}
{\mathscr E}  =& 
-\dfrac{4}{27} \, t^3
+\dfrac{2 \left( i\sqrt{3} \, \nu - \theta_\infty \right)}{3} \, t
+\left(\nu^2 - 3 \, \theta_0^2 - \theta_\infty^2  +\dfrac{5 \, \hbar^2}{12} \right) \, t^{-1} \notag \\[+.3em]
&+\left(
\dfrac{i \, \left( 2 \, \nu^3 - 9 \, (3\theta_0^2 + \theta_\infty^2) \, \nu  \right)}{\sqrt{3}}  + 9 \, \theta_0^2\theta_\infty - \theta_\infty^3 
+ \dfrac{7 \, i\nu}{2\sqrt{3}}\hbar^2\right) \, t^{-3}+ O(t^{-4}).
\label{eq:final-AP-HIV-2}
\end{align}

\smallskip
\noindent
{\bf Comparison with ${\mathscr W}$.} \, 
 Let us recall the quantum expansion of ${\mathscr Z}$, 
called ``square exp singularity of ${\rm QPIV}$ / $H_2$ with heavy hypers", 
given in \cite[(4.103)--(4.114) in \S 4.7]{BST25}\footnote{
In \cite{BST25}, the notation $\alpha_D = i a_D/ \sqrt{3}$ was used; 
here, we adopt a formulation expressed solely in terms of $a_D$.}: :  
\begin{align}
& - \varepsilon_1 \varepsilon_2 \log {\mathscr Z} \notag \\ 
& \quad =
\dfrac{s^2}{108}
+\dfrac{(2i\sqrt{3} \, a_D + 12 \, {m_3} -3 \, (\varepsilon_1+\varepsilon_2) )}{12} \, s
-\dfrac{12 \, a_D^2 + 144 \, {e_2} + 5 \, \varepsilon_1^2 + 6 \, \varepsilon_1\varepsilon_2 + 5 \, \varepsilon_2^2}{24} \, \log s\notag\\
& \qquad -\dfrac{4i\sqrt{3} \, a_D^3 + 
i\sqrt{3} \, \bigl(
216 \, {e_2} + 7 \, \varepsilon_1^2 + 6 \, \varepsilon_1 \varepsilon_2 
+ 7 \, \varepsilon_2^2 \bigr) \, a_D 
-648 \, e_3  }{6} \, s^{-1} 
+O(s^{-2}). 
\label{eq:CB-HIV-2}
\end{align}
The result of \cite{BST25} suggests that the associated classical conformal block also exists:
\begin{align}
{\mathscr W}  =& - \varepsilon_1 \varepsilon_2 \log {\mathscr Z}  
\Bigl|_{(\varepsilon_1, \varepsilon_2) = (\hbar, 0)}\notag\\
=& \dfrac{s^2}{108} 
+ \dfrac{2i\sqrt{3} \, a_D + 12 \, {m_3}-3 \, \hbar}{12} \, s
-\dfrac{12 \, a_D^2 + 144 \, {e_2}+5 \, \hbar^2}{24} \, \log s\notag\\[+.3em]
&-\dfrac{4i\sqrt{3} \, a_D^3 
+ i\sqrt{3} \, \bigl( 216 \, e_2 + 7 \, \hbar^2 \bigr) \, a_D 
-648 \, e_3}{6} \, s^{-1}+O(s^{-2}).
\label{eq:del-CCB-HIV-2}
\end{align}
Then, our conjecture claims 


\begin{conj}
The classical conformal block \eqref{eq:del-CCB-HIV-2} exists. 
Moreover, under the identification 
\begin{equation}
s=2 t^2, \quad a_D = -\nu, \quad 
{m_3}=\dfrac{2 \, \theta_\infty + 3 \, \hbar}{12}, \quad
{e_2}=-\dfrac{3 \, \theta_0^2 + \theta_\infty^2}{12}, \quad 
{{e}_3}=\dfrac{9 \, \theta_0^2\theta_\infty-\theta_\infty^3}{108}, 
\end{equation}
the accessory parameter 
\eqref{eq:final-AP-HIV-2}
and the classical conformal block \eqref{eq:del-CCB-HIV-2} are related as follows:
\begin{equation} \label{eq:main-conj-HIV-2}
{\mathscr E} = -\frac{d}{dt} {\mathscr W} 
\end{equation} 
\end{conj}

The above computational results show that the equality \eqref{eq:main-conj-HIV-2} holds up to $O(t^{-4})$.

\bigskip
\subsection{Doubly confluent Heun equation $H_{{\rm III}_1}$}

The equation $H_{{\rm III}_1}$ has the following potential
\begin{equation}
Q_{{\rm III}_1} = \frac{t^2}{4x^4} +  \frac{t \theta_0}{x^3} 
- \frac{{\mathscr E}}{x^2} + \frac{\theta_\infty}{x} + \frac{1}{4}.
\end{equation}

\smallskip
\subsubsection{Expansion at $t \to 0$} ~

\smallskip
\noindent
{\bf Series expansion of ${\mathscr E}$.} \, 
In this case we do not have to perform any rescaling 
(i.e., $d_{x} = d_{y} = d_{\mathscr E} = 0$), 
and can take $\Lambda = t$ as the expansion parameter, and we have
\begin{equation}
Q_0^{\rm res} = \frac{\Lambda^2}{4X^4} +  \frac{\theta_0 \, \Lambda}{X^3} 
- \frac{{\mathscr G}_0}{X^2} + \frac{\theta_\infty}{X} + \frac{1}{4} 
\quad \xrightarrow{ \Lambda \to 0 } \quad 
Q_0^{[0]} = - \frac{{\mathscr G}_0^{[0]}}{X^2} + \frac{\theta_\infty}{X} + \frac{1}{4}. 
\end{equation}
In this limit, it is readily seen that two of the four zeros of 
$Q^{\rm res}_0(X)$ coalesce at the origin, 
and the limiting curve $Y^2 = Q_0^{[0]}(X)$ is of genus $0$ (which is called the Kummer curve in \cite{IKT18-2}).  
We therefore choose, as the vanishing cycle $\gamma$ in this limit, 
the closed path encircling these two zeros together with the origin.
Then, the elliptic integrals defining the Voros period are described 
by the residues at $X = 0$ of the coefficients in the small $\Lambda$-expansion.
 
When imposing the condition 
\begin{equation}
V_{-1} = 2 \pi i \nu, \quad 
V_{0} = \pi i , \quad 
V_m = 0 \quad (m \ge 1)
\end{equation} 
for the Voros period associated with $\gamma$
(which implies \eqref{eq:partial-IMD} with $-$-sign), 
under the assumption $\nu \ne 0$, 
we find
\begin{subequations} \label{eq:AP-HIII-1-1}
\begin{align} 
{\mathscr G}_0 & = - \nu^2 - \frac{\theta_0 \theta_\infty}{2 \, \nu^2} \, \Lambda 
-\frac{\nu^4 - 3 \, (\theta_0^2 + \theta_\infty^2) \, \nu^2 + 5 \, \theta_0^2 \theta_\infty^2}{32 \, \nu^6} \, \Lambda^2 + O(\Lambda^3), \\
{\mathscr G}_1 & = \frac{1}{4} - \frac{\theta_0 \theta_\infty}{8 \, \nu^4} \, \Lambda
-\frac{4 \, \nu^4 - 15 \, (\theta_0^2 + \theta_\infty^2) \, \nu^2 + 42 \, \theta_0^2 \theta_\infty^2}{128 \, \nu^8} \, \Lambda^2 + O(\Lambda^3), \\
{\mathscr G}_2 & = - \frac{\theta_0 \theta_\infty}{32 \, \nu^6} \, \Lambda
- \frac{16 \, \nu^4 - 63 \, (\theta_0^2 + \theta_\infty^2) \, \nu^2 + 219 \, \theta_0^2 \theta_\infty^2}{512 \, \nu^{10}} \,
\Lambda^2 + O(\Lambda^3), \dots
\end{align}
\end{subequations}
Furthermore, by formally interchanging the order of $\hbar$-expansion and large $t$-expansion, 
we obtain the following formal series expansion for the accessory parameter:
\begin{align}
{\mathscr E} 
& = \left( - \nu^2 + \frac{\hbar^2}{4}   \right) 
+  \left( - \frac{\theta_0\theta_\infty}{2 \, \nu^2} 
- \frac{\theta_0\theta_\infty}{8 \, \nu^4} \, \hbar^2 
- \frac{\theta_0\theta_\infty}{32 \, \nu^6} \, \hbar^4 + O(\hbar^6) \right) \, t 
\notag \\
& \quad + \left( -\frac{\nu^4 - 3 \, (\theta_0^2 + \theta_\infty^2)\, \nu^2 +5 \, \theta_0^2 \theta_\infty^2}{32 \, \nu^6} + O(\hbar^2) \right) \,t^{2} + O(t^3)  \notag \\
& = \left( - \nu^2 + \frac{\hbar^2}{4}  \right) - 
\frac{2 \, \theta_0 \theta_\infty}{4 \, \nu^2 - \hbar^2} \, t  + 
\frac{P_2}{32 \, (4 \, \nu^2-\hbar^2)^3(\nu^2 - \hbar^2)} \, t^2 + O(t^3), 
\label{eq:final-AP-HIII-1-1}
\end{align}
where 
\begin{align}
P_2 & = 
- 64 \,\left( \nu^6 
- 3 \, (\theta_0^2 + \theta_\infty^2) \, \nu^4 
+ 5 \, \theta_0^2 \theta_\infty^2 \nu^2 \right) \notag \\[+.3em]
& \qquad + 16 \, \left(3 \, \nu^4 
- 6 \, (\theta_0^2 + \theta_\infty^2)\nu^2
-7 \, \theta_0^2 \theta_\infty^2 
\right) \, \hbar^2 
- 12 \, (\nu^2 - \theta_0^2 - \theta_\infty^2) \, \hbar^4
+ \hbar^6.
\end{align}
The rational expression of the coefficients 
in the last equality in \eqref{eq:final-AP-HIII-1-1}
is obtained by the same method as described in Section \ref{subsec:example-AP-III-regular}.

\bigskip
\noindent
{\bf Comparison with ${\mathscr W}$.} \, 
The relevant conformal block was introduced by \cite{Gaiotto09} in terms of the Virasoro algebra with a regular singular expansion.
It is also identified with the Nekrasov partition function with $N_f = 2$ through the AGT correspondence. 
The associated classical conformal block is expected to coincide with the NS twisted superpotential, and its expansion coefficients are explicitly given as follows (cf.\, \cite[Section 2.2.3, (2.20a)]{LN21}):
\begin{align}\label{eq:CCB-HIII-1-1}
    \mathscr{W} & =
    \delta_{\sigma} \, \log t
+ \frac{\theta_0 \theta_{\infty}}{2\, \delta_{\sigma}}\, t 
+ \Biggl(
- \frac{\theta_{0}^2\theta_{\infty}^2}{16 \, \delta_{\sigma}^3}
+ \frac{(3 \, \theta_{0}^2 + \delta_{\sigma})(3 \, \theta_{\infty}^2+\delta_{\sigma})}
{16 \, \delta_{\sigma}^2 \, (4 \, \delta_{\sigma} + 3 \, \hbar^2)}
\Biggr)t^2
+O(t^3).
\end{align}

Then, the Zamolodchikov-type conjecture, which has already been suggested in previous works, claims the following:
\begin{conj}[\cite{PP16, LN21}]
The classical conformal block \eqref{eq:CCB-HIII-1-1} exists. 
Moreover, under the identification 
\begin{equation}
\delta_\sigma=-\nu^2 + \dfrac{\hbar^2}{4},
\end{equation}
the accessory parameter \eqref{eq:final-AP-HIII-1-1} 
and the classical conformal block \eqref{eq:CCB-HIII-1-1}
are related as follows:
\begin{equation} \label{eq:main-conj-HIII-1-1}
\mathscr{E} = t\dfrac{d}{dt}\mathscr{W}.
\end{equation} 
\end{conj}

The above computational results show that the equality \eqref{eq:main-conj-HIII-1-1} holds up to $O(t^{3})$.

\bigskip
\subsubsection{Expansion at $t \to \infty$} ~

\smallskip 
\noindent
{\bf Series expansion of ${\mathscr E}$.} \,  We perform the rescaling with 
$(d_{x}, d_{y}, d_{\mathscr E}) = (1/2 , 1/2, 1)$ 
and take $\Lambda = t^{-1/2}$.
The rescaled spectral curve becomes
\begin{equation}
Q_{0}^{\rm res} = \dfrac{1}{4X^4}-\dfrac{\mathscr{G}_0}{X^2}+\dfrac{1}{4} + \frac{\theta_0 \, \Lambda}{X^3}+\frac{\theta_\infty \, \Lambda}{X}
\quad \xrightarrow{ \Lambda  \to 0 } \quad 
Q_0^{[0]} =  \dfrac{1}{4X^4} - \dfrac{\mathscr{G}^{[0]}_0}{X^2}+\dfrac{1}{4}. 
\end{equation}
We take ${\mathscr G}_0^{[0]} = -1/2$
so that the limiting spectral curve 
\begin{equation}
Y^2 = Q_0^{[0]}(X) = \frac{(X^2+1)^2}{4X^4}
\end{equation}
is of genus $0$\footnote{
There is another choice ${\mathscr G}_0^{[0]} = 1/2$: 
In this case, we can also obtain a similar series expansion of accessory parameter which is equivalent to \eqref{eq:final-AP-HIII-1-2}. 
More precisely, the accessory parameter is obtained from \eqref{eq:final-AP-HIII-1-2}
by $t^{-k/2} \mapsto i^k t^{-k/2}$ and $(\theta_0,\theta_\infty) \mapsto (-\theta_0, \theta_\infty)$. 
}. 
In this limit, it is readily seen that 
two of the zeros of $Q^{\rm res}_0$ coalesce at $X=i$, 
while the other two of the zeros coalesce at $X=-i$.
Therefore, for sufficiently small $\Lambda$, 
we choose $\gamma = \gamma_{i} - \gamma_{-i}$, 
where $\gamma_{\pm i}$ is the positively oriented 
residue cycle around $X = \pm i$.
Then, the elliptic integrals defining the Voros period are described 
by the difference of the residues at $X = \pm i$ 
of the coefficients in the small $\Lambda$-expansion.
 
When imposing the condition
\begin{equation}
V_{-1} = 2 \pi i \nu, \quad 
V_{m} = 0 \quad (m \ge 0)
\end{equation} 
for the Voros period associated with $\gamma$
(which implies \eqref{eq:partial-IMD} with $+$-sign), 
the resulting computation of the accessory parameter is as follows:
\begin{subequations} \label{eq:AP-HIII-1-2}
\begin{align} 
{\mathscr G}_0 & = - \frac{1}{2} 
- i \nu \, \Lambda 
+ \dfrac{\nu^2 - 3 \, (\theta_0+\theta_\infty)^2 + 4 \, \theta_0 \theta_\infty}{8}\, \Lambda^2
- \dfrac{ i \nu^3 - i \nu \, \left( 
9 \,(\theta_0+\theta_\infty)^2-4 \, \theta_0 \theta_\infty \right) }{64}\, \Lambda^3
+ O(\Lambda^4), \\
{\mathscr G}_1 & =  
 \frac{3}{8} \, \Lambda^{2}
- \frac{3 i \, \nu}{64} \, \Lambda^{3}
- \frac{17 \, \nu^2 - 21 \, (\theta_0+\theta_\infty)^2  + 20 \, \theta_0 \theta_\infty}{512} \, \Lambda^{4}
+ O(\Lambda^5), \\
{\mathscr G}_2 & = \frac{9}{1024}\, \Lambda^{4}
+ \frac{405 \, i \nu}{16384}\, \Lambda^{5} + O(\Lambda ^6), \dots
\end{align}
\end{subequations}
Furthermore, by formally interchanging the order of $\hbar$-expansion and large $t$-expansion, 
we obtain the following formal series expansion for the accessory parameter:
\begin{align}
{\mathscr E} & = 
- \frac{t}{2} - i  \nu \, t^{1/2} 
+ \dfrac{\nu^2 - 3 \, (\theta_0+\theta_\infty)^2 + 4 \, \theta_0 \theta_\infty + 3 \, \hbar^2}{8}
\notag \\
& \quad 
- \dfrac{i  \nu^3 - i  \nu \, \left( 
9 \, (\theta_0+\theta_\infty)^2 - 4 \, \theta_0 \theta_\infty -3 \, \hbar^2 \right)}{64}\, t^{-1/2} 
\notag \\ 
& \quad 
- \frac{1}{1024} \, \biggl(
5 \, \nu^4
- \left(
102 \, ( \theta_0+\theta_\infty)^2 
- 24 \, \theta_0 \theta_\infty
- 34 \, \hbar^2
\right) \, \nu^2
\notag \\
& \qquad 
+ 33 \, (\theta_0+\theta_\infty)^4
- 136 \, (\theta_0+\theta_\infty)^2 \theta_0 \theta_\infty
+ 16 \, \theta_0^2 \theta_\infty^2
- \left(42 \, (\theta_0+\theta_\infty)^2 - 40 \, \theta_0\theta_\infty \right) \, \hbar^2
+ 9 \, \hbar^4
\biggr) \, t^{-1}
\notag \\
& \qquad + O(t^{-3/2}).
\label{eq:final-AP-HIII-1-2}
\end{align}

\smallskip
\noindent
{\bf Comparison with ${\mathscr W}$.} \, Let us recall the quantum expansion of ${\mathscr Z}$, 
called ``${\rm QPIII}_1 / N_f = 2$", 
given in \cite[(4.40)--(4.47) in \S 4.3]{BST25}:  
\begin{align}
& - \varepsilon_1 \varepsilon_2 \log {\mathscr Z} = 
- \frac{s^2}{64}
+ \frac{a_D \, s}{2}
- \left( a_D^2 - \frac{e_2}{2} - \frac{3 \, \big(w_2 - (\varepsilon_1+\varepsilon_2)^2 \big)}{4} \right) \, \log s 
\notag \\
& \quad 
- \left(2 \, a_D^3 - \frac{14 \, e_2 + 9 \, w_2 + 2 \, \varepsilon_1 \varepsilon_2 - 3 \, (\varepsilon_1+\varepsilon_2)^2}{2} \, a_D \right)  \, s^{-1} 
\notag \\
& \quad 
- \Biggl( 
5 \, a_D^4 
- \frac{90 \, e_2 + 51 \, w_2 + 14 \, \varepsilon_1 \varepsilon_2 - 17 \, (\varepsilon_1+\varepsilon_2)^2}{2} \, a_D^2
\notag \\
& \quad\quad 
- \frac{\bigl( 62 \, e_2 + 33 \, w_2 + 12 \, \varepsilon_1 \varepsilon_2 - 9 \, (\varepsilon_1+\varepsilon_2)^2\bigr) 
\bigl(2 \, e_2 - w_2 + (\varepsilon_1+\varepsilon_2)^2 \bigr)}{16}
\Biggr) \, s^{-2} + O(s^{-3}).
\end{align}
As far as the computations in \cite{BST25} go, the expansion coefficients 
on the right-hand side do not involve $\varepsilon_1, \varepsilon_2$ in the denominator. 
One can expect that the same property holds at all orders, and in particular, 
the following classical conformal block to exist: 
\begin{align}
{\mathscr W} 
& = - \varepsilon_1 \varepsilon_2 \log {\mathscr Z}  
\Bigl|_{(\varepsilon_1, \varepsilon_2) = (\hbar, 0)} 
\notag \\
& = - \frac{s^2}{64}
+ \frac{a_D}{2} \, s 
- \left( a_D^2 - \frac{e_2}{2} - \frac{3 \, \big(w_2 - \hbar^2 \big)}{4} \right) \,
\log s
\notag \\ 
& \quad - \left(2 \, a_D^3 - \frac{14 \, e_2 + 9 \, w_2 
- 3 \, \hbar^2}{2} \, a_D \right) \, s^{-1}
\notag \\ 
& \quad - \biggl(
5 \, a_D^4 
- \frac{90 \, e_2 + 51 \, w_2 - 17 \, \hbar^2}{2} \, a_D^2
- \frac{\bigl( 62 \, e_2 + 33 \, w_2 - 9 \, \hbar^2\bigr) 
\bigl(2 \, e_2 - w_2 + \hbar^2 \bigr)}{16}
\biggr) \, s^{-2} + O(s^{-3}).
\label{eq:CCB-HIII-1-2}
\end{align}

Then, our conjecture claims 

\begin{conj}
The classical conformal block \eqref{eq:CCB-HIII-1-2} exists. 
Moreover, under the identification 
\begin{equation}
s = 8 i \, t^{1/2}, \quad
a_D = \frac{\nu}{2}, \quad 
e_2 =  \theta_0 \theta_\infty, 
\quad w_2 = \theta_0^2+\theta_\infty^2, 
\end{equation}
the accessory parameter  \eqref{eq:final-AP-HIII-1-2} 
and the classical conformal block \eqref{eq:CCB-HIII-1-2}
are related as follows:
\begin{equation} \label{eq:main-conj-HIII-1-2}
{\mathscr E} = - t \frac{d}{dt} {\mathscr W} + \frac{t}{2}. 
\end{equation} 
\end{conj}

The above computational results show that the equality \eqref{eq:main-conj-HIII-1-2} holds up to $O(t^{-3/2})$.

\subsection{Reduced doubly confluent Heun equation $H_{{\rm III}_2}$}

The equation $H_{{\rm III}_2}$ has the following potential
\begin{equation}
Q_{{\rm III}_2} = \dfrac{t}{x^3}-\dfrac{\mathscr{E}}{x^2}+\dfrac{\theta_\infty}{x}+\dfrac{1}{4}.
\end{equation}

\bigskip
\subsubsection{Type 1 expansion at $t \to 0$} ~ 

\smallskip
\noindent
{\bf Series expansion of ${\mathscr E}$.} \, 
We perform the rescaling with 
$(d_{x}, d_{y}, d_{\mathscr E}) = (0, 0, 0)$ 
and take $\Lambda = t$.
The rescaled spectral curve becomes
\begin{equation}
Q_{0}^{\rm res} = \dfrac{\Lambda}{X^3}-\dfrac{\mathscr{G}_0}{X^2} +\dfrac{\theta_\infty}{X}+ \dfrac{1}{4}
\quad \xrightarrow{ \Lambda \to 0 } \quad 
Q_0^{[0]} = - \dfrac{\mathscr{G}_0^{[0]}}{X^2}+\dfrac{\theta_\infty}{X}+\dfrac{1}{4}. 
\end{equation}
In this limit, it is readily seen that one of the three zeros of 
$Q^{\rm res}_0$ coalesces at $0$, 
and the limiting curve ${\mathcal C}^{\rm res}_{\rm deg}$ is a genus $0$ curve (which is called the Kummer curve in \cite{IKT18-2}). 

Therefore, for sufficiently small $\Lambda$, 
we choose the cycle $\gamma$ to be a circle that encloses the two turning points 
which tend to  
the zeros of $Q_{0}^{[0]}$ together with the origin inside, 
while keeping the other turning point outside. 
Then, the elliptic integrals defining the Voros period are described 
by the residues at $X = 0$ of the coefficients in the small $\Lambda$-expansion.

When imposing the condition
\begin{equation}
V_{-1} = 2 \pi i \nu, \quad 
V_{0} = \pi i, \quad 
V_{m} = 0 \quad (m \ge 1)
\end{equation} 
for the Voros period associated with $\gamma$
(which implies \eqref{eq:partial-IMD} with $-$-sign), 
under the assumption $\nu \ne 0$, 
we find 
\begin{subequations} \label{eq:AP-HIII-2-1}
\begin{align} 
{\mathscr G}_0 & = 
-\nu^2 
- \dfrac{\theta_\infty}{2 \, \nu^2}\, \Lambda
+ \dfrac{3 \, \nu^2-5 \, \theta_\infty^2}{32 \, \nu^6}\, \Lambda^2
+ \dfrac{7 \, \theta_\infty\nu^2-9 \, \theta_\infty^3}{64 \, \nu^{10}}\, \Lambda^3\notag\\
& \quad - \dfrac{153 \, \nu^4-1430 \, \theta_\infty^2\nu^2+1469 \,\theta_\infty^4}{8192 \, \nu^{14}}\, \Lambda^4
+O(\Lambda^5), \\
{\mathscr G}_1 & =  
\dfrac{1}{4}
- \dfrac{\theta_\infty}{8 \,\nu^4}\, \Lambda
+ \dfrac{3 \,(5 \,\nu^2-14 \,\theta_\infty^2)}{64 \, \nu^8}\, \Lambda^2
+ \dfrac{129 \,\theta_\infty\nu^2-220 \, \theta_\infty^3}{256 \,\nu^{12}}\, \Lambda^3\notag\\
& \quad - \dfrac{2475 \, \nu^4-29549 \, \theta_\infty^2\nu^2+36890 \, \theta_\infty^4}{16384 \, \nu^{16}}\, \Lambda^4
+O(\Lambda^5), \\
{\mathscr G}_2 & = 
- \dfrac{\theta_\infty}{32 \, \nu^6}\, \Lambda
+ \dfrac{3 \, (21 \, \nu^2-73 \, \theta_\infty^2)}{512 \,\nu^{10}}\, \Lambda^2
+ \dfrac{23 \, (33 \, \theta_\infty\nu^2-65 \, \theta_\infty^3)}{512 \, \nu^{14}}\, \Lambda^3\notag\\
& \quad - \dfrac{101439 \, \nu^4-1413300 \, \theta_\infty^2\nu^2+1971898 \, \theta_\infty^4}{131072 \, \nu^{18}}\, \Lambda^4
+O(\Lambda^5), \dots
\end{align}
\end{subequations}
Furthermore, by formally interchanging the order of $\hbar$-expansion and large $t$-expansion, 
we obtain the following formal series expansion for the accessory parameter:
\begin{align}
{\mathscr E} & = 
\left(-\nu^2 + \dfrac{\hbar^2}{4} \right)+\left( -\dfrac{\theta_\infty}{2 \, \nu^2} - \dfrac{\theta_\infty}{8 \, \nu^4} \, \hbar^2 -  \dfrac{\theta_\infty}{32 \, \nu^6} \, \hbar^4+O(\hbar^6)\right)\, t \notag \\
&\quad +\left(\dfrac{3 \, \nu^2-5 \, \theta_\infty^2}{32 \, \nu^6} + \dfrac{3 \, (5 \, \nu^2 - 14 \, \theta_\infty^2)}{64 \, \nu^8} \, \hbar^2 + \dfrac{3 \, (21 \, \nu^2-73 \, \theta_\infty^2)}{512 \, \nu^{10}} \, \hbar^4 + O(\hbar^6)\right)\, t^2 \notag \\
&\quad +\left(\dfrac{7\,\theta_\infty\nu^2-9\,\theta_\infty^3}{64\,\nu^{10}}
+\dfrac{129\,\theta_\infty\nu^2-220\,\theta_\infty^3}{256\,\nu^{12}} \, \hbar^2 
+\dfrac{23\,(33 \, \theta_\infty\nu^2-65 \, \theta_\infty^3)}{512 \, \nu^{14}} \, \hbar^4 + O(\hbar^6)\right)\, t^3  + O(t^4)
\notag \\
&  = 
\left(-\nu^2 + \dfrac{\hbar^2}{4}\right)
- \frac{\theta_\infty}{4 \, \nu^2-\hbar^2} \, t 
+ \frac{
48 \, \nu^4 - \left(80 \, \theta_\infty^2 + 24 \, \hbar^2\right) \, \nu^2 - \left(28 \, \theta_\infty^2 \hbar^2  - 3 \, \hbar^4\right) }
{8 \left(4 \, \nu^2 - \hbar^2\right)^3 \left(\nu^2 - \hbar^2\right)} \, t^2 \notag \\
& \quad + \frac{\theta_\infty \Bigl( 
448 \, \nu^6
- 48 \left(12\, \theta_\infty^2 - \hbar^2\right) \, \nu^4
- 4 \left(232 \, \theta_\infty^2 + 27 \, \hbar^2\right) \, \hbar^2 \nu^2 
- \left(116 \, \hbar^4 \theta_\infty^2  - 17 \, \hbar^6\right)
\Bigr)}{\left(4 \, \nu^2 - \hbar^2\right)^5 \left(4 \, \nu^2 - 9 \, \hbar^2\right) \left(\nu^2 - \hbar^2\right)} \, t^3 + O(t^4).
\label{eq:final-AP-HIII-2-1}
\end{align}

\smallskip
\noindent 
{\bf Comparison with ${\mathscr W}$.} \, 
 The relevant classical conformal block is expected to coincide with the NS twisted superpotential with $N_f = 1$, and its expansion coefficients are explicitly given as follows (\cite[Section 2.2.3, (2.20b)]{LN21}):
\begin{align}\label{eq:CCB-HIII-2-1}
    {\mathscr W} =\delta_\sigma \,  \ln t+\dfrac{\theta_\infty}{2\delta_\sigma} \, t 
    + \dfrac{(5 \, \delta_\sigma-3 \, \hbar^2) \, \theta_\infty^2+3 \, \delta_\sigma^2}{16 \, \delta_\sigma^3\, (4\, \delta_\sigma+3\, \hbar^2)} \, t^2+O(t^3).
\end{align}

The Zamolodchikov-type conjecture, which has already been suggested in previous works, claims the following:

\begin{conj}[Cf.\, \cite{PP16, LN21}]
The classical conformal block \eqref{eq:CCB-HIII-2-1} exists. 
Moreover, under the identification 
\begin{equation}
\delta_\sigma=-\nu^2+\dfrac{\hbar^2}{4}
\end{equation}
the accessory parameter 
\eqref{eq:final-AP-HIII-2-1} and the classical conformal block \eqref{eq:CCB-HIII-2-1}
are related as follows:
\begin{equation} \label{eq:main-conj-HIII2-1}
\mathscr{E} = t\dfrac{d}{dt}\mathscr{W} 
\end{equation} 
\end{conj}

\bigskip


\subsubsection{Type 2 expansion at $t \to \infty$} 

\,

\smallskip
\noindent
{\bf Series expansion of ${\mathscr E}$.} \, 
 We perform the rescaling with 
$(d_{x}, d_{y}, d_{\mathscr E}) = (1/3, 1/3, 2/3)$ 
and take $\Lambda = t^{-1/3}$.
The rescaled spectral curve becomes
\begin{equation}
Q_{0}^{\rm res} = \dfrac{1}{X^3}-\dfrac{\mathscr{G}_0}{X^2}+\dfrac{\theta_\infty}{X}\Lambda+\dfrac{1}{4}
\quad \xrightarrow{ \Lambda  \to 0 } \quad 
Q_0^{[0]} =  \dfrac{1}{X^3}-\dfrac{\mathscr{G}^{[0]}_0}{X^2}+\dfrac{1}{4}.
\end{equation}
We take ${\mathscr G}_0^{[0]} = \frac{3}{2^{4/3}}$
to have the limiting spectral curve 
\begin{equation}
Y^2 = Q_0^{[0]}(X) = \frac{(X-2^{1/3})^2(X+2^{4/3})}{4X^3}
\end{equation}
is of genus $0$. 
Then, in this limit, it is readily seen that two of the three zeros of 
$Q^{\rm res}_0$ coalesce at $X = 2^{1/3}$, 
and the limiting curve ${\mathcal C}^{\rm res}_{\rm deg}$ is a genus $0$ curve. 
Therefore, for sufficiently small $\Lambda$, 
we choose the cycle $\gamma$ to be a circle that encloses the two turning points 
which tend to the double zero $X=2^{1/3}$ of $Q_{0}^{[0]}$ inside, 
while keeping the other two turning point outside.
Then, the elliptic integrals defining the Voros period are described 
by the residue at $X = 2^{1/3}$ of the coefficients in the small $\Lambda$-expansion.
 
When imposing the condition
\begin{equation}
V_{-1} = 2 \pi i \nu, \quad 
V_0 = -\pi i, \quad 
V_{m} = 0 \quad (m \ge 1)
\end{equation} 
for the Voros period associated with $\gamma$
(which implies \eqref{eq:partial-IMD} with $-$-sign), 
the resulting computation of the accessory parameter is as follows:
\begin{subequations} \label{eq:AP-HIII-2-2}
\begin{align} 
{\mathscr G}_0 & =  \dfrac{3}{2^{4/3}} - 2^{1/3} \left( \sqrt{3} \, \nu - \theta_\infty \right) \Lambda +\dfrac{\nu^2-2\,\theta_\infty^2}{6}\, \Lambda^2
\notag \\ 
& \quad -\dfrac{5\,\nu^3-24\sqrt{3} \, \theta_\infty \nu^2+72\,\theta_\infty^2 \nu - 16\sqrt{3} \,\theta_\infty^3}{108\cdot 2^{1/3}\cdot \sqrt{3}} \, \Lambda^3
+ O(\Lambda ^4), \\
{\mathscr G}_1 & =  
 \frac{23}{72} \, \Lambda ^{2}
- \frac{25 \, \nu - 8\sqrt{3} \, \theta_\infty}{432\cdot 2^{1/3}\cdot \sqrt{3}} \, \Lambda ^{3}
+ O(\Lambda ^4), \\
{\mathscr G}_2 & = \frac{1571}{559872\cdot 2^{1/3}}\, \Lambda ^{4}
+ O(\Lambda ^5), \dots
\end{align}
\end{subequations}
Furthermore, by formally interchanging the order of $\hbar$-expansion and large $t$-expansion, 
we obtain the following formal series expansion for the accessory parameter:
\begin{align}
{\mathscr E} & = 
\dfrac{3}{2^{4/3}} \, t^{2/3} 
- 2^{1/3} \left( \sqrt{3} \, \nu - \theta_\infty \right) \, t^{1/3}
+\left(\dfrac{\nu^2-2\,\theta_\infty^2}{6}+\dfrac{23 \, \hbar^2}{72}\right)\notag\\[+.3em]
& \quad  - \dfrac{20 \,\nu^3-96\sqrt{3}\,\theta \nu^2+288 \, \theta_\infty^2\nu-64\sqrt{3} \, \theta_\infty^3 - (25 \, \nu-8\sqrt{3} \, \theta_\infty)\,\hbar^2}{432\cdot 2^{1/3}\cdot \sqrt{3}} \, t^{-1/3}
+O(t^{-2/3}).
\label{eq:final-AP-HIII-2-2}
\end{align}

\smallskip
\noindent
{\bf Comparison with ${\mathscr W}$.} \, 
Let us recall the quantum expansion of ${\mathscr Z}$, 
called ``square exp singularity of ${\rm QPIII_2}$ / $N_f=1$",  
given in \cite[(4.55)--(4.64) in \S 4.5]{BST25}\footnote{
We set $\psi=1$ and $\alpha_D = a_D/\sqrt{3}$ in \cite{BST25}. 
As in previous examples, we have omitted the tilde on the variables.
}:

\begin{align}
- \varepsilon_1 \varepsilon_2 \log {\mathscr Z} 
& = -\frac{s^2}{8}
+ \left( \sqrt{3}\, a_D - m_1 \right) \,  s
- \left( \frac{a_D^2}{2} - m_1^2 
+ \frac{\varepsilon_1 \varepsilon_2}{12} 
+ \frac{23 \,  (\varepsilon_1+\varepsilon_2)^2}{24}\right) \, \log s 
\notag \\
& 
\quad + 
\Biggl( -\frac{5}{12 \sqrt{3}} \,  a_D^3
+ 2 \, m_1 a_D^2 
- 
\left(
6 \, m_1^2 + \frac{13 \, \varepsilon_1 \varepsilon_2}{24}
- \frac{25 \, (\varepsilon_1+\varepsilon_2)^2}{48} 
\right) \frac{a_D}{\sqrt{3}} 
\notag \\
& \qquad  +  \frac{m_1 \, \left(
8 \, m_1^2 + 2 \, \varepsilon_1 \varepsilon_2 - (\varepsilon_1+\varepsilon_2)^2 \right)}{6} 
\Biggr) \, s^{-1} + O(s^{-2}).
\label{eq:CB-HIII-2-2}
\end{align}
For this case as well, it is expected from the above results that the following classical conformal block exists:
\begin{align}
{\mathscr W} & =  - \varepsilon_1 \varepsilon_2 \log {\mathscr Z}  
\Bigl|_{(\varepsilon_1, \varepsilon_2) = (\hbar, 0)} \notag\\
& = -\frac{s^2}{8}
+ \left( \sqrt{3} \, a_D - m_1 \right)s
- \left( \frac{a_D^2}{2} - m_1^2 
+ \frac{23 \, \hbar^2}{24}\right) \, \log s 
\notag \\
& 
\quad + 
\Biggl( -\frac{5}{12 \sqrt{3}} \, a_D^3
+ 2 \, m_1 a_D^2 
- 
\left(
6 \, m_1^2 - \frac{25 \, \hbar^2}{48} 
\right) \frac{a_D}{\sqrt{3}}  
+ \frac{m_1 \, \left(8 \, m_1^2  - \hbar^2 \right)}{6} 
\Biggr) \, s^{-1} + O(s^{-2}).
%
%
%
\label{eq:del-CCB-HIII-2-2}
\end{align}
Then, our conjecture claims 

\begin{conj}
The classical conformal block \eqref{eq:del-CCB-HIII-2-2} exists. 
Moreover, under the identification 
\begin{equation}
s = (54 \, t)^{1/3}, \quad 
a_D = \nu, \quad 
{m_1}=\theta_\infty,
\end{equation}
the accessory parameter 
\eqref{eq:final-AP-HIII-2-2} and the classical conformal block \eqref{eq:del-CCB-HIII-2-2}
are related as follows:
\begin{equation} \label{eq:main-conj-HIII2-2}
{\mathscr E} = -t\frac{d}{dt} {\mathscr W}. 
\end{equation} 
\end{conj}

The above computational results show that the equality \eqref{eq:main-conj-HIII2-2} holds up to $O(t^{-2/3})$, and we have further verified the agreement of the coefficients in the subsequent terms.

\bigskip




\subsection{Triconfluent Heun equation $H_{\rm II}$}

The equation $H_{\rm II}$ has the following potential
\begin{equation}
Q_{\rm II} = (x^2+t)^2 + 2 \theta_\infty x +{\mathscr E}.
\end{equation}

\smallskip
\subsubsection{Type 1 expansion at $t \to \infty$} ~ 

\smallskip
\noindent
{\bf Series expansion of ${\mathscr E}$.} \, 
We perform the rescaling with 
$(d_{x}, d_{y}, d_{\mathscr E}) = (1/2, 3/2, 1/2)$ 
and take $\Lambda = t^{-3/2}$.
The rescaled spectral curve becomes
\begin{equation}
Q_{0}^{\rm res} = (X^2 + 1)^2 + (2 \theta_\infty X + {\mathscr G}_0) \,  \Lambda 
\quad \xrightarrow{ \Lambda  \to 0 } \quad 
Q_0^{[0]} =  (X^2 + 1)^2.
\end{equation}
In this limit, it is readily seen that two of the zeros of $Q^{\rm res}_0$ coalesce at $X=i$, 
while the other two of the zeros coalesce at $X=-i$.
The limiting curve ${\mathcal C}^{\rm res}_{\rm deg}$ is a genus $0$ curve. 

Here, for sufficiently small $\Lambda$, 
we choose $\gamma = \gamma_{i} - \gamma_{-i}$, 
where $\gamma_{\pm i}$ is the positively oriented residue cycle around $X = \pm i$.
Then, the elliptic integrals defining the Voros period are described 
by the difference of the residues at $X = \pm i$ of the coefficients in the small $\Lambda$-expansion.
 
When imposing the condition 
\begin{equation}
V_{-1} = 2 \pi i \nu, \quad 
V_{m} = 0 \quad (m \ge 0)
\end{equation} 
for the Voros period associated with $\gamma$
(which implies \eqref{eq:partial-IMD} with $+$-sign), 
the resulting computation of the accessory parameter is as follows:
\begin{subequations} \label{eq:AP-HII-1}
\begin{align} 
{\mathscr G}_0 & =  2 i \, \nu - \frac{3 \, \nu^2 - \theta_\infty^2}{8} \, \Lambda  
+ \frac{i \, (17 \, \nu^3 -  9 \, \theta_\infty^2)}{128} \, \Lambda ^2 
+\frac{375 \, \nu^{4} - 258 \, \theta_{\infty}^{2} \nu^{2} +11 \,\theta_{\infty}^{4}}{4096} \,\Lambda ^3
+ O(\Lambda ^4), \\
{\mathscr G}_1 & =  -\frac{\Lambda}{8} 
+ \frac{19 i \, \nu}{128} \, \Lambda ^{2}
+ \frac{459 \, \nu^{2} - 71 \, \theta_{\infty}^{2}}{2048} \, \Lambda ^{3}
+ O(\Lambda ^4), \\
{\mathscr G}_2 & = \frac{131}{4096}\, \Lambda ^{3}
- \frac{22707i \, \nu}{131072}\, \Lambda^{4} + O(\Lambda^5), \dots
\end{align}
\end{subequations}
Furthermore, by formally interchanging the order of $\hbar$-expansion and large $t$-expansion, 
we obtain the following formal series expansion for the accessory parameter:
\begin{align}
{\mathscr E} & = 
2 i \, \nu \, t^{1/2} - \frac{3\, \nu^2 - \theta_\infty^2 + \hbar^2}{8} \,t^{-1} 
+ \frac{i \, \nu \, (17 \, \nu^2 - 9 \, \theta_\infty^2 + 19 \, \hbar^2)}{128} \,t^{-5/2} 
\notag \\
& \quad + \frac{375 \, \nu^4 - (258\,  \theta_\infty^2 - 918 \, \hbar^2) \, \nu^2 
+  11 \,  \theta_\infty^4 - 142 \, \hbar^2 \theta_\infty^2   + 131 \, \hbar^4}{4096} \, t^{-4}
+ O(t^{-11/2}).
\label{eq:final-AP-HII-1}
\end{align}

\smallskip
\noindent
{\bf Comparison with ${\mathscr W}$.} \, 
Let us recall the quantum expansion of ${\mathscr Z}$, 
called ``linear exp singularity of ${\rm QPII}$ / $H_1$ with light hyper", 
given in \cite[(4.119)--(4.125) in \S 4.8]{BST25}:  
\begin{align}
& - \varepsilon_1 \varepsilon_2 \log {\mathscr Z} 
\notag \\
& \quad =  
\frac{a_D}{3} \, s
- \left(a_D^2 + \frac{(\varepsilon_1 + \varepsilon_2)^2}{12} - \frac{4 \, {m}^2}{3} 
\right) \, \log s
+ \left( 
\frac{17 \, a_{D}^{3}}{3} 
- \Bigl( 12 \, {m}^{2}
+ \frac{\varepsilon_{1}\varepsilon_{2}}{6}
- \frac{19 \, (\varepsilon_1+\varepsilon_2)^{2}}{12}  \Bigr) \, a_{D}
\right) \, s^{-1}
\notag \\[+.3em] 
& \qquad + 
\Biggl( 
\frac{125}{4} \, a_{D}^{4} 
- \Bigl( 86 \, {m}^{2}
+ \frac{15 \, \varepsilon_{1}\varepsilon_{2}}{4}
- \frac{153 \, (\varepsilon_{1}+\varepsilon_{2})^{2}}{8}  \Bigr) a_{D}^{2}
\notag \\
& \quad\qquad 
+ \Bigl( 16 \, {m}^{2} - (\varepsilon_{1}+\varepsilon_{2})^{2} \Bigr)
   \Bigl( \frac{11 \, {m}^{2}}{12}
        + \frac{17 \, \varepsilon_{1}\varepsilon_{2}}{48}
        - \frac{131 \,(\varepsilon_{1}+\varepsilon_{2})^{2}}{192}  \Bigr)
        \Biggr) \, s^{-2} 
+ O(s^{-3}).
\end{align}
As far as the computations in \cite{BST25} go, the expansion coefficients 
on the right-hand side do not involve $\varepsilon_1, \varepsilon_2$ in the denominator. 
One can expect that the same property holds at all orders, and in particular, 
the following classical conformal block to exist: 
\begin{align}
{\mathscr W} & = - \varepsilon_1 \varepsilon_2 \log {\mathscr Z}  
\Bigl|_{(\varepsilon_1, \varepsilon_2) = (\hbar, 0)} 
\notag \\
& = \frac{a_{D}}{3} \, s 
-\frac{12 \, a_{D}^{2} + \hbar^{2} - 16 \, {m}^{2}}{12}\,\log s
+ \frac{68 \, a_{D}^{3} - (144 \, m^2-19 \, \hbar^2) \, a_D}{12} \, s^{-1} 
\notag \\
& \quad 
+ \frac{6000 \, a_{D}^{4} 
-(16512 \, {m}^{2} -  3672 \, \hbar^{2} ) \, a_D^2
+ 2816 \, {m}^{4} - 2272 \, \hbar^{2} {m}^{2}
+ 131 \, \hbar^{4} }{192} \, s^{-2}  + O(s^{-3}).
\label{eq:del-CCB-HII-1}
\end{align}
Then, our conjecture claims 

\begin{conj}
The classical conformal block \eqref{eq:del-CCB-HII-1} exists. 
Moreover, under the identification 
\begin{equation}
s = 8i \, t^{3/2},
\quad a_D = \frac{\nu}{2}, \quad {m} = \frac{\theta_\infty}{4},
\end{equation}
the accessory parameter \eqref{eq:final-AP-HII-1} and 
the classical conformal block \eqref{eq:del-CCB-HII-1}
are related as follows:
\begin{equation} \label{eq:conj-HII-1}
{\mathscr E} = \frac{d}{dt} {\mathscr W}. 
\end{equation} 
\end{conj}

We can observe that \eqref{eq:conj-HII-1}
holds up to $O(t^{-11/2})$.

\bigskip
\subsubsection{Type 2 expansion at $t \to \infty$} ~ 

\smallskip
\noindent
{\bf Series expansion of ${\mathscr E}$.} \, 
We perform the rescaling with 
$(d_{x}, d_{y}, d_{\mathscr E}) = (1/2, 3/2, 2)$ 
and take $\Lambda = t^{-3/2}$.
The rescaled spectral curve becomes
\begin{equation}
Q_{0}^{\rm res} = (X^2+1)^2+ {\mathscr G}_0
 + 2 \theta_\infty X  \Lambda 
\quad \xrightarrow{ \Lambda  \to 0 } \quad 
Q_0^{[0]} =  (X^2+1)^2+ {\mathscr G}^{[0]}_0.
\end{equation}
We take ${\mathscr G}_0^{[0]} = -1$
to have the limiting spectral curve 
\begin{equation}
Y^2 = Q_0^{[0]}(X) = X^2(X^2+2)
\end{equation}
is of genus $0$. 
Then, in this limit, it is readily seen that two of the three original zeros of 
$Q^{\rm res}_0$ coalesce at $X = 0$, 
and the limiting curve ${\mathcal C}^{\rm res}_{\rm deg}$ is a genus $0$ curve. 
Therefore, for sufficiently small $\Lambda$, 
we choose the cycle $\gamma$ to be a circle that encloses the two turning points 
which tend to the double zero $X=0$ of $Q_{0}^{[0]}$ inside, 
while keeping the other two turning point outside.
Then, the elliptic integrals defining the Voros period are described 
by the difference of the residues at $X = 0$ of the coefficients in the small $\Lambda$-expansion.
 
When imposing the condition
\begin{equation}
V_{-1} = 2 \pi i \nu, \quad 
V_{0} = - \pi i, \quad
V_{m} = 0 \quad (m \ge 1)
\end{equation} 
for the Voros period associated with $\gamma$
(which implies \eqref{eq:partial-IMD} with $-$-sign), 
the resulting computation of the accessory parameter is as follows:
\begin{subequations} \label{eq:AP-HII-2}
\begin{align} 
{\mathscr G}_0 & = 
 -1 + 2 \sqrt{2} \, \nu \Lambda - \frac{3 \, \nu^2 - 2 \,\theta_\infty^2}{4} \, \Lambda^2 
- \frac{\nu \, ( 17 \, \nu^2 -  48 \, \theta_\infty^2 )}{32\sqrt{2}} \, \Lambda^3 
\notag \\ & \quad
- \frac{ 375 \, \nu^{4} - 2496 \, \theta_{\infty}^{2} \nu^{2} + 64 \, \theta_{\infty}^{4} }{1024} \,\Lambda^4
+ O(\Lambda^5), \\
{\mathscr G}_1 & =  -\frac{3}{16} \,\Lambda^2 
- \frac{67 \, \nu}{128 \sqrt{2}} \, \Lambda^{3}
- \frac{1707 \,  \nu^{2} - 736 \, \theta_{\infty}^{2}}{2048} \, \Lambda^{4}
+ O(\Lambda^5), \\
{\mathscr G}_2 & = - \frac{1539}{16384}\, \Lambda^{4}
- \frac{305141 \, \nu}{262144 \sqrt{2}}\, \Lambda^{5} + O(\Lambda^6), \dots
\end{align}
\end{subequations}
Furthermore, by formally interchanging the order of $\hbar$-expansion and large $t$-expansion, 
we obtain the following formal series expansion for the accessory parameter:
\begin{align}
{\mathscr E} & = 
-t^{2} 
 + 2 \sqrt{2} \, \nu \, t^{1/2}
 - \frac{12 \, \nu^{2} - 8 \, \theta_{\infty}^{2} +3 \, \hbar^{2} }{16} \, t^{-1}
 - \frac{\nu \, \left(68 \, \nu^2 - 192 \, \theta_{\infty}^{2} + 67 \, \hbar^{2}\right) }{128 \sqrt{2}} \, t^{-5/2}
 \notag \\[+.3em]
 & \quad + \frac{6000 \, \nu^{4} 
- \left(39936 \, \theta_{\infty}^{2} - 13656 \, \hbar^{2} \right)\nu^{2} 
+ \left(1024 \, \theta_{\infty}^{4} - 5888 \, \hbar^{2} \theta_{\infty}^{2} +1539 \, \hbar^{4}  \right)
}{16384} \, t^{-4} + O(t^{-11/2}).
\label{eq:final-AP-HII-2}
\end{align}

\smallskip
\noindent
{\bf Comparison with ${\mathscr W}$.} \, 
Let us recall the quantum expansion of ${\mathscr Z}$, 
called ``square exp singularity of ${\rm QPII}$ / $H_1$ with heavy hyper", 
given in \cite[(4.132)--(4.138) in \S 4.9]{BST25}\footnote{
In \cite{BST25}, the notation $\alpha_D = {i a_D}/{\sqrt{2}}$ was used; here, we adopt a formulation expressed solely in terms of $a_D$. 
We also note that there is a typo in \cite[(4.135)]{BST25}, where the coefficient of $(\varepsilon_1+\varepsilon_2)^4 s^{-2}$ in \eqref{eq:CB-HII-2} is written as ${175}/{32}$.
}:  
\begin{align}
& - \varepsilon_1 \varepsilon_2 \log {\mathscr Z} \notag \\ 
& \quad = 
\frac{s^2}{192} - \frac{i\sqrt{2} \, a_D}{6} \, s 
+ \left(- \frac{a_D^2}{2} + \frac{16 \, m^2}{3}  +\frac{\varepsilon_1 \varepsilon_2}{12} \right) \, \log s
\notag \\
& \qquad + \left( \frac{17 i}{6\sqrt{2}} \, a_D^3 - \frac{i}{\sqrt{2}}\left(128 \, {m}^2 + \frac{31 \, \epsilon_1\epsilon_2}{12} +  \frac{5 \, (\epsilon_1+\epsilon_2)^2}{24}\right)a_D \right) \, s^{-1}
\notag \\ 
& \qquad 
+ \biggl( 
-\frac{125}{16} \, a_D^4
+ 
\left(
832 \, m^2
+ \frac{293 \, \varepsilon_1 \varepsilon_2}{16} 
+ \frac{55 \, (\varepsilon_1 + \varepsilon_2)^2}{32} 
\right) a_D^2
\notag \\ & \quad\quad 
- \frac{1024 \, m^4}{3} 
- \frac{160 \, \left(
\varepsilon_1 \varepsilon_2 - 2 \, (\varepsilon_1 + \varepsilon_2)^2
\right) \, m^2}{3}
- \frac{3 \, \varepsilon_1 \varepsilon_2)^2}{4} 
+ \frac{165 \, \varepsilon_1 \varepsilon_2 (\varepsilon_1 + \varepsilon_2)^2}{256}
+ \frac{175 \, (\varepsilon_1 + \varepsilon_2)^4}{256} 
\biggr) \, s^{-2}
\notag \\ 
&\qquad  
+ O(s^{-3}). 
\label{eq:CB-HII-2}
\end{align}
The result of \cite{BST25} suggests that the associated classical conformal block also exists:
\begin{align}
{\mathscr W} 
& = - \varepsilon_1 \varepsilon_2 \log {\mathscr Z}  
\Bigl|_{(\varepsilon_1, \varepsilon_2) = (\hbar, 0)} 
\notag \\
& 
= \frac{s^2}{192}
- \frac{i \sqrt{2} \, a_{D}}{6} \, s 
- \frac{3 \, a_{D}^{2} - 32 \, {m}^{2}}{6}\,\log s
+ \frac{i \sqrt{2} \, \left( 68 \, a_{D}^{3} - (3071 \, {m}^{2} + 5  \, \hbar^2) \, a_{D}  \right)}{48} \, s^{-1} 
\notag \\
& \quad 
-\frac{
6000 \, a_D^{4}
- (638976 \, {m}^{2} + 1320 \, \hbar^{2})  a_D^{2}
+ 262144 \, {m}^{4}
- 81920 \, \hbar^{2} {m}^{2}
- 525 \, \hbar^{4}}
{768} \,  s^{-2}
+ O(s^{-3}).
\label{eq:del-CCB-HII-2}
\end{align}
Then, our conjecture claims 

\begin{conj}
The classical conformal block \eqref{eq:del-CCB-HII-2} exists. 
Moreover, under the identification 
\begin{equation}
s = 8 i \, t^{3/2}, 
\quad a_D = \nu, \quad 
{m}^2 = \frac{\theta_\infty^2}{16} - \frac{3 \, \hbar^2}{128},
\end{equation}
the accessory parameter \eqref{eq:final-AP-HII-2}
and the classical conformal block \eqref{eq:del-CCB-HII-2}
are related as follows:
\begin{equation} \label{eq:conj-HII-2}
{\mathscr E} = \frac{d}{dt} {\mathscr W}. 
\end{equation} 
\end{conj}

We can observe that \eqref{eq:conj-HII-2}
holds up to $O(t^{-11/2})$.

\bigskip
\subsection{Reduced biconfluent Heun equation $H_{{\rm II}'}$}

The equation $H_{{\rm II}'}$ has the following potential
\begin{equation}
Q_{{\rm II}'} = \frac{\theta_0^2 - \frac{\hbar^2}{4}}{x^2} + \frac{{\mathscr E}}{x} + t + x.
\end{equation}

\subsubsection{Type 1 expansion at $t \to \infty$} ~ 

\smallskip
\noindent
{\bf Series expansion of ${\mathscr E}$.} \, 
We perform the rescaling with 
$(d_{x}, d_{y}, d_{\mathscr E}) = (-1/2, 0, 1/2)$ 
and take $\Lambda = t^{-3/2}$.
The rescaled spectral curve becomes
\begin{equation}
Q_{0}^{\rm res} = \left( \frac{\theta_0^2}{X^2} +\
\frac{{\mathscr G}_0}{X} + 1 \right) + X \Lambda 
\quad \xrightarrow{ \Lambda  \to 0 } \quad 
Q_0^{[0]} = \frac{\theta_0^2}{X^2} + 
\frac{{\mathscr G}_0^{[0]}}{X} + 1. 
\end{equation}
In this limit, it is readily seen that one of the three original zeros of 
$Q^{\rm res}_0$ coalesces at the infinity, 
and the limiting curve ${\mathcal C}^{\rm res}_{\rm deg}$ is a genus $0$ curve 
(which is called the Kummer curve in \cite{IKT18-2}). 
Therefore, for sufficiently small $\Lambda$, 
we choose the cycle $\gamma$ to be a circle that encloses the two turning points 
which tend to the zeros of $Q_{0}^{[0]}$ together with the origin inside, 
while keeping the turning point that merges with infinity outside.
Then, the elliptic integrals defining the Voros period are described 
by the residues at $X = \infty$ of the coefficients in the small $\Lambda$-expansion.

When imposing the condition 
\begin{equation}
V_{-1} = 2 \pi i \nu, \quad 
V_{m} = 0 \quad (m \ge 0)
\end{equation} 
for the Voros period associated with $\gamma$
(which implies \eqref{eq:partial-IMD} with $+$-sign), 
the resulting computation of the accessory parameter is as follows:
\begin{subequations} \label{eq:AP-HII'-1}
\begin{align} 
{\mathscr G}_0 & = - 2 \nu - \frac{3 \, \nu^2 - \theta_0^2}{2} \, \Lambda  
+ \frac{17 \, \nu^3 -  9 \, \theta_0^2}{8} \, \Lambda ^2 
-\frac{375 \, \nu^{4} - 258 \, \theta_{0}^{2} \nu^{2} +11 \, \theta_{0}^{4}}{64} \,\Lambda ^3
+ O(\Lambda ^4), \\
{\mathscr G}_1 & =  -\frac{1}{8} \,\Lambda 
+ \frac{19 \, \nu}{32} \, \Lambda^{2}
- \frac{459 \, \nu^{2} - 71 \, \theta_{0}^{2}}{128} \, \Lambda^{3}
+ O(\Lambda^4), \\
{\mathscr G}_2 & = -\frac{131}{1024}\, \Lambda^{3}
+ \frac{7287 \, \nu}{8192}\, \Lambda^{4} + O(\Lambda^5), \dots
\end{align}
\end{subequations}
Furthermore, by formally interchanging the order of 
$\hbar$-expansion and large $t$-expansion, 
we obtain the following formal series expansion for the accessory parameter:
\begin{align}
{\mathscr E} & = 
-2 \nu \, t^{1/2} - \frac{4 \, (3 \, \nu^2 - \theta_0^2) + \hbar^2}{8} \, t^{-1} 
+ \frac{\nu \, (68 \, \nu^2 - 36 \, \theta_0^2 + 19 \, \hbar^2)}{32} \, t^{-5/2} 
\notag \\
& \quad - \frac{6000 \, \nu^4 - (4128 \, \theta_0^2 - 3672 \, \hbar^2) \, \nu^2 
+  176 \, \theta_0^4 - 568 \, \hbar^2  \theta_0^2  + 131 \, \hbar^4}{1024} \, t^{-4}
+ O(t^{-11/2}).
\label{eq:final-AP-HII'-1}
\end{align}

\smallskip
\noindent
{\bf Comparison with ${\mathscr W}$.} \, 
The series expansion of relevant classical conformal block 
is computed in \eqref{eq:del-CCB-HII-1}. 
Our conjecture claims

\begin{conj}
The classical conformal block \eqref{eq:del-CCB-HII-1} exists. 
Moreover, under the identification 
\begin{equation}
s = 4 \, t^{3/2}, \quad a_D = - \nu, \quad {m} = \frac{\theta_0}{2},
\end{equation}
the accessory parameter \eqref{eq:final-AP-HII'-1} and
the classical conformal block \eqref{eq:del-CCB-HII-1}
are related as follows:
\begin{equation} \label{eq:conj-HII'-1}
{\mathscr E} = \frac{d}{dt} {\mathscr W}. 
\end{equation} 
\end{conj}

We can observe that \eqref{eq:conj-HII'-1} holds up to $O(t^{-11/2})$.


\bigskip
\subsubsection{Type 2 expansion at $t \to \infty$} ~ 

\smallskip
\noindent
{\bf Series expansion of ${\mathscr E}$.} \, 
We perform the rescaling with 
$(d_{x}, d_{y}, d_{\mathscr E}) = (1, 3/2, 2)$ 
and take $\Lambda = t^{-3/2}$.
The rescaled spectral curve becomes
\begin{equation}
Q_{0}^{\rm res} =  1 + X + 
\frac{{\mathscr G}_0}{X} 
+  \frac{\theta_0^2}{X^2} \, \Lambda
\quad \xrightarrow{ \Lambda \to 0 } \quad 
Q_0^{[0]} = 1 + X + 
\frac{{\mathscr G}_0^{[0]}}{X}. 
\end{equation}

We take ${\mathscr G}_0^{[0]} = {1}/{4}$
to have the limiting spectral curve 
\begin{equation}
    Y^2 = Q_0^{[0]}(X) = \frac{(X+\frac{1}{2})^2}{X}
\end{equation} 
is of genus $0$. 
Then, in this limit, it is readily seen that two of the three original zeros of 
$Q^{\rm res}_0$ coalesce at $X = -1/2$, 
and the limiting curve ${\mathcal C}^{\rm res}_{\rm deg}$ is a genus $0$ curve. 
Therefore, for sufficiently small $\Lambda$, 
we choose the cycle $\gamma$ to be a circle that encloses the two turning points 
which tend to the double zero of $Q_{0}^{[0]}$ inside, 
while keeping the other turning point outside.
Then, the elliptic integrals defining the Voros period are described 
by the residues at $X = -1/2$ of the coefficients in the small $\Lambda$-expansion.

When imposing the condition
\begin{equation}
V_{-1} = 2 \pi i \nu, \quad 
V_{0} = - \pi i, \quad
V_{m} = 0 \quad (m \ge 1)
\end{equation} 
for the Voros period associated with $\gamma$
(which implies \eqref{eq:partial-IMD} with $-$-sign), 
the resulting computation of the accessory parameter is as follows:
\begin{subequations} \label{eq:AP-HII'-2}
\begin{align} 
{\mathscr G}_0 & = \frac{1}{4} - {i \sqrt{2}  \nu} \, \Lambda  
- \frac{3 \, \nu^2 - 8 \, \theta_0^2}{4} \, \Lambda ^2 
- \frac{i \nu \, ( 17 \, \nu^2 - 192 \, \theta_0^2)}{16 \sqrt{2}} \,\Lambda ^3
+ O(\Lambda ^4), \\
{\mathscr G}_1 & =  -\frac{3}{16} \,\Lambda^2 
- \frac{67 \, i \nu}{64 \sqrt{2}} \, \Lambda^{3}
+ \frac{1707 \, \nu^{2} - 2944 \, \theta_{0}^{2}}{512} \, \Lambda ^{4}
+ O(\Lambda^5), \\
{\mathscr G}_2 & = \frac{1539}{4096}\, \Lambda^{4}
+ \frac{305141 \, i \nu}{32768\sqrt{2}}\, \Lambda^{5} + O(\Lambda^6), \dots
\end{align}
\end{subequations}
Furthermore, by formally interchanging the order of $\hbar$-expansion and large $t$-expansion, 
we obtain the following formal series expansion for the accessory parameter:
\begin{align}
{\mathscr E} & = 
\frac{t^2}{4} - i\sqrt{2} \, t^{1/2}
- \frac{12 \, \nu^2 - 32 \, \theta_0^2 + 3 \, \hbar^2}{16} \, t^{-1} 
- \frac{i \nu \,(68 \, \nu^2 - 768 \, \theta_0^2 + 67 \, \hbar^2)}{64\sqrt{2}} \,t^{-5/2}
\notag \\
& \quad + \frac{6000 \, \nu^4 - 24\,(6656 \,\theta_0^2 - 569 \,\hbar^2) \nu^2 
+  16384 \, \theta_0^4 - 23552 \, \hbar^2 \theta_0^2   + 1539 \, \hbar^4}{4096} t^{-4}
+ O(t^{-11/2}).
\label{eq:final-AP-HII'-2}
\end{align}

\smallskip
\noindent
{\bf Comparison with ${\mathscr W}$.} \, 
The series expansion of relevant classical conformal block 
is computed in \eqref{eq:del-CCB-HII-2}. 
Our conjecture claims

\begin{conj}
The classical conformal block \eqref{eq:del-CCB-HII-2} exists. 
Moreover, under the identification 
\begin{equation}
s = 4 \,t^{3/2}, \quad a_D = \nu, \quad {m}^2 = \frac{\theta_0^2}{4} - \frac{3 \, \hbar^2}{128},
\end{equation}
the accessory parameter \eqref{eq:final-AP-HII'-2} and
the classical conformal block \eqref{eq:del-CCB-HII-2}
are related as follows:
\begin{equation} \label{eq:conj-HII'-2}
{\mathscr E} = \frac{d}{dt} {\mathscr W}. 
\end{equation} 
\end{conj}
 
We can observe that \eqref{eq:conj-HII'-2} holds up to $O(t^{-11/2})$.

\bigskip 
\subsection{Reduced triconfluent Heun equation $H_{\rm I}$}

The equation $H_{{\rm I}}$ has the following potential
\begin{equation}
Q_{{\rm I}} = 4x^3 + 2t x + {\mathscr E}.
\end{equation}
In the following, we derive the expansion of ${\mathscr E}$ as $t \to \infty$, 
and test the Zamolodchikov-type conjecture.

\smallskip
\noindent
{\bf Series expansion of ${\mathscr E}$.} \, 
We perform the rescaling with 
$(d_{x}, d_{y}, d_{\mathscr E}) = (1/2, 5/4, 3/2)$ 
and take $\Lambda = t^{-5/4}$.
The rescaled spectral curve becomes
\begin{equation}
Q_{0}^{\rm res} = 4X^3+2X+{\mathscr G}_0
\quad \xrightarrow{ \Lambda \to 0 } \quad 
Q_0^{[0]} =  4X^3+2X+{\mathscr G}_0^{[0]}. 
\end{equation}
We set $q = \sigma i \sqrt{{1}/{6}}$ 
with $\sigma = +1$ or $-1$, and 
take 
${\mathscr G}_0^{[0]} = - {4q}/{3}$
to have the limiting spectral curve 
\begin{equation}
Y^2 = Q_0^{[0]}(X) = 4(X-q)^2(X+2q)
\end{equation}
is of genus $0$. 
Then, in this limit, it is readily seen that two of the three original zeros of 
$Q^{\rm res}_0$ coalesce at $X = q$. 
Therefore, for sufficiently small $\Lambda$, 
we choose the cycle $\gamma$ to be a circle that encloses the two turning points 
which tend to the double zero of $Q_{0}^{[0]}$ inside, 
while keeping the other turning point outside.
Then, the elliptic integrals defining the Voros period are described 
by the residues at $X = q$ of the coefficients in the small $\Lambda$-expansion.

When imposing the condition 
\begin{equation}
V_{-1} = 2 \pi i \nu, \quad 
V_{0} = - \pi i, \quad
V_{m} = 0 \quad (m \ge 1)
\end{equation} 
for the Voros period associated with $\gamma$
(which implies \eqref{eq:partial-IMD} with $-$-sign), 
the resulting computation of the accessory parameter is as follows\footnote{
The choice of the branch of $\sqrt{q}$ depends on the choice of 
the branch of $\sqrt{Q_0}$ along $\gamma$. 
Since this choice causes no issue in the later comparison 
with the classical conformal block, we do not discuss the choice of the branch here.}: 
\begin{subequations} \label{eq:AP-HI}
\begin{align} 
{\mathscr G}_0 & = -\frac{4q}{3}
+ 4\sqrt{3q}\, \nu \Lambda
- \frac{5 \nu^2}{4} \, \Lambda^2
- \frac{235\, \nu^3}{192 \sqrt{3q}} \, \Lambda^3
- \frac{38585\,\nu^4}{1492992 \,q} \, \Lambda^4 
+ O(\Lambda^5), \\
{\mathscr G}_1 & = -\frac{7}{48} \, \Lambda^2
- \frac{385\, \nu}{768\sqrt{3q}} \, \Lambda^3
- \frac{69685\, \nu^2}{110592\,q} \, \Lambda^4
+ O(\Lambda^5), \\
{\mathscr G}_2 & = -\frac{101479\,}{2654208\,q} \, \Lambda^4
- \frac{43147783\,\nu}{84934656 \,q\sqrt{3q}} \, \Lambda^5  + O(\Lambda^6), \dots
\end{align}
\end{subequations}
Furthermore, by formally interchanging the order of $\hbar$-expansion and large $t$-expansion, 
we obtain the following formal series expansion for the accessory parameter:
\begin{align}
{\mathscr E} & = 
-\frac{4}{3} q\, t^{3/2}
+ 4\sqrt{3q} \, \nu\, t^{1/4}
- \frac{60 \, \nu^2 + 7 \,\hbar^2}{48} \, t^{-1}
- \frac{5 \,(188 \, \nu^3+77 \, \hbar^2\nu)}{768 \sqrt{3q}} \, t^{-9/4}
\notag \\
& \quad - \frac{1852080 \, \nu^4 + 1672440 \, \hbar^2\nu^2 + 101479 \, \hbar^4}{2654208\,q} \, t^{-7/2}
+ O(t^{-19/4}).
\label{eq:final-AP-HI}
\end{align}

\smallskip
\noindent
{\bf Comparison with ${\mathscr W}$.} \, 
Let us recall the quantum expansion of ${\mathscr Z}$, 
called ``${\rm QPI}$ / $H_0$", 
given in \cite[(4.143)--(4.146) in \S 4.10]{BST25}:  
\begin{align}
& - \varepsilon_1 \varepsilon_2 \log {\mathscr Z} 
\notag \\
& \quad = - \frac{s^2}{103680}
+ \frac{a_D}{60} \, s
- \left(
\frac{a_D^2}{2}
 - \frac{2\,\varepsilon_1 \varepsilon_2-7 \,(\varepsilon_1+\varepsilon_2)^2}{120}
  \right)\log s 
\notag \\
& \qquad   +
   \left( 47 \, a_D^3 - \frac{34\,\varepsilon_1\varepsilon_2 -77\,(\epsilon_1+\epsilon_2)^2}{4} a_D \right) \, s^{-1}
  \notag \\[+.2em] 
&  \qquad  + 
  \Biggl( 
  \frac{7717}{2} \, a_D^4
  - \frac{7354\,\varepsilon_1\varepsilon_2
    - 13937\, (\varepsilon_1+\varepsilon_2)^2}{4}a_D^2
   \notag \\ & \quad\quad 
  + 7 \Bigl(
      \frac{24\,(\varepsilon_1\varepsilon_2)^2}{5}
      - \frac{4597\,\varepsilon_1\varepsilon_2 (\varepsilon_1+\varepsilon_2)^2}{120}
      + \frac{14497 \, (\varepsilon_1+\varepsilon_2)^4}{480}
    \Bigr)
\Biggr) \, s^{-2}  + O(s^{-3}).     
\label{eq:CB-HI}
\end{align}
The series is expected to agree with the irregular conformal block of rank $5/2$ introduced in \cite{PP23, IILZ25}.
The result of \cite{BST25} suggests that
the associated classical conformal block also exists\footnote{
To compare with the accessory parameter \eqref{eq:final-AP-HI}, 
we set $\varepsilon_1 = \sqrt{2} \hbar$ instead of  $\varepsilon_1 = \hbar$.
}:
\begin{align}
{\mathscr W} & = - \varepsilon_1 \varepsilon_2 \log {\mathscr Z}  
\Bigl|_{(\varepsilon_1, \varepsilon_2) = (\sqrt{2} \, \hbar, 0)} 
= -\frac{s^2}{103680}
+ \frac{a_D}{60} \, s
- \left( \frac{30 \, a_D^2 + 7 \, \hbar^2}{60} \right) \log s
\notag \\
& \quad
+ \frac{94 \, a_D^3 + 77 \, \hbar^2  a_D}{2} \, s^{-1}
+ \frac{463020\, a_D^{4} + 836220 \, \hbar^{2} a_D^{2} + 101479\, \hbar^{4}}{120} \, s^{-2} + O(s^{-3}).
\label{eq:CCB-HI}
\end{align}
Then, our conjecture claims 

\begin{conj}
The classical conformal block \eqref{eq:CCB-HI} exists. 
Moreover, under the identification 
\begin{equation}
s = 96 \sqrt{6 q} \, t^{5/4}, \quad a_D = \sqrt{2} \, \nu, 
\end{equation}
the accessory parameter \eqref{eq:final-AP-HI} and
the classical conformal block \eqref{eq:CCB-HI}
are related as follows:
\begin{equation} \label{eq:conj-HI}
{\mathscr E} =  \frac{d}{dt}  {\mathscr W}. 
\end{equation} 
\end{conj}

We can observe that \eqref{eq:conj-HI} holds up to $O(t^{-19/4})$.



\appendix

\bigskip
\section{Introduction of $\hbar$ via scaling}
\label{section:introduce-hbar}

{\footnotesize 

\begin{table}[t]
  \begin{tabular}{l|l|l} \hline
     $J$ & Name of equation & $V_{J}(x)$ \\ \hline \hline
\parbox[c][5.5em][c]{0em}{}
VI
&
Heun
& ~~~\quad
\begin{minipage}{.42\textwidth}
\begin{center}
$\displaystyle \frac{\theta_0^2 - \frac{1}{4}}{x^2} + \frac{\theta_1^2 - \frac{1}{4}}{(x-1)^2} 
+ \frac{\theta_t^2 - \frac{1}{4}}{(x-t)^2}$    \\[+.5em]
$\displaystyle + \frac{\theta_\infty^2 - \theta_0^2 - \theta_1^2 - \theta_t^2 + \frac{1}{2}}{x(x-1)} 
- \frac{{\mathscr E}}{x(x-1)(x-t)}$
\end{center}
\end{minipage}
\\\hline
\parbox[c][3.3em][c]{0em}{}
V
&
confluent Heun
& ~~~\quad
\begin{minipage}{.42\textwidth}
\begin{center}
$\displaystyle \frac{\theta_0^2 - \frac{1}{4}}{x^2}+ \frac{\theta_t^2 - \frac{1}{4}}{(x-t)^2} 
+ \frac{1}{4} + \frac{\theta_\infty}{x} - \frac{{\mathscr E}}{x(x-t)}$ 
\end{center}
\end{minipage}
\\\hline
\parbox[c][3.3em][c]{0em}{}
IV
&
biconfluent Heun
& ~~~\quad
\begin{minipage}{.42\textwidth}
\begin{center}
$\displaystyle \frac{\theta_0^2 - \frac{1}{4}}{x^2} + \frac{{\mathscr E}}{x} + 2 \theta_\infty + (x+t)^2$
\end{center}
\end{minipage}
\\\hline
\parbox[c][3.3em][c]{0em}{}
${\rm III}_1$
&
doubly confluent Heun
& ~~~\quad
\begin{minipage}{.42\textwidth}
\begin{center}
$\displaystyle \frac{t^2}{4x^4}  + \frac{t \theta_0}{x^3} - \frac{{\mathscr E}}{x^2} + \frac{\theta_\infty}{x} + \frac{1}{4}$ 
\end{center}
\end{minipage}
\\\hline

\parbox[c][3.3em][c]{0em}{}
${\rm III}_2$
&
reduced doubly confluent Heun
& ~~~\quad
\begin{minipage}{.42\textwidth}
\begin{center}
$\displaystyle \frac{t}{x^3} - \frac{{\mathscr E}}{x^2} + \frac{\theta_\infty}{x} + \frac{1}{4}$ 
\end{center}
\end{minipage}
\\\hline
\parbox[c][3.3em][c]{0em}{}
${\rm III}_3$
&
doubly reduced doubly confluent Heun
& ~~~\quad
\begin{minipage}{.42\textwidth}
\begin{center}
$\displaystyle \frac{t}{x^3} - \frac{{\mathscr E}}{x^2} + \frac{1}{x}$ 
\end{center}
\end{minipage}
\\\hline

\parbox[c][3.3em][c]{0em}{}
II
&
triconfluent Heun
& ~~~\quad
\begin{minipage}{.42\textwidth}
\begin{center}
$\displaystyle (x^2+t)^2 + 2 \theta_\infty x + {\mathscr E}$ 
\end{center}
\end{minipage}
\\\hline

\parbox[c][3.3em][c]{0em}{}
${\rm II}'$
&
reduced biconfluent Heun
& ~~~\quad
\begin{minipage}{.42\textwidth}
\begin{center}
$\dfrac{\theta_0^2 - \frac{1}{4}}{x^2}
+ \dfrac{{\mathscr E}}{x} + t + x$ 
\end{center}
\end{minipage}
\\\hline

\parbox[c][3.3em][c]{0em}{}
I
&
reduced triconfluent Heun
& ~~~\quad
\begin{minipage}{.42\textwidth}
\begin{center}
$\displaystyle 4x^3 + 2 t x + {\mathscr E}$ 
\end{center}
\end{minipage}
\\\hline

  \end{tabular}
\end{table}

}

The list of original (i.e., $\hbar = 1$ case of) Heun equations in the SL-form is given as follows
(cf.\, \cite[Table 1]{LN21}):
\[
\frac{d^2 \psi}{dx^2} = V_J(x) \psi
\]
Below, we will review how to introduce the small parameter $\hbar$ to these equations. 

\smallskip
\begin{itemize}
\item  \underline{Heun equation $H_{\rm VI}$.}  
~~ Applying the rescaling 
\[
{\mathscr E} \mapsto \hbar^{-2} {\mathscr E}, \quad 
\theta_0 \mapsto \hbar^{-1} \theta_0, \quad
\theta_1 \mapsto \hbar^{-1} \theta_1, \quad
\theta_t \mapsto \hbar^{-1} \theta_t, \quad
\theta_\infty \mapsto \hbar^{-1} \theta_\infty,
\]
we have $V_{\rm VI}(x) \, dx^2 \mapsto \hbar^{-2} Q_{\rm VI}(x) \, dx^2$ with 
\[
Q_{\rm VI} = \frac{\theta_0^2 - \frac{\hbar^2}{4}}{x^2}+  \frac{\theta_1^2 - \frac{\hbar^2}{4}}{(x-1)^2} +
\frac{\theta_t^2 - \frac{\hbar^2}{4}}{(x-t)^2} 
+ \frac{\theta_\infty^2 - \theta_0^2 - \theta_1^2 - \theta_t^2+\frac{\hbar^2}{2}}{x(x-1)} 
- \frac{{\mathscr E}}{x(x-1)(x-t)}.
\]

\smallskip
\item  \underline{Confluent Heun equation $H_{\rm V}$.} 
~~ Applying the rescaling 
\[
x \mapsto \hbar^{-1}x, \quad 
t \mapsto \hbar^{-1}t, \quad
{\mathscr E} \mapsto \hbar^{-2} {\mathscr E}, \quad
\theta_0 \mapsto \hbar^{-1} \theta_0, \quad
\theta_t \mapsto \hbar^{-1} \theta_t, \quad
\theta_\infty \mapsto \hbar^{-1} \theta_\infty , 
\]
we have $V_{\rm V}(x) \, dx^2 \mapsto \hbar^{-2} Q_{\rm V}(x) \, dx^2$ with 
\[
Q_{\rm V} = \frac{\theta_0^2 - \frac{\hbar^2}{4}}{x^2}+ \frac{\theta_t^2 - \frac{\hbar^2}{4}}{(x-t)^2} 
+ \frac{1}{4} + \frac{\theta_\infty}{x} - \frac{{\mathscr E}}{x(x-t)}.
\]

\smallskip
\item  \underline{Biconfluent Heun equation $H_{\rm IV}$.} 
~~ Applying the rescaling 
\[
x \mapsto \hbar^{-1/2}x, \quad 
t \mapsto \hbar^{-1/2}t, \quad
{\mathscr E} \mapsto \hbar^{-3/2} {\mathscr E}, \quad
\theta_0 \mapsto \hbar^{-1} \theta_0, \quad
\theta_\infty \mapsto \hbar^{-1} \theta_\infty,
\]
then we have $V_{\rm IV}(x) \, dx^2 \mapsto \hbar^{-2}  Q_{\rm IV}(x) \, dx^2$ with 
\[
Q_{\rm IV} = 
\frac{\theta_0^2 - \frac{\hbar^2}{4}}{x^2} + \frac{{\mathscr E}}{x} + 2 \theta_\infty + (x+t)^2.
\]

\smallskip
\item  \underline{Doubly confluent Heun equation $H_{{\rm III}_1}$.} 
~~  Applying the rescaling 
\[
x \mapsto \hbar^{-1}x, \quad 
t \mapsto \hbar^{-2}t, \quad
{\mathscr E} \mapsto \hbar^{-2} {\mathscr E}, \quad
\theta_0 \mapsto \hbar^{-1} \theta_0, \quad
\theta_\infty \mapsto \hbar^{-1} \theta_\infty,
\]
then we have $V_{{\rm III}_1}(x) \, dx^2 \mapsto \hbar^{-2}  Q_{{\rm III}_1}(x) \, dx^2$ with 
\[
Q_{{\rm III}_1} = \displaystyle \frac{t^2}{4x^4}  + \frac{t \theta_0}{x^3} - \frac{{\mathscr E}}{x^2} + \frac{\theta_\infty}{x} + \frac{1}{4}. 
\]



\smallskip
\item  \underline{Reduced doubly confluent Heun equation $H_{{\rm III}_2}$.} 
~~  Applying the rescaling 
\[
x \mapsto \hbar^{-1}x, \quad 
t \mapsto \hbar^{-3}t, \quad
{\mathscr E} \mapsto \hbar^{-2} {\mathscr E}, \quad
\theta_\infty \mapsto \hbar^{-1} \theta_\infty,
\]
then we have $V_{{\rm III}_2}(x) \, dx^2 \mapsto \hbar^{-2}  Q_{{\rm III}_2}(x) \, dx^2$ with 
\[
Q_{{\rm III}_2} = \displaystyle \frac{t}{x^3} - \frac{{\mathscr E}}{x^2} 
+ \frac{\theta_\infty}{x} + \frac{1}{4}. 
\]

\smallskip
\item  \underline{Doubly reduced doubly confluent Heun equation $H_{{\rm III}_3}$.} 
~~  Applying the rescaling 
\[
x \mapsto \hbar^{-2}x, \quad 
t \mapsto \hbar^{-4}t, \quad
{\mathscr E} \mapsto \hbar^{-2} {\mathscr E}, 
\]
then we have $V_{{\rm III}_3}(x) \, dx^2 \mapsto \hbar^{-2}  Q_{{\rm III}_3}(x) \, dx^2$ with 
\[
Q_{{\rm III}_3} = \displaystyle \frac{t}{x^3} - \frac{{\mathscr E}}{x^2} + \frac{1}{x}. 
\]

\smallskip
\item  \underline{Triconfluent Heun equation $H_{\rm II}$.}
~~  Applying the rescaling 
\[
x \mapsto \hbar^{-1/3}x, \quad 
t \mapsto \hbar^{-2/3}t, \quad
{\mathscr E} \mapsto \hbar^{-4/3} {\mathscr E}, \quad
\theta_\infty \mapsto \hbar^{-1} \theta_\infty,
\]
then we have $V_{\rm II}(x) \, dx^2 \mapsto \hbar^{-2} Q_{\rm II}(x) \, dx^2$ with 
\[
Q_{\rm II} = (x^2+t)^2 + 2 \theta_\infty x + {\mathscr E}.
\]

\smallskip
\item  \underline{Reduced biconfluent Heun equation $H_{{\rm II}'}$.}
~~  Applying the rescaling 
\[
x \mapsto \hbar^{-2/3}x, \quad 
t \mapsto \hbar^{-2/3}t, \quad
{\mathscr E} \mapsto \hbar^{-4/3} {\mathscr E}, \quad
\theta_0 \mapsto \hbar^{-1} \theta_0,
\]
then we have $V_{\rm II'}(x) \, dx^2 \mapsto \hbar^{-2} Q_{\rm II'}(x) \, dx^2$ with 
\[
Q_{\rm II'} = 
\frac{\theta_0^2 - \frac{\hbar^2}{4}}{x^2} + 
\frac{{\mathscr E}}{x} + t + x.
\]

\smallskip
\item  \underline{Reduced triconfluent Heun equation $H_{\rm I}$.}
~~  Applying the rescaling 
\[
x \mapsto \hbar^{-2/5}x, \quad 
t \mapsto \hbar^{-4/5}t, \quad
{\mathscr E} \mapsto \hbar^{-6/5} {\mathscr E},
\]
then we have $V_{\rm I}(x) \, dx^2 \mapsto \hbar^{-2} Q_{\rm I}(x) \, dx^2$ with 
\[
Q_{\rm I} = 4x^3 + 2 t x + {\mathscr E}.
\]

\end{itemize}

\bigskip

\end{document}